\documentclass[acmsmall,screen,table]{acmart}
\pdfoutput=1

\setcopyright{rightsretained}
\acmPrice{}
\acmDOI{10.1145/3622861}
\acmYear{2023}
\copyrightyear{2023}
\acmSubmissionID{oopslab23main-p481-p}
\acmJournal{PACMPL}
\acmVolume{7}
\acmNumber{OOPSLA2}
\acmArticle{285}
\acmMonth{10}

\usepackage{natbib}

\usepackage{booktabs}
\usepackage{subcaption}

\usepackage[utf8]{inputenc}
\usepackage{amsmath}
\usepackage{stmaryrd}
\usepackage{soul}
\usepackage{empheq}
\usepackage[breakable]{tcolorbox}
\usepackage{multirow}
\usepackage{wrapfig}
\usepackage{fancyvrb}
\usepackage{listings}
\usepackage{thmtools,thm-restate}
\usepackage{hhline}
\usepackage{setspace}

\newcommand\Tstrut{\rule{0pt}{2.6ex}}         
\newcommand\Bstrut{\rule[-0.9ex]{0pt}{0pt}}   

\usepackage[linesnumbered,ruled,vlined]{algorithm2e}

\SetCommentSty{mycommfont}
\newcommand{\mypar}[1]{\vspace{1mm}\noindent\textit{#1.}}

\xspaceaddexceptions{[]\{\}}
\newcommand{\rone}{(\emph{i})\xspace}
\newcommand{\rtwo}{(\emph{ii})\xspace}
\newcommand{\rthree}{(\emph{iii})\xspace}
\newcommand{\rfour}{(\emph{iv})\xspace}

\newcommand{\query}{\Phi}
\newcommand{\signatures}{\query}
\newcommand{\varsig}{V_\query}
\newcommand{\signature}{\sigma}

\newcommand{\bestset}{\mathcal{P}}

\newcommand{\lang}{\mathcal{L}}
\newcommand{\fls}{\bot}
\newcommand{\tru}{\top}

\newcommand{\lproperty}{$\lang$-property\xspace}
\newcommand{\lconjunction}{$\lang$-conjunction\xspace}
\newcommand{\lconjunctions}{$\lang$-conjunctions\xspace}
\newcommand{\lproperties}{$\lang$-properties\xspace}
\newcommand{\lconjuncts}{\lproperties}

\newcommand{\phiprog}{\varphi_{\signatures}}
\newcommand{\phiand}{\varphi_{\wedge}}

\newcommand{\phiinit}{\varphi_{init}}
\newcommand{\philast}{\varphi_{last}}
\newcommand{\phil}{\varphi}

\newcommand{\ex}{e}
\newcommand{\posex}{\ex^+}
\newcommand{\negex}{\ex^-}
\newcommand{\eplus}{E^+}
\newcommand{\eminus}{E^-}
\newcommand{\eminusmust}{E^-_{must}}
\newcommand{\eminusmay}{E^-_{may}}
\newcommand{\eminusdelta}{E^-_\delta}
\newcommand{\cs}{\textnormal{{\textsc{CheckSoundness}}}\xspace}
\newcommand{\cp}{\textnormal{{\textsc{CheckPrecision}}}\xspace}
\newcommand{\synth}{\textnormal{{\textsc{Synthesize}}}\xspace}
\newcommand{\maxsynth}{\textnormal{{\textsc{MaxSynthesize}}}\xspace}

\newcommand{\synthproperty}{\textnormal{{\textsc{SynthesizeStrongestConjunct}}}\xspace}
\newcommand{\synthproperties}{\textnormal{{\textsc{SynthesizeStrongestConjunction}}}\xspace}

\newcommand{\interp}[1]{\llbracket #1 \rrbracket}

\newcommand{\name}{\textsc{spyro}\xspace}
\newcommand{\sketch}{\textsc{sketch}\xspace}
\newcommand{\smt}{\textsc{smt}\xspace}
\newcommand{\jlibsketch}{\textsc{JLibSketch}\xspace}
\newcommand{\sygus}{\textsc{SyGuS}\xspace}
\newcommand{\semgus}{\textsc{SemGuS}\xspace}
\newcommand{\synquid}{\textsc{Synquid}\xspace}
\newcommand{\messy}{\textsc{MESSY}\xspace}

\newcommand{\numNotimeout}{40\xspace}
\newcommand{\numSolved}{41\xspace}
\newcommand{\numBenchmarks}{45\xspace}

\newcommand{\numcvcSygusFailed}{4\xspace}
\newcommand{\numcvcSygusSolved}{6\xspace}
\newcommand{\numcvcSygusBenchmarks}{10\xspace}

\newcommand{\numSketchNotimeout}{34\xspace}

\newcommand{\numSketchSolved}{35\xspace}
\newcommand{\numSketchBenchmarks}{35\xspace}

\newcommand{\numSygusEquivFail}{6\xspace}
\newcommand{\numSygusEquivBenchmarks}{8\xspace}

\newcommand{\numSygusTimeout}{3\xspace}
\newcommand{\numSygusNotimeout}{1\xspace}

\newcommand{\numSygusBenchmarks}{7\xspace}

\newcommand{\numBVNotimeout}{8\xspace}
\newcommand{\numBVSolved}{9\xspace}
\newcommand{\numBVBenchmarks}{9\xspace}

\newcommand{\numSensitivityBenchmarks}{9\xspace}
\newcommand{\numSensitivityBenchmarksTotal}{18\xspace}

\newcommand{\numSynquidResynthFailed}{4\xspace}
\newcommand{\numSynquidResynth}{12\xspace}
\newcommand{\numSynquidSynth}{16\xspace}

\newcommand{\numSynquidBenchmarks}{24\xspace}

\newcommand{\numOtherBenchmarks}{14\xspace}

\newcommand{\numDafnyProved}{23\xspace}
\newcommand{\numDafnyProvedLemma}{33\xspace}
\newcommand{\numDafnyTotal}{35\xspace}

\newcommand{\ratioSynth}{13.69\xspace}
\newcommand{\ratioSound}{26.78\xspace}
\newcommand{\ratioPrecision}{42.33\xspace}

\newcommand{\ratioLastCall}{19.61\xspace}
\newcommand{\ratioFreeze}{3.06}

\newcommand{\exname}[1]{\texttt{#1}}

\newcommand{\x}{x}
\newcommand{\y}{y}

\newcommand{\vard}{d}
\newcommand{\varo}{o}
\newcommand{\vark}{k}
\newcommand{\varl}{l}
\newcommand{\varn}{n}
\newcommand{\varm}{m}
\newcommand{\varb}{b}
\newcommand{\varq}{q}
\newcommand{\vart}{t}
\newcommand{\vars}{s}
\newcommand{\varval}{v}
\newcommand{\idx}{idx}

\newcommand{\hdist}{D_h}
\newcommand{\edist}{D_e}

\newcommand{\arr}{arr}

\newcommand{\xone}{x_1}
\newcommand{\xtwo}{x_2}
\newcommand{\xthree}{x_3}

\newcommand{\yone}{y_1}
\newcommand{\ytwo}{y_2}

\newcommand{\oone}{o_1}
\newcommand{\otwo}{o_2}
\newcommand{\valone}{v_1}
\newcommand{\valtwo}{v_2}
\newcommand{\lone}{l_1}
\newcommand{\ltwo}{l_2}
\newcommand{\sone}{s_1}
\newcommand{\stwo}{s_2}
\newcommand{\tone}{t_1}
\newcommand{\ttwo}{t_2}
\newcommand{\tn}{t_n}
\newcommand{\kone}{k_1}
\newcommand{\ktwo}{k_2}
\newcommand{\done}{d_1}
\newcommand{\dtwo}{d_2}

\newcommand{\xout}{ox}
\newcommand{\valout}{ov}
\newcommand{\sout}{os}
\newcommand{\soneout}{os_1}
\newcommand{\stwoout}{os_2}
\newcommand{\qout}{oq}
\newcommand{\tout}{ot}
\newcommand{\lout}{ol}
\newcommand{\mout}{om}
\newcommand{\loneout}{ol_1}
\newcommand{\ltwoout}{ol_2}
\newcommand{\toneout}{ot_1}
\newcommand{\ttwoout}{ot_2}

\newcommand{\error}{err}

\newcommand{\setempty}{\exname{emptySet}}
\newcommand{\setremove}{\exname{remove}}
\newcommand{\setadd}{\exname{add}}
\newcommand{\setcontains}{\exname{contains}}
\newcommand{\setsize}{\exname{size}}

\newcommand{\arrayempty}{\exname{emptyArray}}
\newcommand{\arrayadd}{\exname{add}}
\newcommand{\arrayget}{\exname{get}}

\newcommand{\mapempty}{\exname{emptyMap}}
\newcommand{\mapput}{\exname{put}}
\newcommand{\mapget}{\exname{get}}

\newcommand{\maxtwo}{\exname{max2}}
\newcommand{\maxthree}{\exname{max3}}
\newcommand{\maxfour}{\exname{max4}}
\newcommand{\maxfive}{\exname{max5}}

\newcommand{\diff}{\exname{diff}}
\newcommand{\diffone}{\exname{diff1}}
\newcommand{\difftwo}{\exname{diff2}}
\newcommand{\arraytwo}{\exname{arrSearch2}}
\newcommand{\arraythree}{\exname{arrSearch3}}

\newcommand{\opfoo}{\exname{foo}}

\newcommand{\lappend}{\exname{append}}
\newcommand{\lemptylist}{\exname{[]}}
\newcommand{\lcons}{\exname{cons}}
\newcommand{\ldelete}{\exname{deleteFirst}}
\newcommand{\ldeleteall}{\exname{delete}}
\newcommand{\ldrop}{\exname{drop}}
\newcommand{\lelem}{\exname{elem}}
\newcommand{\lelemidx}{\exname{elemIndex}}
\newcommand{\lith}{\exname{ith}}
\newcommand{\lmin}{\exname{min}} 
\newcommand{\lreplicate}{\exname{replicate}}
\newcommand{\lreverse}{\exname{reverse}\xspace}
\newcommand{\lreversetwice}{\exname{reverse\_twice}}
\newcommand{\lsnoc}{\exname{snoc}}
\newcommand{\lstutter}{\exname{stutter}}
\newcommand{\ltake}{\exname{take}}
\newcommand{\ltail}{\exname{tail}}
\newcommand{\lconsdelete}{\exname{cons\_delete}}

\newcommand{\bvsquare}{\exname{square}}
\newcommand{\bvcube}{\exname{cube}}
\newcommand{\bvhalf}{\exname{half}}
\newcommand{\bvsquareineq}{\exname{squareIneq}}
\newcommand{\bvconstneq}{\exname{constNeq}}
\newcommand{\bvconj}{\exname{conjunction}}
\newcommand{\bvsinglepoint}{\exname{singlePoint}}
\newcommand{\bvfourpoints}{\exname{fourPoints}}
\newcommand{\bvdisj}{\exname{disjunction}}

\newcommand{\telem}{\exname{elem}}
\newcommand{\tempty}{\exname{emptyTree}}
\newcommand{\tbranch}{\exname{branch}}
\newcommand{\tleft}{\exname{left}}
\newcommand{\tright}{\exname{right}}
\newcommand{\trootval}{\exname{rootval}}

\newcommand{\bstempty}{\exname{emptyBST}}
\newcommand{\bstinsert}{\exname{insert}}
\newcommand{\bstdelete}{\exname{delete}}
\newcommand{\bstfind}{\exname{elem}}

\newcommand{\stempty}{\exname{emptyStack}}
\newcommand{\stpush}{\exname{push}}
\newcommand{\stpop}{\exname{pop}}
\newcommand{\stpushpop}{\exname{push\_pop}}

\newcommand{\qempty}{\exname{emptyQueue}}
\newcommand{\qenqueue}{\exname{enqueue}}
\newcommand{\qdequeue}{\exname{dequeue}}
\newcommand{\qtolist}{\exname{toList}}

\newcommand{\iaabs}{\exname{abs}}
\newcommand{\ialinsum}{\exname{linSum}}
\newcommand{\ianonlinsum}{\exname{nonlinSum}}

\newcommand{\isempty}{\exname{isEmpty}}
\newcommand{\equal}{\exname{eq}}
\newcommand{\size}{\exname{len}}
\newcommand{\forallm}{\forall}
\newcommand{\existsm}{\exists}

\newcommand{\BV}{\mathit{BV}}
\newcommand{\BVCONJ}{{\mathit{BV}(\land)}}

\newcommand{\alphaHat}{\widehat{\alpha}}

\newcommand{\nonprop}{D}
\newcommand{\nonap}{AP}
\newcommand{\nonint}{I}
\newcommand{\nonsize}{S}
\newcommand{\nonlist}{L}
\newcommand{\nonstack}{ST}
\newcommand{\nonqueue}{Q}
\newcommand{\nontree}{T}
\newcommand{\nonhash}{H}
\newcommand{\nonmap}{M}
\newcommand{\nonkey}{K}
\newcommand{\nonvalue}{V}
\newcommand{\noncoeff}{C}
\newcommand{\nonguard}{G}
\newcommand{\inputvars}{\langle \text{input variables} \rangle}
\newcommand{\outputvars}{\langle \text{output variables} \rangle}
\newcommand{\progconsts}{\langle \text{constants from code} \rangle}

\definecolor{lgreen}{RGB}{248,255,248}
\definecolor{dgreen}{RGB}{0,128,0}
\definecolor{titlecol}{RGB}{170,199,250}
\definecolor{dred}{RGB}{200,0,0}
\definecolor{lyellow}{RGB}{255,255,40}
\definecolor{plum}{rgb}{0.56, 0.27, 0.52}

\newcommand{\hlm}[1]{\colorbox{lyellow}{$#1$}}
\newcommand{\hlg}[1]{\colorbox{lime}{$#1$}}
\newcommand{\hlr}[1]{\colorbox{pink}{$#1$}}

\usepackage{tikz}
\usetikzlibrary{shapes,snakes}
\usepackage[framemethod=tikz]{mdframed}
\usepackage{pgfplots}

\usepackage{csquotes}
\usepackage{framed}
\usepackage{arydshln}
\usepackage{afterpage}
\usepackage{enumitem}

\newenvironment{mybox}[1][gray!20]{
	\begin{tcolorbox}[   
		breakable,
		left=0pt,
		right=0pt,
		top=0pt,
		bottom=-1pt,
		colback=#1,
		colframe=#1,
		width=\dimexpr\textwidth\relax,
		boxsep=2pt,
		arc=0pt,outer arc=0pt,
		]
	}{
\end{tcolorbox}
}

\newtheorem{example}{Example}[section]
\newtheorem{definition}{Definition}[section]
\setlist[itemize]{align=parleft,left=0pt..1em, topsep=2pt}
\setlist[description]{topsep=2pt}

\begin{document}

\title{Synthesizing Specifications}

\author{Kanghee Park}
\orcid{0009-0005-7983-233X}
\email{khpark@cs.wisc.edu}
\affiliation{%
  \institution{University of Wisconsin-Madison}
  \city{}
  \country{USA}
}

\author{Loris D'Antoni}
\orcid{0000-0001-9625-4037}
\email{loris@cs.wisc.edu}
\affiliation{%
  \institution{University of Wisconsin-Madison}
  \city{}
  \country{USA}
}

\author{Thomas Reps}
\orcid{0000-0002-5676-9949}
\email{reps@cs.wisc.edu}
\affiliation{%
  \institution{University of Wisconsin-Madison}
  \city{}
  \country{USA}
}

\begin{abstract}
Every program should be accompanied by a specification that describes important aspects of the code's behavior, but writing good specifications is often harder than writing the code itself.
%
This paper addresses the problem of synthesizing specifications automatically, guided by user-supplied inputs of two kinds:
\rone a query $\query$ posed about a set of function definitions, and
\rtwo a domain-specific language $\lang$ in which each extracted property $\phil_i$ is to be expressed
(we call properties in the language \lproperties).
Each of the $\phil_i$ is a \emph{best} \lproperty for $\signatures$:
there is no other \lproperty for $\signatures$ that is strictly more precise than $\phil_i$.
Furthermore, the set $\{ \phil_i \}$ is exhaustive:
no more \lproperties can be added to it to make the conjunction $\bigwedge_i \phil_i$ more precise.

We implemented our method in a tool, \name.
The ability to modify both $\query$ and $\lang$ provides a \name user  with ways to customize the kind of specification to be synthesized.
We use this ability to show that \name can be used in a variety of applications, such as mining program specifications,
performing abstract-domain operations,
and synthesizing algebraic properties of program modules.
\end{abstract}



\maketitle

\section{Introduction}
\label{Se:Introduction}

Specifications make us understand how code behaves.
They also have many uses in testing, verifying, repairing, and synthesizing code.
Because programmers iteratively refine their code to meet a desired intent (that often changes along the way), writing and maintaining specifications is often harder than writing and maintaining the code itself.
A number of approaches have been proposed for automatically generating specifications, but these approaches are restricted to certain types of specifications, limited types of properties, and are based on dynamic testing---i.e., they yield likely specifications that though correct on the observed test cases might be unsound in general.

In this paper, we present the first customizable framework for synthesizing provably sound, most-precise (i.e., ``best'') specifications from a given set of function definitions.
Our framework can be used to mine specifications from code, but also to enable several applications where obtaining precise specifications is crucial---e.g., generating algebraic specifications for modular program synthesis~\cite{mariano2019algspecs}, automating sensitivity analysis of programs~\cite{DAntoniGAHP13}, and enabling abstract interpretation for new abstract domains \cite{DBLP:journals/pacmpl/YaoSHZ21}.
The engine/primitive that drives the framework is an algorithm for the following problem:
\textbf{\textit{
Given a query $\query$ posed about a set of function definitions, find a most-precise conjunctive formula---expressed in the user-supplied, domain-specific language (DSL) $\lang$---that is implied by $\query$.
}}

Our algorithm 
synthesizes a set of properties $\{ \phil_i \}$---each expressed in $\lang$---that are consequences of $\query$. 
Each $\phil_i$ is a most-precise \lproperty for $\signatures$:
there is no \lproperty for $\signatures$ that is strictly more precise than $\phil_i$.
Furthermore, the set $\{ \phil_i \}$ is exhaustive:
the conjunction of the properties $\bigwedge_i \phil_i$ is a best \lconjunction---i.e., no \lproperties can be added to make the conjunction more precise.

This primitive is quite flexible.
For instance, if the user desires a specification of the \emph{input-output behavior} of a function $\opfoo$, then $\query$ is ``$\textit{out} = \opfoo(\textit{in})$.''
Here the objective of specification synthesis is to perform a kind of ``projection" operation that provides information about the behavior of $\opfoo$ solely in terms of the variables visible on entry and exit (i.e., the ones visible to a client of $\opfoo$).
On the other hand, if the user wants a specification of the properties of \emph{combinations} of the stack operations $\stpush$ and $\stpop$, then $\query$ is
$\soneout {=} \stpush(\sone, \xone) \land (\stwoout, \xtwo) {=} \stpop(\stwo)$.
In this case, when given an appropriate DSL $\lang$, our tool \name synthesizes the \lproperties
\rone $\equal(\soneout, \stwo) \Rightarrow \xone {=} \xtwo$ (the element $\xout$ obtained after $(\sout, \xout) = \stpop(\stpush(s,x))$ is the pushed element $x$),
\rtwo $\equal(\soneout, \stwo) \Rightarrow \equal(\sone, \stwoout)$ (the stack $\sout$ obtained after $(\sout, \xout) = \stpop(\stpush(s,x))$ is the original stack $s$), and
\rthree $\equal(\stwoout, \sone) \wedge \xone {=} \xtwo \Rightarrow \equal(\soneout, \stwo)$ (after $\stpop(s)$, a $\stpush$ of the popped element restores stack $s$).

Our approach is different from existing specification-mining algorithms in that it meets three important objectives:
\rone \textbf{Expressiveness:}
    The synthesized specifications are not limited to a fixed type, but are customizable through the user-provided DSL.
\rtwo\textbf{Soundness:}
    The synthesized specifications are sound for all inputs to the function definitions, not just for a specified set of test cases.\footnote{As explained in \S\ref{se:problem-definition}, the framework requires a definition of the semantics of the function symbols that appear in query $\signatures$ (e.g., $\stpush$, $\stpop$, $\lreverse$, etc.) and DSL $\lang$.
The specification obtained with our framework is \emph{sound with respect to the supplied semantics}, but our implementation sometimes uses bounded or approximate semantics}.
\rthree\textbf{Precision:}
    The synthesized specifications are precise in that no specification in the DSL $\lang$ is more precise than the synthesized ones---both at the level of each individual \lproperty $\phil_i$ and at the level of the synthesized best \lconjunction $\bigwedge_i \phil_i$.\footnote{Readers familiar with symbolic methods for abstraction interpretation~\cite{DBLP:conf/vmcai/RepsSY04} will recognize that our problem is an instance of the \emph{strongest-consequence problem}.
  Given a formula $\Phi$ in logic $\mathcal{L}_1$ (with meaning function $\interp{\cdot}_1$), find the strongest formula $\Psi$ in a different logic $\mathcal{L}_2$ (with meaning function $\interp{\cdot}_2$) such that $\interp{\Phi}_1 \subseteq \interp{\Psi}_2$.
  %
  }


\mypar{The Key Challenge}
Finding a sound \lproperty for $\query$ is trivial: ``$\textit{true}$'' is always one.
However, finding a \emph{best} \lproperty is difficult because it is a kind of optimization problem, requiring \name to find a \emph{most-precise} solution.
Proving most-preciseness of an \lproperty requires showing that synthesizing a more precise \lproperty is an \emph{unrealizable} problem---i.e., no such a property exists.

\mypar{Key Ideas}
Our algorithm is based on a form of counterexample-guided inductive synthesis (CEGIS) that iteratively accumulates positive and negative examples of $\query$.
We build upon the algorithm for synthesizing abstract transformers proposed by  \citet{DBLP:journals/pacmpl/KalitaMDRR22} and generalize it to the problem of synthesizing specifications.
Our algorithm uses three key ideas (the first and third of which differ from \citeauthor{DBLP:journals/pacmpl/KalitaMDRR22}).
First, although the overall goal of the framework is to obtain a specification as a conjunctive formula, the core algorithm that underlies the synthesis process ``dives below'' this goal and focuses on synthesizing a most-precise \emph{individual conjunct} (i.e., a  most-precise \lproperty for $\signatures$).
The smaller size of individual \lproperties compared to a full conjunctive \lconjunction is better aligned with the capabilities of synthesis tools.
Second, to handle competing objectives of soundness and precision, the algorithm treats negative examples specially.
Some negative examples \textit{must} be rejected by the \lproperty we are synthesizing and some \textit{may} be rejected.
Third, to speed up progress, the algorithm accumulates \textit{must} negative examples \emph{monotonically}, so that once a sound \lproperty is identified, 
\textit{may} negative examples change status to \textit{must} negative examples, and the algorithm 
only searches for (better) \lproperties that also reject those examples.

\mypar{Our Framework}
The core algorithm is a CEGIS loop that handles some negative-example classifications as ``maybe'' constraints, and guarantees progress via monotonic constraint hardening until an \lproperty is found that is both sound and precise.
By repeatedly calling the core algorithm to synthesize incomparable \lproperties, a most-precise \lconjunction is created.

The core algorithm relies on three simple primitives.
\begin{description}
  \item[\synth:]
    synthesizes an \lproperty that accepts a set of positive examples and rejects a set of negative examples; it returns $\bot$ if no such a property exists.
  \item[\cs:]
    checks if the current \lproperty is sound;
    if is not, \cs returns a new positive example that the property fails to accept.
  \item[\cp:]
    checks if the current \lproperty is precise;
    if it is not, \cp returns a new \lproperty that accepts all positive examples, rejects all negative examples, and rejects one more negative example (which is also returned).
\end{description}
Our current implementation of each primitive relies on satisfiability modulo theory (SMT) solvers, limiting the scope of our framework.
Nevertheless, as long as one has implementations for such primitives, the algorithm is sound.
When the DSL $\lang$ is finite, the algorithm is also complete.

\mypar{Contributions} Our work makes the following contributions:
\begin{itemize}
  \item
    A formal framework for the problem of synthesizing best \lconjunctions (\S\ref{se:problem-definition}).
  \item
    An algorithm for synthesizing best \lconjunctions  (\S\ref{se:algorithm}).
  \item
    A tool that we implemented to support our framework, called \name.
    There are two instantiations of \name: \name[\sketch] and \name[\smt], which have different capabilities
    (see \S\ref{se:implementation}).
  \item
    An evaluation of \name on a variety of benchmarks, showcasing four different applications of \name (\S\ref{se:evaluation}): mining specifications~\cite{BOOK:MSSSA17}, generating algebraic specifications for modular program synthesis~\cite{mariano2019algspecs}, automating sensitivity analysis of programs~\cite{DAntoniGAHP13}, and enabling abstract interpretation for new abstract domains \cite{DBLP:journals/pacmpl/YaoSHZ21}.
  %
    
\end{itemize}
\S\ref{se:related-work} discusses related work.
\S\ref{se:conclusion} concludes.
In the extended paper~\cite{park2023synthesizing},
\S\ref{App:Proofs} contains proofs;
\S\ref{app:sketch} contains implementation details;
and \S\ref{App:data} contains further details about the evaluation.

\section{Problem Definition}
\label{se:problem-definition}

In this section, we define the problem addressed by our framework. 
Throughout the paper, we use a running example in which the goal is to synthesize interesting consequences of the following query, which allows obtaining properties of up to two calls of the list-reversal function.
\begin{equation}
  \label{eq:signatures-rev2}
  \loneout = \lreverse(\lone) \land \ltwoout = \lreverse(\ltwo),
\end{equation}
In particular, we are interested in identifying properties that are consequences of query formula \eqref{eq:signatures-rev2} and are expressible in the DSL defined by the following grammar $\lang_{\textit{list}}$:
\begin{equation}
\label{eq:grammar-rev2}    
\begin{array}{rcl}
    \nonprop  & := & \top \mid \nonap \mid \nonap \vee \nonap \mid \nonap \vee \nonap \vee \nonap \\
    \nonap & := & \isempty(\nonlist) \mid \neg \isempty(\nonlist) \mid \equal(\nonlist, \nonlist) \mid \neg \equal(\nonlist, \nonlist) \mid \nonsize ~\{ \leq \mid < \mid \geq \mid > \mid = \mid \neq \}~ \nonsize + \{0 \mid 1\}\\
    \nonsize  & := & 0 \mid \size(\nonlist) \\
    \nonlist  & := &  \lone \mid \ltwo \mid \loneout \mid \ltwoout
\end{array}
\end{equation}


An \emph{\lproperty} is a property expressible in a DSL $\lang$.
We say that an \lproperty is \emph{sound} if it is a consequence of---i.e., implied by---the given query formula $\query$.
%
The goal of our framework is to synthesize a set of incomparable \textit{sound and most-precise} \lproperties (i.e., a conjunctive specification), not just any \lproperties.
For example, the $\lang_{\textit{list}}$-property $\size(\lone)\leq \size(\loneout)$ is sound but not most-precise ($\size(\lone)= \size(\loneout)$ is a more precise sound $\lang_{\textit{list}}$-property).

In the rest of this section, we describe what a user of the framework has to provide to solve this problem, and what they obtain as output.
The user needs to provide the following inputs:

\noindent
\textbf{Input 1: Query.}
The query formula consists of a finite set of atomic formulae $\signatures=\{\signature_1, \ldots, \signature_n\}$ (denoting their conjunction). 
Each atomic formula $\sigma_i$ is of the form $o^i = f^i(x^i_1, \ldots, x^i_n)$, where $o^i$ is an output variable, each $x^i_j$ is an input variable, and $f^i$ is a function symbol.
In our running example, the query is given in Eq.~\eqref{eq:signatures-rev2}.

\noindent\textbf{Input 2: Grammar of \lproperties.}
The grammar of the DSL $\lang$ in which the synthesizer is to express properties.
In our example, the DSL $\lang_{\textit{list}}$ is defined in Eq.~\eqref{eq:grammar-rev2}.
%

\noindent\textbf{Input 3: Semantics of function symbols.}
A specification of the semantics of the function symbols that appear in query $\signatures$ (e.g., $\lreverse$) and in the DSL $\lang$ (e.g., $\size$).

We assume semantic definitions are given in---or can be translated to---formulas in some logic fragment.
%
%
For pragmatic reasons, in our implementation semantic definitions are given as code that is then automatically transformed into first-order formulas.
For instance, the semantics of $\lreverse$ is given as a program in the
$\sketch$ programming language~\cite{DBLP:journals/sttt/Solar-Lezama13} from which we automatically extract
the following bounded semantics $\varphi_\lreverse^k(l, ol)$ for a given bound $k > 0$:
\begin{equation}
\label{eq:semantics-formula-example}
\begin{array}{r@{\hspace{0.5ex}}c@{\hspace{0.5ex}}l}
\varphi^0_\lreverse(l, ol) & := & \fls
\\
\varphi^n_\lreverse(l, ol) & := &
[\varphi_\equal(l, []) \Rightarrow \varphi_\equal(ol, [])]\;\wedge \\
& & 
\exists v_{hd}, l_{tl}, l' [ \varphi_\equal(l, v_{hd} :: l_{tl}) \Rightarrow 
\varphi_\lreverse^{n-1}(l_{tl}, l') \wedge
\varphi_\lsnoc(l', v_{hd}, ol)]
\hfill \textit{ if } n > 0
\end{array}
\end{equation}
We discuss this limitation---i.e., that the semantics is bounded---and how we mitigate it in Section~\ref{se:evaluation}.

Let $\varsig$ be the set of all variables in $\signatures$.
We use $\phiprog$ to denote the formula that exactly characterizes the space of valid models over the variables $\varsig$ in $\signatures$.
For example, let $\varphi_\lreverse(l,ol)$ be the formula that exactly characterizes the result of reversing a list $l$ and storing the result in $ol$---e.g., $(l,ol)=([1,2],[2,1])$ is a valid model of $\varphi_\lreverse$.
We use $\interp{\phil}$ to denote the set of models of a formula $\phil$.
Then, in our example 
$\interp{\phiprog(\lone, \loneout,\ltwo,\ltwoout)} = \interp{\varphi_\lreverse(\lone, \loneout) \wedge \varphi_\lreverse(\ltwo, \ltwoout)}$.
Henceforth, we omit variables in formulas when no confusion should result, and merely write $\phiprog$.

\noindent\textbf{Output: Best \lproperties.}
The goal of our method is to synthesize a set of incomparable \textit{sound and most precise} \lproperties that are consequences of query $\query$.
Ideally, the best \lproperty would be one that exactly describes $\phiprog$, but in general the language $\lang$ might not be expressive enough to do so.
We argue that this feature is actually a desirable one!\footnote{
  The idea of imposing a limit on the expressibility of the language in which properties can be stated is related to the concept of \emph{inductive bias} in machine learning \cite[\S2.7]{Book:Mitchell97}: When there is no inductive bias, ``[a] concept learning algorithm is ... completely unable to generalize beyond the observed examples.''
}
The customizability of our approach via a DSL is what allows our work to focus on identifying small and readable properties (rather than complex first-order formulas), and to apply our method to different use cases (see \S\ref{se:evaluation}).

Because in general there might not be an \lproperty that is equivalent to $\phiprog$, the goal becomes instead to find \lproperties that tightly approximate $\phiprog$.
\begin{definition}[A best \lproperty]
An \lproperty $\phil$ is \emph{a best \lproperty} for a query $\signatures$
if and only if
\rone $\phil$ is \emph{sound} with respect to $\signatures$: $\interp{\phiprog} \subseteq \interp{\phil}$.
\rtwo $\phil$ is \emph{precise} with respect to $\signatures$ and $\lang$: $\neg \exists \phil' \in \lang.\, \interp{\phiprog} \subseteq \interp{\phil'} \subset \interp{\phil}$. 
We use $\bestset(\signatures)$ to denote the set of all best \lproperties for $\signatures$.
\end{definition}

When we refer to ``a sound \lproperty,'' soundness is always relative to some $\phiprog$.
Strictly speaking, we should say ``a $\phiprog$-sound \lproperty,'' but  $\phiprog$ should always be clear from context.

A best \lproperty is \textit{a strongest consequence} of $\phiprog$ that is expressible in $\lang$.
Because $\lang$ is constrained, there may be multiple, incomparable \lproperties that are all strongest consequences of $\phiprog$---thus, we speak of \textit{a} best \lproperty and not \textit{the} best \lproperty.
In our running example, $\size(\lone)= \size(\loneout)$ is a best \lproperty and so is $\neg \equal(\loneout, \ltwo) \vee \equal(\lone, \ltwoout)$. (Stated as an implication: $\equal(\loneout, \ltwo) \Rightarrow \equal(\lone, \ltwoout)$.)
The former states that the sizes of the input and output of reverse are the same, while the latter states that applying reverse twice to a list yields the same list.

The goal of this paper is to find semantically minimal sets of incomparable best \lproperties.

\begin{definition}[Best \lconjunction]
A potentially infinite set of \lproperties $\Pi = \{\phil_i\}$ forms a \emph{best \lconjunction} $\phiand=\bigwedge_{i} \phil_i$ for query $\signatures$ if and only if
\rone
    every  $\phil\in \Pi$ is a best \lproperty for $\signatures$;
  \rtwo
    every two distinct $\phil_i, \phil_j \in \Pi$ are \emph{incomparable}---i.e.,  $\interp{\phil_i}\setminus \interp{\phil_j}\neq \emptyset$ and $\interp{\phil_j}\setminus \interp{\phil_i}\neq \emptyset$;
  \rthree
    the set is
    \emph{semantically minimal}---i.e.,
    for every best \lproperty $\phil\in \bestset(\signatures)$ we have $\interp{\phiand}\subseteq \interp{\phil}$.
\end{definition}

While there can be multiple best \lconjunctions, they are all logically equivalent
and they are all equivalent to a strongest \lconjunction. However, note that a strongest \lconjunction is not necessarily a best \lconjunction; it could contain \lproperties that are not best and could potentially have repeated or redundant \lproperties.
\begin{restatable}{theorem}{bestconjunction}
If $\phiand$ is a best \lconjunction, then its interpretation coincides with the conjunction of all possible best properties:
$\interp{\phiand}=\interp{\bigwedge_{\phil\in \bestset(\signatures)} \phil}$.
\end{restatable}

We are now ready to state our problem definition:
\begin{definition}[Problem definition]
  Given query $\signatures$,
  the concrete semantics $\phiprog$ for the function symbols in $\signatures$, and a domain-specific language $\lang$ with its corresponding semantic definition, synthesize a best \lconjunction for $\signatures$. 
\end{definition}

\noindent
As illustrated in Section~\ref{se:algorithm}, given query~\eqref{eq:signatures-rev2},
the DSL in Eq.~\eqref{eq:grammar-rev2}, 
and the semantic definitions of $\lreverse$, $\isempty$, $\size$, etc., our tool \name synthesizes the set of \lproperties shown below in Eq.~\eqref{eq:reverse-prop}, and establishes that the conjunction of these properties is a best \lconjunction.
(For clarity, we write properties of the form $\neg a \vee b$ as $a\Rightarrow b$.)
\begin{equation}
\label{eq:reverse-prop}
\begin{array}{c@{\hspace{6.0ex}}c@{\hspace{6.0ex}}c}
\size(\lone) = \size(\loneout) & \size(\ltwo) = \size(\ltwoout) &
\equal(\ltwo, \ltwoout) \vee \size(\ltwo) > 1 \\
\size(\lone) > 1 \vee \equal(\lone, \loneout) &
\equal(\ltwoout, \lone) \Rightarrow \equal(\ltwo, \loneout) & \equal(\loneout, \ltwoout) \Rightarrow \equal(\lone, \ltwo)\\
\equal(\loneout, \ltwo) \Rightarrow \equal(\lone, \ltwoout) & \equal(\lone, \ltwo) \Rightarrow \equal(\loneout, \ltwoout)\\
\end{array}
\end{equation}

Even though $\lreverse$ is a simple function,
its corresponding best \lconjunction (w.r.t. the DSL $\lang$) for query~\eqref{eq:signatures-rev2} is non-trivial.
For example, our approach can discover properties involving single function calls (e.g., $\lreverse$ behaves like the identity function on a list of length 0 or 1), but also hyperproperties, i.e., properties involving multiple calls to the same function. 
For example, the property $\equal(\loneout, \ltwo) \Rightarrow \equal(\lone, \ltwoout)$ states that applying the \lreverse function twice to an input returns the same input, while the property $\equal(\loneout, \ltwoout) \Rightarrow \equal(\lone, \ltwo)$ shows that $\lreverse$ is injective!

Moreover, because the user has control over the DSL $\lang$, they can change the language in which properties are to be expressed.
In particular, if the formulas returned by \name are too complicated for the user's taste, they can modify $\lang$ and reinvoke \name until they are satisfied with the results.

Depending on the DSL, a best \lconjunction may need to be an infinite formula.

\begin{example}[Infinite \lconjunction]\label{Exa:InfiniteBestLConjunction}
Consider again the running example, and assume we change the DSL to the one defined by the following grammar $\lang_{\textit{inf}}$:
\[
\begin{array}{rcl}
    \textit{Root}   :=  \lone = \textit{CL} \land \ltwo = \loneout \Rightarrow \ltwoout = \textit{CL} && \textit{CL}  :=  \lemptylist \mid 1 :: \textit{CL} \mid 2 :: \textit{CL}
\end{array}
\]
where ``::'' denotes the infix $\lcons$ operator.
There exists only one best $\lang_{\textit{inf}}$-conjunction, which has an infinite number of conjuncts.
\[
\begin{array}{c@{\hspace{5.0ex}}c}
  \lone = \lemptylist \land \ltwo = \loneout \Rightarrow \ltwoout = \lemptylist 
  &
  \lone = 1 :: \lemptylist \land \ltwo = \loneout \Rightarrow \ltwoout = 1 :: \lemptylist
  \\
  \lone = 2 :: \lemptylist \land \ltwo = \loneout \Rightarrow \ltwoout = 2 :: \lemptylist
  &
  \lone = 1 :: 1 :: \lemptylist \land \ltwo = \loneout \Rightarrow \ltwoout = 1 :: 1 :: \lemptylist
  \\
  \lone = 1 :: 2 :: \lemptylist \land \ltwo = \loneout \Rightarrow \ltwoout = 1 :: 2 ::\lemptylist
  &
  \ldots
\end{array}
\]
\end{example}
Our implementation focuses on DSLs for which this problem does not arise. 
Assumptions on DSLs are formally discussed in Section~\ref{se:algorithm}.

All the inputs to the framework are reusable.
To synthesize a best \lconjunction for a different query $\signatures$ that still operates over lists, one only needs to supply the semantic definition of the functions in $\signatures$, and (if needed) modify the variable names generated by nonterminal $\nonlist$ of Eq.~\eqref{eq:grammar-rev2}.

For example, for the function that takes a list and duplicates its entries $\lout= \lstutter(\varl)$---e.g., $\lstutter([1,2])=[1,1,2,2]$)---\name synthesizes the following \lconjunction using the DSL $\lang_{\textit{list}}$.
\begin{equation}
\begin{array}{c}
    \size(\lout) {=} \size(\varl) {+} 1  \vee \size(\varl) {>} 1 \vee \isempty(\lout)\\ 
    \size(\lout) {\leq} 0 \vee \size(\varl) {=} 1 \vee \size(\lout) {>} \size(\varl) {+} 1\qquad
    \size(\lout) {\leq} 0 \vee \size(\lout) {>} \size(l) \\
\end{array}
\end{equation}

Because our DSL does not contain multiplication by $2$, \name could not synthesize the property that states that the length of the output list is twice the length of the input list.
If we modify the DSL $\lang_{\textit{list}}$ to contain multiplication by $2$ and the ability to describe when an element appears both in the input and output list, \name successfully synthesizes the following new best \lproperties:
\begin{equation}
\label{eq:stutter-newdsl}
\begin{array}{c}
\size(\lout) = 2 \cdot \size(\varl) \\
\forallm \varval.\, (\existsm \x {\in} \varl.\, \x {=} \varval) \Rightarrow (\existsm \x {\in} \lout.\, \x {=} \varval)
\quad \forallm \varval.\, (\existsm \x {\in} \lout.\, \x {=} \varval) \Rightarrow (\existsm \x {\in} \varl.\, \x {=} \varval) \\
\forallm \varval.\, (\forallm \x {\in} \varl.\, \x {\leq} \varval) \Rightarrow (\forallm \x {\in} \lout.\, \x {\leq} \varval)
\quad
\forallm \varval.\, (\forallm \x {\in} \lout.\, \x {\leq} \varval) \Rightarrow (\forallm \x {\in} \varl.\, \x {\leq} \varval)
\\
\forallm \varval.\, (\forallm \x {\in} \varl.\, \x {\geq} \varval) \Rightarrow (\forallm \x {\in} \lout.\, \x {\geq} \varval)
\quad
\forallm \varval.\, (\forallm \x {\in} \lout.\, \x {\geq} \varval) \Rightarrow (\forallm \x {\in} \varl.\, \x {\geq} \varval)
\end{array}
\end{equation}
The ability to modify the DSL empowers the user of \name with ways to customize the type of properties they are interested in synthesizing.
As we will show in Section~\ref{se:evaluation},
customizing the DSL also allows us to use \name for different applications and case studies---e.g., synthesizing abstract transformers and algebraic properties of programs.


\section{An Algorithm for Synthesizing Best \lconjunctions}
\label{se:algorithm}

In this section, we present the main contribution of the paper: an algorithm for synthesizing a best \lconjunction.
The algorithm synthesizes one most-precise \lproperty at a time.
It keeps track of the \lproperties it has synthesized and uses this information to synthesize a new most-precise \lproperty that is incomparable to all the ones synthesized so far.

\subsection{Positive and Negative Examples}
\label{se:examples}

Given query $\signatures$, the concrete semantics $\phiprog$ for the function symbols in $\signatures$, and a domain-specific language $\lang$ with its corresponding semantic definition,
our algorithm synthesizes best \lproperties and a best \lconjunction for $\signatures$ using an example-guided approach.
%
\begin{definition}[Examples]
Given a model $\ex$, we say that $\ex$ is a \emph{positive example} if $\ex\in\interp{\phiprog}$ and a \emph{negative example} otherwise.
\end{definition}
Given a formula $\phil$ and an example $\ex$, we write $\phil(\ex)$ to denote $\ex\in\interp{\phil}$ and $\neg\phil(\ex)$ to denote $\ex\not\in\interp{\phil}$.
Given a  set of examples $E$, we write $\phil(E)$ to denote $\wedge_{\ex\in E} \phiprog(\ex)$ and $\neg\phil(E)$ to denote $\wedge_{\ex\in E} \neg\phiprog(\ex)$.

\begin{example}
Given the query $\loneout = \lreverse(\lone)$,
the model that assigns $\lone$ to the list $[1,2]$ and $\loneout$ to the list $[2,1]$ is a positive example.
For brevity, we use the notation ${\color{dgreen}([1,2],[2,1])}$ to denote such an example. 
The following examples are negative ones: ${\color{dred}([1],[2])}$, ${\color{dred}([1,2],[1,3])}$, ${\color{dred}([1,2],[1])}$, ${\color{dred}([],[1])}$.

When considering the query from Eq.~\eqref{eq:signatures-rev2},
which contains two calls on \lreverse, ${\color{dgreen}([1,2],[2,1],[1],[1])}$ is a positive example (where the values denote $\lone$, $\loneout$, $\ltwo$, and $\ltwoout$, respectively), while ${\color{dred}([1,2],[1,3],[1],[1])}$ and ${\color{dred}([1,2],[1,3],[1],[0])}$ are negative examples.
\end{example}

%
Intuitively,  a best \lproperty must accept all positive examples while also excluding as many negative examples as possible.

Positive examples can be treated as ground truth---i.e., a best \lproperty should always accept all positive examples---but negative examples are more subtle.
First, there can be multiple, incomparable, best \lproperties, each of which rejects a different set of negative examples.
Second, there may be negative examples that no best \lproperty can reject---they are accepted by \emph{every} best \lproperty.

\begin{figure}[tb!]
     \begin{subfigure}[b]{0.48\textwidth}
        \centering
        \scalebox{0.83}{
        \begin{tikzpicture}
        \draw[draw=black,thick] (0,0) rectangle (8,5);
        \node[text width=3cm] at (5.3,4.7) {\small{$\phil_1\equiv\size(\lone) {=} \size(\loneout)$}};
        \draw[color=blue!60, thick,dashed](3.1,2.5) ellipse (2.2 and 2.1);
        \node[text width=4cm] at (6.2,0.35) {\small{$\phil_2\equiv\size(\lone) {>} 1 \vee \equal(\lone, \loneout)$}};
        \draw[color=blue!60, thick,dashed](3.7,2.1) ellipse (2.8 and 1.6);
        \draw[color=dgreen, fill=lgreen, thick](3.1,1.9) ellipse (2.0 and 1);
        \node[text width=3cm] at (3.3,2.4) {\small{$\interp{\loneout=\lreverse(\lone)}$}};

        \node[text width=3cm,color=dred] at (2.0,4.2) {\Small{([],[1, 2])}};
        \node[text width=3cm,color=dred] at (6.8,2.0) {\Small{([1,2],[1])}};        
        \node[text width=3cm,color=dred] at (3.3,4.0) {\Small{([1],[2])}};        
        \node[text width=3cm,color=dred] at (4.8,3.2) {\Small{([1,2],[1,3])}};  
        \node[text width=3cm,color=dgreen] at (5.3,1.9) {\Small{([1],[1])}};
        \end{tikzpicture}
        }
         \caption{
         {\color{dgreen}Positive} and {\color{dred}Negative} examples. Different negative examples are rejected by different best \lproperties $\phil_1$ and $\phil_2$}
         \label{fig:positive-negative-examples}
     \end{subfigure}
    \quad
     \begin{subfigure}[b]{0.48\textwidth}
         \centering
        \scalebox{0.83}{
        \begin{tikzpicture}
        \draw[draw=black,thick] (0,0) rectangle (8,5);
        \draw[color=dgreen, thick, fill=lgreen](3.1,1.9) ellipse (2.0 and 1);
        \node[text width=3cm] at (4.5,3.9) {\small{$\phil_3\equiv\equal(\lone, \loneout)$}};
        \draw[color=blue!60, thick,dashed](3.7,2.6) ellipse (2.1 and 1.0);
        \node[text width=3cm] at (3.3,2.4) {\small{$\interp{\loneout=\lreverse(\lone)}$}};
        
        \node[text width=3cm,color=dred] at (3.3,4.0) {\Small{([1],[2])}};        
        \node[text width=3cm,color=dred] at (5.5,3.0) {\Small{([1,2],[1,2])}};  
        \node[text width=3cm,color=black] at (3.0,1.5) {\Small{\hlg{([1,2],[2,1])}}};
        \node[text width=3cm,color=dgreen] at (5.3,1.9) {\Small{([1],[1])}};
        \end{tikzpicture}
        }
         \caption{A possible soundness counterexample \hlg{([1,2],[2,1])} produced when performing $\cs(\phil_3,\phiprog)$}
         \label{fig:soundness}
     \end{subfigure}
\caption{The role of examples and \cs. \label{fig:examples-and-soudness}}
\end{figure}

\begin{example}\label{Exa:positive-negative-examples}
Consider again the query $\loneout = \lreverse(\lone)$
and the diagram shown in Figure~\ref{fig:positive-negative-examples}.
The \lproperties $\phil_1 \equiv \size(\lone) = \size(\loneout)$ 
and $\phil_2 \equiv \size(\lone) {>} 1 \vee \equal(\lone, \loneout)$ are both best \lproperties.
While they both accept the positive example {\color{dgreen}$([1],[1])$} and reject the negative example {\color{dred}$([],[1, 2])$}, we can see that $\phil_1$ rejects the negative example {\color{dred}$([1,2],[1])$} whereas $\phil_2$ accepts it.
Similarly $\phil_2$ rejects the negative example {\color{dred}$([1],[2])$} whereas $\phil_1$ accepts it.
Finally, neither property rejects the negative example {\color{dred}$([1,2], [1,3])$}.
In fact, no best \lproperty in the given DSL can reject this example, and thus any best \lconjunction will accept this negative example.
\end{example}

In most cases, we use common datatypes---integers or lists---as the domain of our examples;
however, more complicated definitions are sometimes required.
If we are interested in queries involving binary search trees (BSTs), a tree datatype can describe the syntactic structure of the examples, but cannot capture BST invariants, 
e.g., for every node $n$, all values in the left (resp.\ right) subtree of $n$ must be $\leq$ (resp.\ $\geq$) $n$'s value.
In our implementation, the set of valid BSTs is defined using the following \sketch program---called a generator---that uses BST insertion operations to generate valid binary search trees:
\begin{equation*}
\label{eq:nondet-generator}
\texttt{generate\_BST}() := 
\texttt{if } (??) \texttt{ then } \bstempty()
\texttt{ else } \bstinsert(\texttt{generate\_BST}(), ??)
\end{equation*}
The code is then transformed into a bounded formula similar to the one in Eq.~\ref{eq:semantics-formula-example}, but where the holes (i.e., \texttt{??}) at each recursive call are replaced with unknown variables.
In this case, different values for the holes result in different BSTs.
In the rest of the paper, we assume that examples are only drawn from their valid domains regardless of how this domain is expressed.

\subsection{Soundness and Precision}
\label{se:soundness-precision}

Now that we have established how examples relate to best \lproperties, we can introduce the two key operations that make our algorithm work: \cs and \cp.
These two operations are similar to the ones used by \citet{DBLP:journals/pacmpl/KalitaMDRR22} to synthesize abstract transformers, and are used by the inductive-synthesis algorithm to determine whether an \lproperty is sound (i.e., a valid \lproperty) and precise (i.e., a best \lproperty).
We modify the definition of \cp proposed by \citet{DBLP:journals/pacmpl/KalitaMDRR22} to account for already synthesized best \lproperties and thus facilitate the synthesis of distinct best \lproperties.

\subsubsection{Checking Soundness}

Given an \lproperty $\phil$, $\cs(\phil, \phiprog)$ checks whether $\phil$ is an overapproximation of $\phiprog$.
In other words, \cs checks if
there exists a positive example $\posex \in \interp{\phiprog}$ that is not accepted by $\phil$;
it returns that example if it exists, and $\bot$ otherwise.
The soundness check can be expressed as 
$\exists \posex.\, \neg\phil(\posex)\wedge \phiprog(\posex)$.

\begin{example}[\cs]
Consider again the query $\loneout = \lreverse(\lone)$.
We describe how \cs operates using the example depicted in Figure~\ref{fig:soundness}.
The property $\phil_3\equiv \equal(\lone, \loneout)$ is unsound because it does not overapproximate $\phiprog$.
$\cs(\phil_3,\phiprog)$ would return a positive example that is incorrectly rejected by $\phil_3$---e.g., \emph{\hlg{$([1,2],[2,1])$}}.
$\cs(\phil_1,\phiprog)$ on the property $\phil_1\equiv\size(\lone)=\size(\loneout)$ would instead return $\bot$ because the property $\phil_1$ is sound (see Figure~\ref{fig:positive-negative-examples}).
\end{example}

\subsubsection{Checking Precision}
\label{se:CheckingPrecision}

Given an \lproperty $\phil$, a set of positive examples $\eplus$ accepted by $\phil$, a set of negative examples $\eminus$ rejected by $\phil$, and a Boolean formula $\psi$ denoting the set from which examples can be drawn, $\cp(\phil, \psi, \eplus, \eminus)$, checks whether there exist an \lproperty $\phil'$ and a negative example $\negex$ such that:
\rone $\phil'$ accepts all the positive examples in $\eplus$ and rejects all the negative examples in $\eminus$;
\rtwo $\psi(\negex)$ and $\phil'$ rejects $\negex$, whereas $\phil$ accepts $\negex$.
Formally, 
$$\exists \phil', \negex. \; \psi(\negex) \wedge \phil(\negex) \wedge \neg\phil'(\negex) \wedge \phil'(\eplus)\wedge \neg\phil'(\eminus)$$

\cp can be thought of as a primitive that synthesizes a negative example and a formula that can reject such an example at the same time (or proves whether the synthesis problem does not admit a solution).
The formula $\phil'$ is a witness that the negative example $\negex$ can be rejected by some \lproperty.
In our algorithm, the set $\psi$ is used to ensure that the negative example produced by $\cp$ is not already rejected by best \lproperties we already synthesized.

\begin{example}[\cp]
\label{ex:precision}
Consider again the query $\loneout = \lreverse(\lone)$.
We describe how \cp operates using the example depicted in Figure~\ref{fig:precision}.
The property $\phil_4\equiv \size(\lone) > 0 \vee \equal(\lone, \loneout)$ is sound but imprecise.

Figure~\ref{fig:precision-sound} shows how running $\cp(\phil_4,{\neg\phiprog},\{{\color{dgreen}([1],[1])}\},\{{\color{dred}([],[1, 2])}\})$ could for example return $\phil_2\equiv \size(\lone) > 1 \vee \equal(\lone, \loneout)$ and the negative example $\negex=\emph{\hlr{([1],[])}}$.

Figure~\ref{fig:precision-unsound} shows how running $\cp(\phil_4,{\neg\phiprog},\{{\color{dgreen}([1],[1])}\},\{{\color{dred}([],[1, 2])}\})$ could alternatively return $\phil_3\equiv\equal(\lone, \loneout)$ and the negative example $\negex={\color{dred}([1,2],[1,3])}$.
While $\phil_3$ satisfies all the requirements of \cp---i.e., it correctly classifies the current positive and negative examples---the formula $\phil_3$ is unsound because it incorrectly rejects, among others, the positive example ${\color{dgreen}([1,2],[2,1])}$.
Moreover, in this case \cp has returned the negative example $\emph{\hlr{([1,2],[1,3])}}$, which---as observed in Example~\ref{Exa:positive-negative-examples}---is not rejected by any sound \lproperty!
\end{example}

\begin{figure}[tb!]
     \begin{subfigure}[b]{0.48\textwidth}
        \centering
        \scalebox{0.82}{
        \begin{tikzpicture}
        \draw[draw=black,thick] (0,0) rectangle (8,5);
        \node[text width=4cm] at (5.3,4.7) {\small{$\phil_4\equiv\size(\lone) > 0 \vee \equal(\lone, \loneout)$}};
        \draw[color=blue!60, thick,dashed](4,2.25) ellipse (3.7 and 2.2);
        \node[text width=4cm] at (4.3,3.9) {\small{$\phil_2\equiv\size(\lone) > 1 \vee \equal(\lone, \loneout)$}};
        \draw[color=dred, thick,dashed](3.7,2.1) ellipse (2.8 and 1.6);
        \draw[color=dgreen, fill=lgreen, thick](3.1,1.9) ellipse (2.0 and 1);
        \node[text width=3cm] at (3.3,2.4) {\small{$\interp{\loneout=\lreverse(\lone)}$}};
        \node[text width=3cm,color=dred] at (2.0,4.2) {\Small{([],[1, 2])}};
        \node[text width=3cm,color=black] at (8,2.5) {\Small{\hlr{([1],[])}}};
        \node[text width=3cm,color=dgreen] at (5.3,1.9) {\Small{([1],[1])}};
        \end{tikzpicture}
        }
         \caption{
         Possible result $(\phil_2, \hlr{([1],[])})$ obtained when calling $\cp$ on $\phil_4$ and the current examples. $\phil_2$ is \textit{sound}.
         }
         \label{fig:precision-sound}
     \end{subfigure}
    \quad
     \begin{subfigure}[b]{0.48\textwidth}
         \centering
        \scalebox{0.82}{
        \begin{tikzpicture}
        \draw[draw=black,thick] (0,0) rectangle (8,5);
        \draw[color=dgreen, fill=lgreen, thick](3.1,1.9) ellipse (2.0 and 1);
        \node[text width=4cm] at (5.3,4.7) {\small{$\phil_4\equiv\size(\lone) > 0 \vee \equal(\lone, \loneout)$}};
        \draw[color=blue!60, thick,dashed](4,2.25) ellipse (3.7 and 2.2);
        \node[text width=3cm] at (4.5,3.9) {\small{$\phil_3\equiv\equal(\lone, \loneout)$}};
        \draw[color=dred, thick,dashed](3.7,2.6) ellipse (2.1 and 1.0);
        \node[text width=3cm] at (3.3,2.4) {\small{$\interp{\loneout=\lreverse(\lone)}$}};
        \node[text width=3cm,color=dred] at (2.0,4.2) {\Small{([],[1, 2])}};
        \node[text width=3cm,color=black] at (6.8,1.5) {\Small{\hlr{([1,2],[1,3])}}};
        \node[text width=3cm,color=dgreen] at (5.3,1.9) {\Small{([1],[1])}};
        \end{tikzpicture}
        }
         \caption{
         Possible result $(\phil_3, \hlr{([1,2],[1,3])})$ obtained when calling $\cp$ on $\phil_4$ and the current examples. $\phil_3$ is \textit{unsound}.
         }
         \label{fig:precision-unsound}
     \end{subfigure}
\caption{Possible results produced by \cp \label{fig:precision}}
\end{figure}

\subsubsection{Monotonicity of Sound \lproperties}
\label{subsub:monotonicity}
The fact that \cp can return (i) an unsound \lproperty, and (ii) a negative example that cannot be rejected by any sound \lproperty, makes designing a synthesis algorithm that can solve our problem challenging.

A key property that we exploit in our algorithm is that, under some assumptions on the language $\lang$, once a sound \lproperty is found, there must exist a best \lproperty that implies it.
This property allows searching for sound $\mathcal{L}$-properties in a monotonic manner. Once a sound $\mathcal{L}$-property is found, the search space can be narrowed down to a smaller set that includes all $\mathcal{L}$-properties that are more precise than the one already found.
We will show in Theorem~\ref{thm:completeness} that this narrowing is guaranteed to be finite under some assumptions about the language $\lang$.

We say that a relation $\preceq$ on a set $X$ is a \textit{well-quasi order} if $\preceq$ is a reflexive and transitive relation such that any infinite sequence of elements $x_0, x_1, x_2, \ldots$ from $X$ contains an increasing pair ${\displaystyle x_{i} \preceq x_{j}}$ with ${\displaystyle i<j.}$
If the consequence relation $\Rightarrow$ for language $\lang$ is a well-quasi order, we have that any descending sequence of sound \lproperties cannot be infinite---i.e., for any \lproperty $\phil$, there exist only finitely many \lproperties that imply $\phil$.
Moreover, a well-quasi order has no infinite anti-chains (i.e., there are no infinite sequences of pairwise incomparable elements).
Clearly, if $\lang$ is \emph{finite}, then $\Rightarrow$ is a well-quasi order, but finiteness is not a necessary condition.
For example, consider the absolute-value function $\varo=\iaabs(\x)$, and a grammar $\lang_{inf}$ that defines properties of the form $-20\leq \x \leq 10 \Rightarrow \varo \leq N$ (for any natural number $N$).
The set of properties is infinite,
but $\Rightarrow$ is a well-quasi order on the set of sound $\lang_{inf}$-properties---i.e., for any concrete value of $N$, it is only possible to decrease the value of $N$ and strengthen the property a finite number of times.

\begin{restatable}[Monotonicity of Sound \lproperties]{lemma}{monotonicity}
\label{lem:monotonicity}
If $\Rightarrow$ is a well-quasi order on the set of \lproperties, for every sound \lproperty $\phil$, there exists a best \lproperty $\phil'$ such that $\phil'\Rightarrow \phil$.
\end{restatable}

Consequently, if $\phil$ is a sound \lproperty that rejects a set of negative examples $\eminus$, there must exist a best \lproperty that also rejects the examples in $\eminus$.
This property lets us infer when a set of negative examples \textit{must} be rejected by a best \lproperty (i.e., after a sound property is found). 

\subsection{Synthesizing One Most Precise \lproperty}
\label{se:oneproperty}

We are now ready to describe the method used to synthesize an individual best \lproperty (Algorithm~\ref{alg:SynthesizeProperty}).
The procedure \synthproperty takes as input 
\begin{itemize}
    \item $\phiprog$: a formula describing the behavior of the program,
    \item $\psi$: a formula describing the domain from which examples can be drawn,
    \item $\phiinit$: an initial sound \lproperty that rejects all of the examples in $\eminusmust$,
    \item $\eplus$: a set of positive examples that \textit{must} be accepted by the returned \lproperty,
    \item $\eminusmust$: a set of negative examples that \textit{must} be rejected by the returned \lproperty,
\end{itemize}
and produces as output
\begin{itemize}
    \item a sound \lproperty $\phil$ that accepts all the examples in $\eplus$, rejects all the examples in $\eminusmust$,
    \item an updated set of positive examples $\eplus$ that the synthesized \lproperty accepts,
    \item an updated set of negative examples $\eminusmust$ that the synthesized \lproperty rejects.
\end{itemize}

We next discuss the procedure \synth, which \synthproperty uses to identify \lproperties that behave correctly on a finite set of examples.

\mypar{Synthesis from examples}
Besides \cs and \cp, which we have already discussed, \synthproperty uses the procedure $\synth(\eplus, \eminus)$, which returns---when one exists---an \lproperty $\phil$ that accepts all of the examples in $\eplus$ (i.e., $\phil(\eplus)$) and rejects all of the examples in $\eminus$ (i.e., $\neg\phil (\eminus)$).
If no such a property exists---i.e., the synthesis problem is \emph{unrealizable}---$\synth(\eplus, \eminus)=\bot$. 

\begin{example}[Example-based Synthesis]
Figure~\ref{fig:positive-negative-examples} showed that there may be negative examples that no sound \lproperty can reject.
Our algorithm uses $\eminusmay$ to handle such examples and make sure that they never end up in $\eminusmust$.
With $\lang_{\textit{list}}$ from Eq.~\eqref{eq:grammar-rev2}, if $\eplus=\{{\color{dgreen}([1],[1])}\}$, $\eminusmust=\{\}$, and $\eminusmay=\{{\color{dred}([1,2],[1,3])}\}$, then
$\synth(\eplus, \eminusmust \cup\eminusmay)$  can return the \lproperty $\size(\lone) = 1$, which is clearly unsound.
In this case, \cs will repeatedly add positive $\eplus$ (without changing $\eminusmust$ and $\eminusmay$) until $\synth(\eplus, \eminusmust \cup\eminusmay)$ returns $\bot$---e.g., when $\eplus=\{{\color{dgreen}([1],[1])}, {\color{dgreen}([1,2],[2,1])}\}$.
Even when each member of a set of negative examples can be rejected (individually) by a sound \lproperty, there may not exist a single sound \lproperty that rejects all members of the set.
For example, the negative examples $\{{\color{dred}([1],[2])},{\color{dred}([1,2],[1])}\}$ in Figure~\ref{fig:positive-negative-examples} cannot both be rejected by a single best \lproperty.
If $\eplus=\{{\color{dgreen}([1,2],[2,1])}\}$, $\eminusmust=\{{\color{dred}([1],[])}, {\color{dred}([1],[2])}\}$, and $\eminusmay=\{{\color{dred}([1,2],[1])}\}$, then
$\synth(\eplus, \eminusmust \cup\eminusmay)$  will return $\bot$.
\end{example}

\begin{algorithm}[t]
{\it
\caption{\synthproperty($\phiprog, \psi, \phiinit, \eplus, \eminusmust$)}
\label{alg:SynthesizeProperty}
\DontPrintSemicolon
$\phil, \philast \gets \phiinit\text{; } \eminusmay \gets \emptyset$ \tcp*{initial sound \lproperty that rejects $\eminusmust$}
\While{true}{
    $\posex \gets \cs(\phil, \phiprog)$ \tcp*{check soundness first}  \label{Li:CallCheckSoundness}
    \eIf {$\posex \neq \bot$} {
        $\eplus \gets \eplus \cup \{\posex\}$ \tcp*{unsound, update positive examples}  \label{Li:UpdatePlus}
        $\phil' \gets \synth(\eplus, \eminusmust \cup \eminusmay)$ \tcp*{learn a new \lproperty}    \label{Li:CallSynthesize}
        
        \eIf {$\phil' \neq \bot$} {
            $\phil \gets \phil'$ \tcp*{new candidate {\lproperty}}
        }{
            $\phil \gets \philast\text{; } \eminusmay \gets \emptyset$ \tcp*{no sound \lproperty rejects $\eminusmay$, revert to $\philast$} \label{Li:RevertLastSound}
        }
    }{
        $\eminusmust \gets \eminusmust \cup \eminusmay\text{; } \quad\eminusmay \gets \emptyset$ \tcp*{sound, so $\eminusmay$ example is added to $\eminusmust$} 
        \label{Li:MergeMayIntoMust}
        $\philast \gets \phil$ \tcp*{remember sound \lproperty}
        \label{Li:RememberSoundProperty}
        $\negex, \phil' \gets \cp(\phil, \neg \phiprog \wedge \psi, \eplus, \eminusmust)$ \tcp*{sound, check precision}    \label{Li:CallCheckPrecision}
        \eIf {$\negex \neq \bot$} {
            $\eminusmay \gets \eminusmay\cup \{\negex\}$ \tcp*{update negative examples}  \label{Li:ReinitializeMinusMay}
            $\phil \gets \phil'$  \tcp*{new candidate formula}             \label{Li:ReinitializePhi}
        }{
            \Return $\phil, \eplus, \eminusmust \cup \eminusmay$ \tcp*{sound and precise}    \label{Li:SynthPropReturn}
        }
    }
}
}
\end{algorithm}

\mypar{Preserved Invariants}
We describe \synthproperty and the invariants it maintains.

\begin{mybox}
\textbf{Invariant 1:} At the beginning of each loop iteration, the \lproperty $\phil$ accepts all the examples in the current set $\eplus$ and rejects all the examples in the current set $\eminusmust\cup\eminusmay$.
\end{mybox}

In each iteration, \synthproperty checks if the current property is sound using \cs (line~\ref{Li:CallCheckSoundness}).
If the property is sound, it is then checked for precision using \cp (line~\ref{Li:CallCheckPrecision}).
The algorithm terminates once the property is sound and precise (line~\ref{Li:SynthPropReturn}).

After a new candidate is found by \cp, the set $\eminusmay$ stores one negative example that might be rejected by a sound \lproperty.
If \cs determines that $\phil$ is sound, the example in $\eminusmay$ is added to $\eminusmust$ (line~\ref{Li:MergeMayIntoMust}) in accordance with Lemma~\ref{lem:monotonicity}.
If a new negative example is returned by \cp, together with a new property $\phil'$, the set $\eminusmay$ is reinitialized, and $\phil'$ becomes the current property to check in the next iteration (lines~\ref{Li:ReinitializeMinusMay} and~\ref{Li:ReinitializePhi}).

The next invariant is guaranteed by Lemma~\ref{lem:monotonicity} and by the fact that negative examples are added to $\eminusmust$ only when a sound property is found (line~\ref{Li:MergeMayIntoMust}).
\begin{mybox}
\textbf{Invariant 2:} There exists a sound \lproperty that rejects all the negative examples in $\eminusmust$.
\end{mybox}

When $\eminusmust$ is augmented with $\eminusmay$ in line \ref{Li:MergeMayIntoMust}, \synthproperty
stores in $\philast$ the current sound \lproperty that rejects all of the negative examples in $\eminusmust$.
In the next iteration, if a new positive example is returned by \cs, the positive examples are updated, and a new property is synthesized using \synth (line~\ref{Li:CallSynthesize}).
If \synth cannot synthesize a property that is consistent with the new set of examples, \name discards the conflicting negative example in $\eminusmay$ and reverts to the last sound property it found (line~\ref{Li:RevertLastSound}).

\begin{mybox}
\textbf{Invariant 3:} At the beginning of each loop iteration, the \lproperty $\philast$ accepts all of the examples in the current set $\eplus$ and rejects all of the examples in the current set $\eminusmust$.
\end{mybox}

Monotonically increasing the set of negative examples $\eminusmust$ when a sound \lproperty is found is one of the contributions of our algorithm.
While the algorithm is sound even without line~\ref{Li:MergeMayIntoMust}, this step prevents the algorithm from often oscillating between multiple best \lproperties throughout its execution. 
We found that this optimization gives a \ratioFreeze\% speedup to the algorithm (see ~\S\ref{se:furtheranalysis}).

\begin{example}[Algorithm~\ref{alg:SynthesizeProperty} Run]
Consider again the query $\loneout = \lreverse(\lone)$,
and a call to \synthproperty($\phiprog$, $\tru$, $\tru$, $\emptyset$, $\emptyset$)---i.e., $\phiinit=\psi=\tru$ and $\eplus=\eminusmust=\emptyset$.

\noindent\textul{Iteration 1.}
The run starts with $\cs(\top, \phiprog)$ (line \ref{Li:CallCheckSoundness}), which returns $\bot$ because the property $\top$ is sound.
Then all negative examples in $\eminusmay$ are added to $\eminusmust$ (line \ref{Li:MergeMayIntoMust}),
but because $\eminusmay=\eminusmust=\emptyset$, both sets remain empty.
\cp($\top$, ${\neg\phiprog}$, $\emptyset$, $\emptyset$) (line \ref{Li:CallCheckPrecision})
returns a new candidate $\phil_1 \equiv \equal(\lone, \loneout)$ with a negative example $\hlr{([], [1, 2])}$ (line \ref{Li:CallCheckPrecision}).
This negative example is added to $\eminusmay$ (line \ref{Li:ReinitializeMinusMay}), and the code goes back to line \ref{Li:CallCheckSoundness}.

\noindent\textul{Iteration 2.}
$\cs(\phil_1, \phiprog)$ returns a positive example $\hlg{([1, 2], [2, 1])}$, $\eplus$ is updated to $\{{\color{dgreen}([1, 2], [2, 1])}\}$ (line \ref{Li:UpdatePlus}), and $\synth(\{ {\color{dgreen}([1, 2], [2, 1])} \}, \{ {\color{dred}([], [1, 2])} \})$ (line \ref{Li:CallSynthesize}) returns a new candidate $\phil_2 \equiv \size(\lone) {\neq} 0$.

\noindent\textul{Iteration 3.}
$\cs(\phil_2, \phiprog)$ returns a positive example $\hlg{([], [])}$,
$\eplus$ is updated to $\{{\color{dgreen}([1, 2], [2, 1])}, {\color{dgreen}([], [])}\}$ (line \ref{Li:UpdatePlus}). $\synth(\eplus, \{ {\color{dred}([], [1, 2])} \})$
(line \ref{Li:CallSynthesize}) returns a new candidate $\phil_3 \equiv \size(\lone) {>} 1 \vee \equal(\lone, \loneout)$.
$\eminusmay$ is unchanged, and the code goes to line \ref{Li:CallCheckSoundness}.

\noindent\textul{Iteration 4.}
$\cs(\phil_3, \phiprog)$ returns $\bot$ because $\phil_3$ is sound.
Because the synthesizer found a sound property, the negative example in $\eminusmay$ is added to $\eminusmust$ (line \ref{Li:MergeMayIntoMust})---i.e., $\eminusmust=\{{\color{dred}([], [1, 2])}\}$.
Although $\phil_3$ is a best \lproperty,
$\cp(\phil_3, {\neg\phiprog}, \eplus,  \{ {\color{dred}([], [1, 2])} \})$ (line \ref{Li:CallCheckPrecision})  returns $\phil_4 \equiv \size(\lone) {=} \size(\loneout)$ with a negative example $\hlr{([1, 2], [1])}$.
$\eminusmay$ is set to $\{ {\color{dred}([1, 2], [1])} \}$, and the code goes to line \ref{Li:CallCheckSoundness}.
%
The current set of examples is not enough for \cp to prove that $\phil_3$ is indeed a best \lproperty---i.e., \cp was able to satisfy all requirements from \S\ref{se:CheckingPrecision} by finding a different, \emph{incomparable} \lproperty $\phil_4$ and an additional negative example.

\noindent\textul{Iteration 5.}
$\cs(\phil_4, \phiprog)$ returns $\bot$, and the negative example ${\color{dred}([1, 2], [1])}$ in $\eminusmay$ is added to $\eminusmust$ (line \ref{Li:MergeMayIntoMust}). 
$\cp(\phil_4, {\neg\phiprog}, \eplus,  \eminusmust)$ (line \ref{Li:CallCheckPrecision}) returns $\bot$, which means that we found a sound and precise \lproperty. 
\synthproperty terminates with $\eminusdelta = \emptyset$, $\phil \equiv \phil_4 \equiv \size(\lone) = \size(\loneout)$, $\eplus = \{ {\color{dgreen}([1, 2], [2, 1])}, {\color{dgreen}([], [])} \}$, and $\eminusmust = \{{\color{dred}([], [1, 2])}, {\color{dred}([1, 2], [1])}\}$.
\end{example}

The last invariant merely states that all the examples in $\eminusmay$ are elements of $\psi$.
\begin{mybox}
\textbf{Invariant 4:} For every example $\ex\in\eminusmay$, we have $\psi(\ex)$.
\end{mybox}

We say that a sound \lproperty $\phil$ is \textit{precise for $\phiprog$ with respect to} $\psi$ if there does not exists a negative example $\negex\in \interp{\psi}$ and \lproperty $\phil'$ such that $\phiprog \Rightarrow \phil'\Rightarrow \phil$ and $\phil'$ rejects $\negex$, whereas $\phil$ accepts $\negex$.
The following lemma characterizes the behavior of \synthproperty.
\begin{restatable}[Soundness and Relative Precision of \synthproperty]{lemma}{soundprecise}
\label{lem:soundprecise}
If \synthproperty terminates, it returns a sound \lproperty $\phil$ that accepts all the examples in $\eplus$, rejects all the examples in $\eminusmust \cup \eminusmay$, and is precise for $\phiprog$ with respect to $\psi$.
\end{restatable}


\subsection{Synthesizing a Most-Precise \lconjunction}
\label{se:allproperties}

In this section, we present \synthproperties (Algorithm~\ref{alg:SynthesizeAllProperties}), which uses \synthproperty to synthesize a best \lconjunction of \lproperties.

On each iteration, \synthproperties maintains a conjunction of best \lproperties $\phiand$, and uses \synthproperty to synthesize a best \lproperty that rejects some negative examples that are still accepted by $\phiand$ (i.e., negative examples in $\phiand \land \neg\phiprog$).
It also maintains the set of positive examples $\eplus$ that have been observed so far.

Each iteration performs three steps:
First, it uses \synthproperty to try to find an \lproperty $\phil$ that rejects new negative examples $\eminusmust$ that no \lproperty synthesized so far could reject---i.e, by calling \synthproperty with $\psi = \phiand \land \neg\phiprog$ (line~\ref{Li:CallSynthPropertyOne}).
    
Second, it checks whether $\phil$ rejects some example that was not rejected by $\phiand$ (line~\ref{Li:CheckForImprovement}). 
      If it does not, the algorithm terminates, and returns the \lproperties in $\Pi$ synthesized so far.
      They are all best \lproperties and their conjunction is a best \lconjunction.
    
Finally, if we reach line~\ref{Li:CallSynthPropertyTwo}, we know that $\phil$ rejects negative examples in $\eminusmust$ that $\phiand$ did not reject. 
Furthermore, because of the guarantees of \synthproperty, $\phil$ is precise with respect to $\phiand$---i.e., no sound \lproperty $\phil'$ exists that can reject more negative examples in $\phiand$ than $\phil$ could reject. 
However, there may be a more precise \lproperty that rejects more negative examples outside of $\phiand$ that $\phil$ does not reject, while still rejecting all the negative examples in $\eminusmust$.
The call to \synthproperty in line \ref{Li:CallSynthPropertyTwo} addresses this issue; it computes a best \lproperty starting from $\phil$ and makes sure that the \lproperty obtained rejects everything in $\eminusmust$ while allowing negative examples to be computed anywhere---i.e., $\psi = \tru$. (Compare this call with the one on line~\ref{Li:CallSynthPropertyOne}, which only allows negative examples to be drawn from $\phiand$.)
Because precision with respect to $\tru$ implies actual precision, we have that when $\psi=\tru$, if \synthproperty terminates, it returns the best \lproperty $\phil$ for $\phiprog$ by Lemma \ref{lem:soundprecise}.

\begin{algorithm}[t] {\it
\caption{$\synthproperties(\phiprog)$}
\label{alg:SynthesizeAllProperties}
\DontPrintSemicolon
$\phiand \gets \top$ \tcp*{conjunction of best {\lproperties}}
$\Pi \gets \emptyset$ \tcp*{set of best {\lproperties}} \label{Li:InitializePhiToEmpty}
$\eplus, \eminusmay \gets \emptyset$ \tcp*{initialize examples}  \label{Li:InitEminusMayEmpty}
\While{$\tru$}{
    \tcp{find sound $\phil$ that rejects examples in $\eminusmust$ that are still in {$\phiand$}} \label{Li:BeginWhileLoop}
    $\phil, \eplus, \eminusmust \gets \synthproperty(\phiprog,  \phiand, \tru, \eplus, \emptyset)$  \label{Li:CallSynthPropertyOne} \\
    \If{$\eminusmust = \emptyset$}
    {
        $\negex \gets$ \textsc{IsSat($\phiand \wedge \neg \phil $)} \hfill \tcp*{check if $\phil$ improves {$\phiand$}}  \label{Li:CheckForImprovement}
        \eIf{$\negex \neq \bot$} {
            $\eminusmust \gets \{ \negex \}$ \\
        }{
            \Return $\Pi$ \tcp*{return best {\lconjunction}}\label{Li:SynthPropertiesReturn}
        }
    }
    \tcp{refine $\phil$ to reject more examples outside $\phiand$ if possible}
    $\phil,\eplus, \_$ $\gets$ $\synthproperty(\phiprog, \tru, \phil, \eplus, \eminusmust$)  \label{Li:CallSynthPropertyTwo} \\
    $\Pi \gets \Pi \cup \{ \phil \}$ \tcp*{$\phil$ is a best {\lproperty}}\label{Li:UpdatePhi}
    $\phiand \gets \phiand \wedge \phil$ \tcp*{Update conjunction}\label{Li:UpdateConjunction}
}
}\end{algorithm}

\begin{example}[Algorithm~\ref{alg:SynthesizeAllProperties} Run]
Revisiting the query $\loneout = \lreverse(\lone)$, assume that $\phiand$ has been updated with the recently synthesized property $\size(\lone) = \size(\loneout)$.
We now describe the execution of \synthproperty($\phiprog$, $\phiand$, $\tru$, $\{{\color{dgreen}([1, 2], [2, 1])}, {\color{dgreen}([], [])} \}$, $\emptyset$)---i.e., $\psi=\phiand$, $\phiinit = \tru$, $\eplus = \{ {\color{dgreen}([1, 2], [2, 1])}, {\color{dgreen}([], [])} \}$ and $\eminusmust=\emptyset$.

\noindent\textul{Iteration 1.}
The run starts with $\cs(\top, \phiprog)$ (line \ref{Li:CallCheckSoundness}), which returns $\bot$ because the property $\top$ is sound.
Then all negative examples in $\eminusmay$ are added to $\eminusmust$ (line \ref{Li:MergeMayIntoMust}),
but because $\eminusmay=\eminusmust=\emptyset$, both sets remain empty.
\cp($\top$, $\phiand \land {\neg\phiprog}$, $\eplus$, $\emptyset$) (line \ref{Li:CallCheckPrecision})
returns a new candidate $\phil_1 \equiv \size(\lone) = 0 \vee \neg \equal(\lone, \loneout)$ with a negative example $\hlr{([1, 2], [1, 2])}$ (line \ref{Li:CallCheckPrecision}).
Note that this negative example satisfies $\psi = \phiand$.
This negative example is added to $\eminusmay$ (line \ref{Li:ReinitializeMinusMay}), and the code goes back to line \ref{Li:CallCheckSoundness}.

\noindent\textul{Iteration 2.}
$\cs(\phil_1, \phiprog)$ returns a positive example $\hlg{([1], [1])}$, $\eplus$ is updated to $\{{\color{dgreen}([1, 2], [2, 1])}, {\color{dgreen}([], [])}, {\color{dgreen}([1], [1])}\}$ (line \ref{Li:UpdatePlus}), and $\synth(\eplus, \eminusmust \cup \eminusmay)$ (line \ref{Li:CallSynthesize}) returns a new candidate $\phil_2 \equiv \size(\lone) > 1 \vee \equal(\lone, \loneout)$.

\noindent\textul{Iteration 3.}
$\cs(\phil_2, \phiprog)$ returns $\bot$, and the negative example $\eminusmay$ is added to $\eminusmust$ (line \ref{Li:MergeMayIntoMust}), but because $\eminusmay=\emptyset$, the set $\eminusmust$ remains the same.
$\cp(\phil_2, \phiand \land {\neg\phiprog}, \eplus,  \eminusmust)$
(line \ref{Li:CallCheckPrecision}) returns $\bot$, which means that we found a sound and precise \lproperty. 
So \synthproperty terminates with $\phil \equiv \phil_2 \equiv \size(\lone) > 1 \vee \equal(\lone, \loneout)$, $\eplus = \{ {\color{dgreen}([1, 2], [2, 1])}, {\color{dgreen}([], [])}, {\color{dgreen}([1], [1])} \}$, $\eminusmust = \{{\color{dred}([1, 2], [1, 2])}\}$ and $\eminusdelta = \emptyset$.
\end{example}

Every iteration of the loop computes a best \lproperty, which is conjoined onto $\phiand$ (line~\ref{Li:UpdateConjunction}).
\begin{mybox}
\textbf{Invariant 5:}
$\Pi$ is a set of incomparable best \lproperties.
\end{mybox}

We are now ready to show that \synthproperties is sound.
\begin{restatable}[Soundness of \synthproperties]{theorem}{soundnessconj}
\label{thm:soundnessconj}
If \synthproperties terminates, it returns a best \lconjunction for $\phiprog$.
\end{restatable}

\paragraph{Note:} In some settings,
we might know a priori that certain \lproperties hold, and we would not want to waste time synthesizing them.
In such a situation, the formula $\phiand$ in Algorithm~\ref{alg:SynthesizeAllProperties} can be initialized to hold those properties, in which case Algorithm~\ref{alg:SynthesizeAllProperties} would synthesize only best \lproperties that are not subsumed by $\phiand$. 
For example, consider synthesizing a specification for two calls of the list-reversal function,
as shown in Section~\ref{se:problem-definition}.
We can initialize $\phiand$ with
trivial properties---such as $eq(l_1, l_2) \Rightarrow eq(ol_1, ol_2)$---that are true of every function definition. 
Furthermore, after synthesizing
a property like $eq(l_2, ol_1) \Rightarrow eq(ol_2, l_1)$, we can also include the symmetric property
$eq(l_1, ol_2) \Rightarrow (ol_1, l_2)$
in the conjunction.
This approach enables us to effectively filter out redundant and trivial specifications during the synthesis process.

\subsection{Completeness}
\label{se:completeness}


%

We observe that
in \synthproperty, because positive examples in $\eplus$ are never removed, any property stronger than a property that fails \cs at line \ref{Li:CallCheckSoundness} is never considered again.
Consequently, the sequence of unsound \lproperties in an execution of \synthproperty is non-strengthening.
Thus, if a non-strengthening sequence of unsound \lproperties can only be finite, \synthproperty can only find finitely many unsound \lproperties.

Another key observation about \synthproperty is that at line \ref{Li:CallCheckPrecision} if $\cp(\phil, \ldots)$ 
returns a sound property $\phil'$ with a negative example $\negex$,
\cs will return $\bot$ in the next iteration, and the negative example $\negex$ will be added to $\eminusmust$ (from $\eminusmay$
in line~\ref{Li:MergeMayIntoMust}).
Therefore, any property weaker than $\phil$---i.e., $\phil$ on the current iteration---
will never be considered during this execution of \synthproperty.
(Recall from the definition of $\cp(\phil, \ldots)$ that $\negex$ satisfies $\phil$.)
Thus, if a non-weakening sequence of sound \lproperties can only be finite, \synthproperty can only find finitely many sound \lproperties.

Based on the above two observations,
Theorem~\ref{thm:completeness} provides a sufficient condition for our algorithm to terminate when DSL $\lang$ generates an infinite set of formulas.

\begin{restatable}[Relative Completeness]{theorem}{completeness}
\label{thm:completeness}
Suppose that $\Rightarrow$ is a well-quasi order on the set of sound \lproperties.
Let $\Leftarrow$ denote the inverse of $\Rightarrow$, and suppose that $\Leftarrow$ is a well-quasi order on the set of unsound \lproperties.
If \synth, \cs and \cp are decidable on $\lang$,
then \synthproperty and \synthproperties always terminate.
\end{restatable}

Above, our argument about the second observation involved line~\ref{Li:MergeMayIntoMust} of \synthproperty.
Nevertheless, Theorem \ref{thm:completeness} remains valid even if we eliminate line~\ref{Li:MergeMayIntoMust} from \synthproperty---i.e., line~\ref{Li:MergeMayIntoMust} is an optimization.

During each iteration, \synthproperty either adds a new positive example to $\eplus$ or adds a new negative example to $\eminusmust$. As a result, the number of iterations is also limited by the size of the example domain.
\begin{restatable}{corollary}{finite-completeness}
\label{thm:finite-completeness}
Suppose that either $\lang$ contains finitely many formulas, 
or the example domain is finite.
If \synth, \cs and \cp are decidable on $\lang$, 
then \synthproperty and \synthproperties always terminate.
\end{restatable}

\section{Implementation}
\label{se:implementation}



We implemented our framework in a tool called \name. 
Following \S\ref{se:problem-definition}, \name takes the following inputs:
\rone A query $\signatures$ for which \name is to find a best \lconjunction.
\rtwo The context-free grammar of the DSL $\lang$ in which properties are to be expressed.
\rthree A specification, as a logical formula, of the concrete semantics of the function symbols in $\signatures$ and $\lang$.
\synth and \cp may be undecidable synthesis problems in general, but we show that these primitives can be implemented in practice using program-synthesis tools that are capable of both finding solutions to synthesis problems and establishing that a problem is unrealizable (i.e., it has no solution).

We implemented two versions of \name: \name[\smt] supports problems in which semantics are definable as SMT formulas, and \name[\sketch] supports arbitrary problems but relies on the bounded/underapproximated encoding of program semantics of the \sketch language~\cite{DBLP:journals/sttt/Solar-Lezama13}. 
For the current implementations of \name[\smt] and \name[\sketch], it is necessary to give the inputs in slightly different forms.
In particular, input (iii) is provided to \name[\smt] in SMT-Lib format, whereas it is provided to \name[\sketch] as a piece of code in the \sketch programming language.

In \name[\smt], \cs is just an SMT query, and \synth and \cp can be expressed as \sygus problems.
For the latter two primitives, \name runs two \sygus solvers in parallel and returns the result of whichever terminates first:
(i) CVC5 (v.\ commit \texttt{b500e9d})~\cite{cvc5}, which is optimized for finding solutions to \sygus synthesis queries, and
(ii) a re-implementation of the constraint-based unrealizability-checking technique from~\cite{HuCDR20} that is specialized for finding whether the output of \synth and \cp is $\bot$.

In \name[\sketch], \synth, \cs, and \cp are all implemented by calling
the \sketch synthesizer (v.\ 1.7.6)~\cite{DBLP:journals/sttt/Solar-Lezama13}.
We describe how each primitive is encoded in \sketch in Appendix~\ref{app:sketch} and how \sketch's encoding affects soundness in Section~\ref{se:evaluation}.

\mypar{Timeouts}
We use a timeout threshold of 300 seconds for each call to \synth, \cs, and \cp.
If any such call times out, \synthproperties returns the current \lconjunction, together with an indication that it might not be a best \lconjunction.
(However, each of the individual conjuncts in the returned \lconjunction is a best \lproperty.)

%

\mypar{Additional Tooling}
In our evaluation, we used Dafny \cite{LeinoW14} to verify that the properties obtained by \name[\sketch] were sound for inputs beyond the bounds considered by \sketch.
Furthermore, for the \sygus benchmarks only, we invoked CVC5 to verify whether the properties obtained by \name[\sketch] \emph{exactly} characterized the function (which is a sufficient condition for an answer to be a ``best'' answer).

\section{Evaluation}
\label{se:evaluation}


We evaluated the effectiveness of \name through four case studies: specification mining (\S\ref{se:specmining}), synthesis of algebraic specifications for modular synthesis (\S\ref{se:algspec}), automating sensitivity analysis (\S\ref{se:sensitivity}), and enabling new abstract domains (\S\ref{se:bitvector}).
For each case study, we describe how we collected the benchmarks, present a \textit{quantitative analysis} of the running time and effectiveness of \name, and a \textit{qualitative analysis} of the synthesised \lconjunctions.
%
In \S\ref{se:furtheranalysis}, we describe additional experiments to identify what parameters affect \name's algorithm.

We ran all experiments on an Apple M1 8-core CPU with 8GB RAM. 
All results in this section are for the median run, selected from the results of three runs ranked by their overall synthesis time.

\subsection{Application 1: Specification Mining}
\label{se:specmining}

We considered a total of \numBenchmarks general \textit{specification-mining} problems to evaluate \name: 
\numSygusBenchmarks syntax-guided synthesis (\sygus) problems from the \sygus competition~\cite{https://doi.org/10.48550/arxiv.1904.07146}, where the semantics of operations is expressed using SMT formulas;
\numSynquidBenchmarks type-directed synthesis problems from \synquid~\cite{polikarpova2016program}, where the semantics of operations is expressed using \sketch; 
and \numOtherBenchmarks problems we designed to cover missing interesting types of properties
(11 had their semantics expressed using \sketch and 3 had semantics expressed using SMT formulas).
Cumulatively, we have \numcvcSygusBenchmarks benchmarks for which the semantics of operations is expressed using SMT formulas, and \numSketchBenchmarks benchmarks for which the semantics of operations is expressed using \sketch.
For the \sygus and \synquid benchmarks, we ``inverted'' the roles from the original benchmarks: given the reference implementation, \name synthesized a specification.
Each input problem consists of a set of functions (1 to 14 functions per problem, and the size of each function ranges from 1 to 30 lines of code per function).
The largest problem contains 14 functions (8 list functions and 6 queue functions) and the file contains a total of 140 lines of code.
Most functions are recursive and can call each other---e.g., $\qdequeue$ calls $\lreverse$, which calls $\lsnoc$, which calls $\lcons$.

For each set of similar benchmarks, we designed a DSL that contained operations that could describe interesting properties for the given set of problems. 
The construction of each DSL depended on syntactic information from the code:
the number, types, and names of input and output variables, constants, and function symbols used in the code.
We included operations that are commonly used in each category, such as equality, size primitives, and emptiness checking, but avoided problem-specific information.
For benchmarks involving data structures with structural invariants (e.g., stacks, queues, and binary search trees), we provided data-structure constructors that guaranteed that functions were only invoked with data-structure instances that satisfied the invariants.
The exact grammars are described in Appendix~\ref{App:SpecificationMiningBenchmarks}.
A DSL designed for a specific problem domain was often reused by modifying what function symbols could appear in the DSL.
Overall, we created 7 distinct grammars for 14 different \sygus and arithmetic problems; 10 grammars for 72 \synquid problems; and 3 grammars for 7 Stack and Queue problems.
Although all but one of the DSLs are finite, they are still large languages; our finite DSLs can express between 4 thousand and 14.8 trillion properties, thus making the problem of synthesizing specifications challenging.

\subsubsection{Quantitative Analysis, Part 1: Performance}
\label{se:evalq1}

\name[\smt] synthesized best \lconjunctions for 
{\numcvcSygusSolved/\numcvcSygusBenchmarks}\xspace benchmarks for which the semantics was expressed using SMT formulas.
It took less than 6 minutes each for it to solve the successful examples, and timed out on the remaining \numcvcSygusFailed benchmarks ($\maxfour$, $\arraythree$, $\iaabs$ with the grammar from Eq.\ \eqref{eq:arith-grammar-1}, and hyperproperties of $\diff$)---\name[\smt] typically times out when the synthesis algorithm requires many examples.

Although we did not consider this option in our initial set of benchmarks, for the \numcvcSygusFailed benchmarks on which \name[\smt] failed, we also encoded the semantics of the function symbols in $\signatures$ and $\lang$ using \sketch.
\name[\sketch] could synthesize properties for all \numcvcSygusFailed benchmarks, and guaranteed that {\numSygusNotimeout/\numcvcSygusFailed}\xspace were best \lconjunctions (with respect to \sketch's bounded semantics),
but for the other \numSygusTimeout benchmarks ($\maxfour$, $\arraythree$ and $\diff$) \name[\sketch] timed out on a call to \cp. 
However, the \numSygusTimeout \lconjunctions obtained by the time \cp timed out were indeed best \lconjunctions:
although most-preciseness was not shown by \name[\sketch] within the timeout threshold, we found---using an SMT solver---that the \lconjunctions in hand on the synthesis round on which the timeout occurred defined the \emph{exact} semantics of the functions of interest, which implies they were best \lconjunctions.
The \numSygusNotimeout problem for which \name[\sketch] established most-preciseness terminated within 5 minutes.
For the other \numSygusTimeout problems, if we disregard the last iteration---the one on which most-preciseness of the \lconjunction was to be established---\name[\sketch] found a best \lconjunction within 10 minutes.

\name[\sketch] could synthesize properties for {\numSketchSolved/\numSketchBenchmarks}\xspace
benchmarks for which the semantics was expressed using \sketch, and guaranteed that {\numSketchNotimeout/\numSketchBenchmarks} were best \lconjunctions.
It took less than 10 minutes to solve each List and Tree benchmark, except for the \tbranch\ problem---\name[\sketch] took about 30 minutes to find the best \lconjunction, but failed to show most-preciseness.
It took less than 15 minutes to solve each Stack, Queue, and Integer-Arithmetic benchmark. 
For \ianonlinsum, \name[\sketch] was able to synthesize in 900 seconds a best \lconjunction from a grammar that contains $\approx$14.8 trillion properties (see Eq.~\ref{eq:arith-grammar-2}).

As a baseline, we compared the running time of \name to an estimate of the running time of an algorithm that enumerates all sound properties in $\lang$.
For each benchmark, we estimated the cost of enumerating all terms in the grammar and checking for their soundness by multiplying the size $|\lang|$ of the language generated by the grammar by the average running time of each call to \cs observed when running \name on the same benchmark.
As shown in Table~\ref{tab:benchmarks-selected}, while \name demonstrated a small estimated speedup for smaller problems like $\qempty$ (i.e., 3.5$\times$, with $|\lang| = 64$), \name was 2-5 orders of magnitude faster than the baseline for problems with large languages---$10^4 \leq |\lang| \leq 6 \cdot 10^7$---and 8-10 orders of magnitude faster than the baseline for the two problems with very large languages---$10^{10} \leq |\lang| \leq 1.5 \cdot 10^{13}$.

Together, \name[\sketch] and \name[\smt] synthesized properties for {\numSolved/\numBenchmarks}\xspace  benchmarks ({45/\numBenchmarks} if we consider the 4 benchmarks rewritten using a \sketch semantics), and guaranteed that {\numNotimeout}\xspace were best \lconjunctions ({44} if we consider the 4 benchmarks rewritten using a \sketch semantics and our further analysis using an SMT solver).

\subsubsection{Quantitative Analysis, Part 2: Soundness}
%
To assess whether the properties synthesized by \name[\sketch] were indeed sound beyond the  given input bound considered by \sketch, we used an external verifier: 
Dafny~\cite{LeinoW14} (a general purpose semi-automatic verifier).
%
Dafny successfully verified that {\numDafnyProved/\numDafnyTotal}\xspace \lconjunctions synthesized by \name[\sketch] on non-\sygus benchmarks were sound without any manual input from us.
We could increase this number to {\numDafnyProvedLemma/\numDafnyTotal}\xspace by providing invariants or some logical axioms to Dafny---e.g., $(\forallm l.\, \size(l) \geq 0)$.
%
Dafny failed to verify properties synthesized from \qenqueue\xspace and \lreverse, which require a more expressive $\lang$ to describe the order of elements.

\subsubsection{Qualitative Analysis}
\label{Se:SpecificationMining:QualitativeAnalysis}

\begin{figure}
    \centering
    \scriptsize
    \[
    \arraycolsep=3pt
    \begin{array}[t]{l:l:l}
        \begin{array}[t]{l}
            \cellcolor{titlecol}\sygus\Tstrut\Bstrut\\
            \hlm{\varo {=} \maxtwo(\xone, \xtwo)} \Tstrut\Bstrut\\
            \varo {=} \xone \vee \varo {=} \xtwo\\
            \xtwo {=} \xone \vee \varo {>} \xone \vee \varo {>} \xtwo\\
            \hlm{\varo {=} \maxthree(\xone, \xtwo, \xthree)} \Tstrut\Bstrut\\
            (\xthree {\geq} \xone \wedge \xthree {\geq} \xtwo) \Rightarrow \varo {=} \xthree\\
            \xtwo {=} \xthree \vee \varo {>} \xtwo \vee \varo {>} \xthree\\
            \xone {<} \xtwo \vee \varo {=} \xone \vee \varo {=} \xthree\\
            \xtwo {=} \xone \vee \varo {>} \xtwo \vee \varo {>} \xone\\
            \xone {<} \xthree \vee \varo {=} \xone \vee \varo {=} \xtwo\\
            \hlm{\varo {=} \diff(\x, \y)} \Tstrut\Bstrut\\
            \x {=} \y {+} \varo \vee \y {=} \x {+} \varo\\
            \x {\geq} \y \Rightarrow \x {=} \varo {+} \y\\
            \x {>} \y {+} \y \vee \y {=} \x {+} \varo \vee 0 {<} \x {+} \varo \\
            \x {\leq} \y \Rightarrow \y {=} \varo {+} \x\\
            \hlm{\oone {=} \diff(\xone, \yone)} [\sketch] \Tstrut \\
            \hlm{\otwo {=} \diff(\xtwo, \ytwo)} \Bstrut\\
            (\ytwo {=} \xtwo \wedge \xone {=} \yone) \Rightarrow \otwo {=} \oone\\
            (\xtwo {=} \xone \wedge \yone {=} \ytwo) \Rightarrow \oone {=} \otwo\\
            (\xone {=} \yone \wedge \oone {=} \otwo) \Rightarrow \ytwo {=} \xtwo\\
            (\oone {=} \otwo \wedge \xtwo {=} \ytwo) \Rightarrow \yone {=} \xone\\
            (\xtwo {=} \yone \wedge \ytwo {=} \xone) \Rightarrow \otwo {=} \oone\\
            \hlm{\varo {=} \arraytwo(\xone, \xtwo, \vark)} \Tstrut\Bstrut\\
            \xtwo {\leq} \vark \vee \varo {=} 0 \vee \varo {=} 1 \\
            \vark {<} \xtwo \vee \varo {=} 0 \vee \varo {=} 2 \\
            (\vark {\geq} \xone \wedge \vark {<} \xtwo) \Rightarrow \varo {=} 1\\
            (\vark {\geq} \xone \wedge \vark {\geq} \xtwo) \Rightarrow \varo {=} 2\\
            \vark {<} \xone \Rightarrow \varo {=} 0\\
            \hlm{\varo {=} \iaabs(\x)}  \quad (\text{Eq.}~\ref{eq:arith-grammar-1}) \Tstrut\Bstrut\\
            {-} \x {\geq} \x \Rightarrow {-} \varo {=} x \quad\quad \x {\geq} {-}\x \Rightarrow \varo {=} \x \\
            \hlm{\varo {=} \iaabs(\x)} \quad (\text{Eq.}~\ref{eq:arith-grammar-2} \Tstrut\Bstrut ) [\sketch]\\
            {-}\x {+} 2\varo {=} 0  \vee \x {-} \varo {+} 2 {=} 0 \\ 
            \quad \vee \x {+} 2\varo {-}1 {>} 0 \\
            {-}\x {+} \varo {=} 0 \vee 2\x {+} 2\varo {=} 0\\
            \hlm{\varo {=} \iaabs(\x)} \quad (\text{Eq.}~\ref{eq:arith-grammar-3} \Tstrut\Bstrut  )\\
            {-}20 {\leq} \x {\leq} 10 \Rightarrow \varo {\leq} 20\\
            \cellcolor{titlecol}\text{Other: integer arithmetic}\Tstrut\Bstrut\\  
            \hlm{\varo {=} \ialinsum(\x)}\quad  (\text{Eq.}~\ref{eq:arith-grammar-1})\Tstrut\Bstrut \\
            \varo {=} \x  \vee {-} \varo {=} 0 \\
            {-}\x {=} \x \vee 1 {\leq} \varo \vee \varo {>} \x\\
            \hlm{\varo {=} \ialinsum(\x)}\quad (\text{Eq.}~\ref{eq:arith-grammar-2}) \Tstrut\Bstrut  \\
            4 \varo {=} 0 \vee 3\x {-} 3\varo {=} 0\\
            {-} 3\x {-} \varo {+} 4 {=} 0 \vee -3 \x {+} 4\varo {-} 1 {>} 0  \\
            \quad \vee 4 \x {-} \varo {=} 0 \\
            \hlm{\varo {=} \ianonlinsum(\x)}\quad  (\text{Eq.}~\ref{eq:arith-grammar-1}) \Tstrut\Bstrut\\
            \x {>} 1 \vee \varo {=} 1 \vee \varo {=} 0 \\
            {-} \varo {<} {-} \x \vee \varo {=} 1 \vee \x {=} 0 \\
            {-} \varo {=} \varo \vee \x {=} 1 \vee \varo {>} 0 \\
            \hlm{\varo {=} \ianonlinsum(\x)} \quad  (\text{Eq.}~\ref{eq:arith-grammar-2})\Tstrut\Bstrut \\
            2 \varo {=} 0 \vee - \x {+} 2\varo {-} \x^2 {=} 0 \\
            {-} 4\x {+} 2\varo {-} x^2 {=} 0 \vee 3\x {-}4\varo {+} x^2 {\leq} 0 \\
            \quad {-} \x {+} 2 \varo > 4  \\ 
            3 \x {+} \varo {>} 4 \vee {-}3\varo {=} 0
            \vee {-} \x {+} 2 \varo {-} 4x^2 {=} {-}3  \\ 
            \cellcolor{titlecol}\text{Queue}\Tstrut\Bstrut\\
            \hlm{\qout {=} \qempty()} \Tstrut\Bstrut\\
            \isempty(\qtolist(\qout))\\
            \hlm{\qout {=} \qenqueue(\varq, \x)} \\
            \equal(\qtolist(\qout), \lsnoc(\qtolist(\varq), \x))\Tstrut\Bstrut\\
            \hlm{(\qout, \x) {=} \qdequeue(\varq)} \\
            \equal(\qtolist(\varq), \lcons(\x, \qtolist(\qout)))\Tstrut\Bstrut\\
        \end{array}
         & 
        \begin{array}[t]{l}
            \cellcolor{titlecol}\text{List}\Tstrut\Bstrut\\
            \hlm{\lout {=} \lappend(\lone, \ltwo)} \Tstrut\Bstrut\\
            \size(\lout) {=} \size(\lone) {+} \size(\ltwo)\\
            \equal(\ltwo, \lout) \vee \equal(\lone, \lout) \vee \size(\lout) {>} 1 \\
            \equal(\ltwo, \lout) \vee \size(\lout) {=} \size(\ltwo) {+} 1 \vee \size(\lone) {>} 1 \\
            \equal(\lone, \lout) \vee \size(\lout) {>} \size(\lone) {+} 1 \vee \size(\ltwo) {=} 1 \\
            \hlm{\lout {=} \ldelete(\varl, \varval)} \Tstrut\Bstrut\\
            \equal(\varl, \lout) \vee \size(\varl) {=} \size(\lout) {+} 1 \\
            (\existsm  \x {\in} \varl.\, \x {<} \varval) \Rightarrow (\existsm \x {\in} \lout.\, \x {<} \varval) \\
            (\forallm  \x {\in} \varl.\, \x {\neq} \varval) \Rightarrow \equal(\varl, \lout)  \\
            (\forallm \x {\in} \varl.\, \x {\geq} \varval) \Rightarrow (\forallm  \x {\in} \lout.\, \x {\geq} \varval) \\            
            (\existsm \x {\in} \varl.\, \x {<} \varval) \vee \size(\varl) {=} \size(\lout) {+} 1 \\
            \quad \vee  (\forallm  \x {\in} \lout.\, \x {\neq} \varval)\\
            (\existsm \x {\in} \varl.\, \x {>} \varval) \Rightarrow (\existsm  \x {\in} \lout.\, \x {>} \varval) \\            
            (\existsm \x {\in} \varl.\, \x {=} \varval) \Rightarrow \size(\varl) {=} \size(\lout) {+} 1\\
            (\forallm \x {\in} \varl.\, \x {\leq} \varval) \Rightarrow (\forallm  \x {\in} \lout.\, \x {\leq} \varval) \\
            \hlm{\lout {=} \ldeleteall(\varl, \varval)} \Tstrut\Bstrut\\
            (\forallm \x {\in} \varl.\, \x {\geq} \varval) \Rightarrow (\forallm \x {\in} \lout.\, \x {>} \varval) \\
            (\forallm \x {\in} \varl.\, \x {\neq}  \varval) \Rightarrow \equal(\varl, \lout) \\
            (\existsm \x {\in} \varl.\, \x {>} \varval) \Rightarrow (\exists  \x {\in} \lout.\, \x {>} \varval) \\
            \equal(\varl, \lout) \vee \size(\varl) {\geq} \size(\lout) + 1 \\
            \forallm \x {\in} \lout.\, \x {\neq}  \varval \\
            (\existsm  \x {\in} \varl.\, \x {<} \varval) \Rightarrow (\exists \x {\in} \lout.\, \x {<} \varval)\\            
            (\forallm \x {\in} \varl.\, \x {\leq} \varval) \Rightarrow (\forallm  \x {\in} \lout.\, \x {<} \varval) \\
            \hlm{\varb {=} \ldrop(\varl, \varval)}\Tstrut\Bstrut\\
            \size(\lout) {=} \varn {+} \size(\varl) \vee \varn {>} 1 \\ 
            \quad \vee \size(\varl) {=} \size(\lout) {+} 1 \\
            \size(\varl) {=} \varn {+} \size(\lout)\\
            \equal(\lout, \varl) \vee \varn {=} 1 \vee \size(\varl) {>} \size(\lout) {+} 1\\
            \hlm{\varb {=} \lelem(\varl, \varval)}\Tstrut\Bstrut\\
            b \Leftrightarrow (\existsm  \x {\in} \varl.\, \x {=} \varval)\\
            \hlm{\idx {=} \lelemidx(\varl, \varval)}\Tstrut\Bstrut\\
            \idx {<} \size(\varl)\\
            \idx {=} {-}1 \vee (\existsm  \x {\in} \varl.\, \x {=}  \varval)\\
            \idx {>} {-}1 \vee (\forallm  \x {\in} \varl.\, \x {\neq}  \varval)\\
            (\existsm  \x {\in} \varl.\, \x {\neq} \varval) \vee 0 {=} \idx {+} 1 \vee 0 {=} \idx \\
            \hlm{\varo {=} \lmin(\varl)}\Tstrut\Bstrut\\
            \existsm  \x {\in} \varl.\, \x {=} \varo \quad\quad
            \forallm  \x {\in} \varl.\, \x {\geq} \varo\\
            \hlm{\lout {=} \lreplicate(\varn, \varval)}\Tstrut\Bstrut\\
            \size(\lout) {=} \varn \quad\quad
            \forallm \x {\in} \lout.\, \x {=} \varval\\
            \hlm{\lout {=} \lreverse(\varl)}\Tstrut\Bstrut\\
            \equal(\varl, \lout) \vee \size(\varl) {>} 1 \quad\quad
            \size(\varl) {=} \size(\lout)\\
            \hlm{\loneout {=} \lreverse(\lone)}\Tstrut\\
            \hlm{\ltwoout {=} \lreverse(\ltwo)}\Bstrut\\
            \equal(\loneout, \ltwo) \Rightarrow \equal(\ltwoout, \lone) \\
            \equal(\ltwoout, \lone) \Rightarrow \equal(\loneout, \ltwo)\\
            \equal(\ltwoout, \loneout) \Rightarrow \equal(\lone, \ltwo)\\
            \equal(\lone, \ltwo) \Rightarrow \equal(\loneout, \ltwoout)\\
            \hlm{\varval {=} \lith(\varl, \idx)}\Tstrut\Bstrut\\
            0 {<} \idx {+} 1 \quad\quad
            \size(l) {=} \idx {+} 1 \vee \size(l) {>} 1\\
            \idx {<} \size(l) \quad\quad 
            \existsm  \x {\in} \varl.\, \x {=} \varval \\
            \hlm{\lout {=} \lsnoc(\varl, \varval)}\Tstrut\Bstrut\\
            \size(\lout) {=} \size(\varl) {+} 1 \quad\quad
            \existsm  \x {\in} \varl.\, \x {=} \varval \\
            (\forallm \x {\in} \varl.\, \x {\geq} \varval) \Rightarrow (\forallm  \x {\in} \lout.\, \x {\geq} \varval)\\            
            (\existsm  \x {\in} \varl.\, \x {<} \varval)  \Rightarrow (\existsm \x {\in} \lout.\, \x {<} \varval)\\
            (\existsm  \x {\in} \varl.\, \x {\neq} \varval) \Rightarrow (\existsm  \x {\in} \lout.\, \x {\neq} \varval)\\ 
            (\existsm  \x {\in} \varl.\, \x {>} \varval) \Rightarrow (\existsm  \x {\in} \lout.\, \x {>} \varval)\\            
            (\forallm  \x {\in} \varl.\, \x {\leq} \varval)  \Rightarrow (\forallm \x {\in} \lout.\, \x {\leq} \varval) \\
            \hlm{\lout {=} \lstutter(\varl)}\Tstrut\Bstrut\\
            \size(\lout) {=} \size(\varl) {+} 1 \vee \size(\varl) {>} 1 
            \vee \equal(\varl, \lout)\\
            \isempty(\varl) \vee \size(\lout) {>} \size(l) \\
            \equal(\varl, \lout) \vee \size(\lout) {>} \size(\varl) {+} 1 \vee \size(\varl) {=} 1 \\
            \hlm{\lout {=} \ltake(\varl, \varn)}\Tstrut\Bstrut\\
            \varn {=} \size(\lout)\quad \varn {=} \size(\varl) {\vee} \size(\varl) {\geq} \size(\lout) {+} 1\\
            \equal(\varl, \lout) {\vee} \size(\varl) {=} \varn {+} 1  {\vee} \size(\varl) {>} \size(\lout) {+} 1\\        
        \end{array}         
         &  
        \begin{array}[t]{l}
            \cellcolor{titlecol}\text{Binary Tree}\Tstrut\Bstrut\\
            \hlm{\tout {=} \tempty()}\Tstrut\Bstrut\\
            \isempty(\tout)\\
            \hlm{\tout {=} \tbranch(\varval, \tone, \ttwo)}\Tstrut\Bstrut\\
            \size(\tout) {=} \size(\tone) {+} \size(\ttwo) {+} 1\\
            (\existsm  \x {\in} \tone.\, \x {>} \varval) \Rightarrow (\existsm \x {\in} \tout.\, \x {>} \varval)\\
            \existsm  \x {\in} \tout.\, \x {=} \varval\\
            (\existsm  \x {\in} \tone.\, \x {\neq} \varval) \Rightarrow (\existsm \x {\in} \tout.\, \x {\neq}  \varval)\\
            (\forallm  \x {\in} \tone.\, \x {\leq} \varval) \wedge (\forallm  \x {\in} \ttwo.\, \x {\leq} \varval)  \\
            \quad\Rightarrow (\forallm \x {\in} \tout.\, \x {\leq} \varval) \\
            (\existsm  \x {\in} \tone.\, \x {<} \varval) \Rightarrow (\existsm \x {\in} \tout.\, \x {<} \varval)\\
            (\existsm  \x {\in} \ttwo.\, \x {\neq} \varval) \Rightarrow (\existsm \x {\in} \tout.\, \x {\neq}  \varval)\\        
            (\forallm  \x {\in} \tone.\, \x {\geq} \varval) \wedge (\forallm  \x {\in} \ttwo.\, \x {\geq} \varval) \\
            \quad\Rightarrow (\forallm \x {\in} \tout.\, \x {\geq} \varval) \\
            (\existsm  \x {\in} \ttwo.\, \x {<} \varval) \Rightarrow (\existsm \x {\in} \tout.\, \x {<} \varval)\\
            (\existsm  \x {\in} \ttwo.\, \x {>} \varval) \Rightarrow (\existsm \x {\in} \tout.\, \x {>} \varval)\\
            \hlm{\varb {=} \telem(\vart, \varval)}\Tstrut\Bstrut\\
            b \Leftrightarrow (\existsm  \x {\in} \vart.\, \x {=} \varval)\\
            \hlm{\toneout {=} \tbranch(\varval, \tone, \ttwo)}\Tstrut\\
            \hlm{\ttwoout {=} \tleft(\vart)}\Bstrut\\
            \equal(\ttwo, \vart) \Rightarrow \neg \equal(\toneout, \ttwoout)\\
            \equal(\toneout, \vart) \Rightarrow \equal(\tone, \ttwoout)\\
            \equal(\tone, \vart) \Rightarrow \neg \equal(\toneout, \ttwoout)\\
            \hlm{\toneout {=} \tbranch(\varval, \tone, \ttwo)}\Tstrut\\
            \hlm{\ttwoout {=} \tright(\vart)}\Bstrut\\
            \equal(\ttwo, \vart) \Rightarrow \neg \equal(\toneout, \ttwoout)\\
            \equal(\toneout, \vart) \Rightarrow \equal(\ttwo, \ttwoout)\\
            \equal(\tone, \vart) \Rightarrow \neg \equal(\toneout, \ttwoout)\\
            \hlm{\toneout {=} \tbranch(\varval, \tone, \ttwo)}\Tstrut\\
            \hlm{\valout {=} \trootval(\vart)}\Bstrut\\
            \equal(\toneout, \vart) \Rightarrow \equal(\varval, \valout)\\
            \cellcolor{titlecol}\text{Binary Search Tree}\Tstrut\Bstrut\\
            \hlm{\tout {=} \bstempty()}\Tstrut\Bstrut\\
            \isempty(\tout)\\
            \hlm{\tout {=} \bstinsert(\vart, \varval)}\Tstrut\Bstrut\\
            \size(\tout) {=} \size(\vart) {+} 1 \vee \equal(\vart, \tout) \\
            (\existsm  \x {\in} \tout.\, \x {=} \varval) \\
            (\forallm  \x {\in} \vart.\, \x {\geq} \varval) \Rightarrow (\forallm \x {\in} \tout.\, \x {\geq} \varval)\\
            (\forallm  \x {\in} \vart.\, \x {<} \varval) {\vee} (\existsm \x {\in} \tout.\, \x {>} \varval) {\vee} \equal(\vart, \tout) \\
            (\forallm  \x {\in} \vart.\, \x {\leq} \varval) \Rightarrow (\forallm \x {\in} \tout.\, \x {\leq} \varval)\\
            (\forallm  \x {\in} \vart.\, \x {>} \varval) {\vee} (\existsm \x {\in} \tout.\, \x {<} \varval) {\vee} \equal(\vart, \tout) \\
            \hlm{\tout {=} \bstdelete(\vart, \varval)}\Tstrut\Bstrut\\
            (\forallm  \x {\in} \vart.\, \x {\neq} \varval) \Rightarrow \equal(\vart, \tout)\\
            (\exists  \x {\in} \vart.\, \x {=} \varval) \Rightarrow \size(\vart) {=} \size(\tout) {+} 1\\
            (\forallm \x {\in} \vart.\, \x {\geq} \varval) \Rightarrow (\forallm  \x {\in} \tout.\, \x {>} \varval) \\
            \forallm \x {\in} \tout.\, \x {\neq} \varval\\
            (\forallm \x {\in} \vart.\, \x {\leq} \varval) \Rightarrow (\forallm  \x {\in} \tout.\, \x {<} \varval) \\
            (\existsm  \x {\in} \vart.\, \x {<} \varval) \Rightarrow (\existsm \x {\in} \tout.\, \x {<} \varval) \\
            (\existsm  \x {\in} \vart.\, \x {>} \varval) \Rightarrow (\existsm \x {\in} \tout.\, \x {>} \varval) \\
            \hlm{\varb {=} \bstfind(\vart, \varval)}\Tstrut\Bstrut\\
            b \Rightarrow (\existsm  \x {\in} \vart.\, \x {=} \varval)\\
            (\existsm  \x {\in} \vart.\, \x {=} \varval) \Rightarrow b\\
            \cellcolor{titlecol}\text{Stack}\Tstrut\Bstrut\\ 
            \hlm{\sout {=} \stempty()} \Tstrut\Bstrut\\
            \isempty(\sout)\\
            \hlm{\sout {=} \stpush(\vars, \x)} \Tstrut\Bstrut\\
            \size(\sout) {=} \size(\vars) {+} 1\\
            \hlm{(\sout, \x) {=} \stpop(\vars)}\Tstrut\Bstrut \\
            \size(\vars) {=} \size(\sout) {+} 1\\
            \hlm{\soneout {=} \stpush(\sone, \xone)} \Tstrut\\
            \hlm{(\stwoout, \xtwo) {=} \stpop(\stwo)} \Bstrut\\
            \equal(\soneout, \stwo) \Rightarrow \xone {=} \xtwo\\
            \equal(\soneout, \stwo) \Rightarrow \equal(\sone, \stwoout)\\
            \equal(\stwoout, \sone) \wedge \xone {=} \xtwo \Rightarrow \equal(\soneout, \stwo)\\
        \end{array}         
    \end{array}
    \]
    \vspace{-3mm}
    \caption{\lproperties synthesized by \name. 
    Some properties are rewritten as implications for readability.
    }
    \label{fig:properties}
\end{figure}

Fig.~\ref{fig:properties} shows the properties synthesized by \name on one of our three runs.
In the \sygus benchmarks, ``$[\sketch]$'' denotes cases in which \name[\sketch] terminated with a semantics defined using \sketch, but \name[\smt] did not with a semantics defined using SMT formulas.
Due to space constraints, we omit $\maxfour$ and $\arraythree$.
They are similar to $\maxthree$ and $\arraytwo$, respectively, but result in many properties.

\mypar{\sygus Benchmarks}
The \lconjunctions synthesized by \name[\smt] and \name[\sketch] are more precise or equivalent to the original specifications given in the \sygus problems themselves.
In fact, \name found \lconjunctions that define the exact semantics of 
the given queries.
Inspired by this equivalence, we attempted to use a \sygus solver (CVC5) on the \sygus benchmarks to synthesize an exact formula:
we used a grammar of conjunctive properties (including the ``and'' operator, unlike the grammars used by \name), and the specification was the semantics of the function.
For {\numSygusEquivFail/\numSygusEquivBenchmarks}\xspace cases, CVC5 timed out, thus showing that our approach (of synthesizing one \lproperty at a time) is beneficial even in the artificial situation in which an oracle supplies the semantics of the best \lconjunction.
Moreover, directly synthesizing an \lconjunction---as CVC5 attempts---can yield a set of conjuncts of which some are not most-precise \lconjuncts.

%

\mypar{Synquid Benchmarks}
To evaluate the synthesized properties, we provided the synthesized \lconjunctions to \synquid and asked it to re-synthesize the reference implementation from which we extracted the properties.
In \numSynquidResynth/\numSynquidSynth cases, \synquid could re-synthesize the reference implementation.
%
%
In \numSynquidResynthFailed/\numSynquidSynth cases---\lelemidx, \lith, \lreverse, and \lstutter---the synthesized properties were not precise enough to re-synthesize the reference implementation.
For example, as stated at the end of \S\ref{se:problem-definition}, for the $\lstutter$ benchmark our DSL did not contain multiplication by $2$ and \name could not synthesize a property stating that the length of the output list is twice the length of the input list.
After modifying the DSL to contain multiplication by 2 and the ability to describe when an element appears in both the input and output lists, \name successfully synthesized 7 properties in 154.71 seconds (see Eq.~\ref{eq:stutter-newdsl}).
From the augmented set of properties, \synquid could synthesize the reference implementation of $\lstutter$.
This experiment shows how the ability to modify the DSL empowers the user of \name with ways to customize the type of properties they are interested in synthesizing.

\mypar{Other Benchmarks}
%
A core property of Stack is the principle of Last-In First-Out (LIFO).
%
\name was able to synthesize a simple formula that captures LIFO by looking at the relationship between \stpush\xspace and \stpop:
given the query
$\soneout {=} \stpush(\sone, \xone)$ and $(\stwoout, \xtwo) {=} \stpop(\stwo)$,
\name synthesized the properties
$\equal(\soneout, \stwo) \Rightarrow \xone {=} \xtwo$ and $\equal(\soneout, \stwo) \Rightarrow \equal(\sone, \stwoout)$.

A Queue is a data structure whose formal behavior is somewhat hard to describe.
Unlike Stack, the behavior of a Queue is not expressible by a simple combination of input and output variables.
\name could synthesize formulas describing the behavior of each Queue operation by providing a conversion function from a Queue consisting of two Lists into a List.
For the query
$(\qout, \x) = \qdequeue(\varq)$, \name synthesized the property $\equal(\qtolist(\varq), \lcons(\x, \qtolist(\qout)))$.


\textbf{Finding:} \name can synthesize 
sound best \lproperties and mine specifications for most of the  programs in our three categories.
Furthermore, \name synthesizes desirable properties that can be easily inspected by a user, who can then modify the DSL $\lang$ to obtain other properties if desired.

\subsection{Application 2: Synthesizing Algebraic Specifications for Modular Synthesis} 
\label{se:algspec}

In many applications of program synthesis, one has to synthesize a function that uses an existing implementation of certain external functions such as data-structure operations---i.e., synthesis is to be carried out in a modular fashion.
Even if one has to synthesize a small function implementation, the synthesizer will need to reason about the large amount of code required to represent the external functions, which can hamper performance.
\citet{mariano2019algspecs} recently proposed a new approach to modular synthesis---i.e., functions are arranged in modules---where instead of providing the synthesizer with an explicit implementation of the external functions, one provides an algebraic specification---i.e., one that does not reveal the internals of the module---of how the functions in a module operate.
For example, to describe the semantics of the functions $\setempty$, $\setadd$, $\setremove$, $\setcontains$ and $\setsize$ in a \texttt{HashSet} module, one would provide the algebraic properties in Eq.~\ref{eq:alg-spec-set}, which describe appropriate data-structure invariants, such as handling of
duplicate elements.
\begin{equation}
\label{eq:alg-spec-set}
\begin{array}{c} 
    \setcontains(\setempty, \x) = \fls \qquad
    \xone = \xtwo \Rightarrow \setcontains(\setadd(\vars, \xone), \xtwo) = \tru \\
    \xone \neq \xtwo \Rightarrow \setcontains(\setadd(\vars, \xone), \xtwo) = \setcontains(\vars, \xtwo) 
    \\
    \setremove(\setempty, \x) = \setempty \qquad
    \xone = \xtwo \Rightarrow \setremove(\setadd(\vars, \xone), \xtwo) = \setremove(\vars, \xtwo) \\
    \xone \neq \xtwo \Rightarrow \setremove(\setadd(\vars, \xone), \xtwo) = \setadd(\setremove(\vars, \xtwo), \xone) \\    
\end{array}
\end{equation}

While their approach has shown promise in terms of scalability, to use this idea in practice one has to  provide the algebraic specifications to the synthesizer \textit{manually}, a tricky task because these specifications typically define how multiple functions interact with each other.

In our case study, we used \name to synthesize algebraic specifications for benchmarks used in the evaluation of
\jlibsketch, an extension of the \sketch tool that supports algebraic specifications \cite{mariano2019algspecs}.
We considered the 3 modules---\texttt{ArrayList}, \texttt{HashSet}, and \texttt{HashMap}---that provided  algebraic specifications, did not use \texttt{string} operations
(our current implementation does not support strings), and did not require auxiliary functions that were not present in the implementation to describe the algebraic properties.
For each module, \jlibsketch contained both the algebraic specification of the module and its mock implementation---i.e., a simplified implementation that mimics the intended library's behavior (e.g., \texttt{HashSet} is implemented using an array).
Given the mock implementation of the module, we asked \name to synthesize most-precise algebraic specifications. 

For this case study, designing a grammar that accepted all possible algebraic specifications but avoided search-space explosion proved to be challenging.
Instead, we opted to create multiple grammars for each module to target different parts of the algebraic specifications, and called \name separately for each grammar.
For example, if the \jlibsketch benchmark contained an algebraic specification $\setsize(\setadd(\vars, \x)) = \setsize(\vars) + 1$, we considered the grammar to contain properties of the form $guard \Rightarrow \setsize(\setadd(\vars, \x)) = exp$.
%
All the DSLs designed for algebraic specifications synthesis were reused by modifying what function symbols could appear in the DSL.
%
The detailed grammars are presented in Appendix~\ref{App:AlgebraicSpecificationBenchmarks}.

\name terminated with a best \lconjunction for all the benchmarks (and grammars) in less than 800 seconds per benchmark.
\name was slower than the enumerative baseline presented in Section~\ref{se:evalq1} for very small languages ($|\lang| < 5$)  but faster in all other cases.
The speedups were not as prominent as for Application 1.

\begin{table}[tp]
\caption{
Evaluation results of \name. 
A few representatives benchmarks are selected from each application. 
A (*) indicates a timeout when attempting to prove precision in the last iteration.
In that case, we report as total time the time at which \name timed out. 
The Enum. column reports the estimated time required to run \cs for all formulas in the DSL $\lang$. This estimation is achieved by multiplying the size of the grammar by the average running time of the \cs.
\label{tab:benchmarks-selected}}
{\footnotesize
\setlength{\tabcolsep}{2pt}
\begin{tabular}{cccrrrrrrrrrr} 
\toprule[.1em]
\multicolumn{3}{c}{\multirow{2}{*}[-0.4ex]{Problem}}
& \multicolumn{1}{c}{\multirow{2}{*}{$|\lang|$}}
    & \multicolumn{2}{c}{\synth} & \multicolumn{2}{c}{\cs} & \multicolumn{2}{c}{\cp} & Last Iter. & Enum. & Total \\
    \cmidrule{5-13}
    &  & & & Num & T(sec) & Num & T(sec) & Num & T(sec) & T(sec) & T(sec) & T(sec) \\
\midrule[.1em]
\parbox[t]{3mm}{\multirow{22}{*}{\rotatebox[origin=c]{90}{Application 1}}} &
\parbox[t]{2mm}{\multirow{4}{*}{\rotatebox[origin=c]{90}{\sygus}}}
    & \maxtwo & $1.57 \cdot 10^5$ & 6 & 0.13 & 12 & 0.06 & 9 & 1.37 & 0.05 & 787.32 & 1.55 \\
   & & \maxthree & $8.85 \cdot 10^5$ & 19 & 13.83 & 48 & 3.53 & 36 & 131.44 & 0.46 & $6.51 \cdot 10^4$ & 148.79 \\
    & & \diff & $5.66 \cdot 10^7$ & 18 & 19.41 & 41 & 1.06 & 28 & 236.61 & 0.48 & $1.46 \cdot 10^6$ & 257.08 \\
    & & \arraytwo & $3.37 \cdot 10^6$ & 16 & 3.44 & 41 & 1.44 & 31 & 60.97 & 0.28 & $1.19 \cdot 10^5$ & 65.84 \\
\cmidrule{2-13}
 & \parbox[t]{2mm}{\multirow{2}{*}{\rotatebox[origin=c]{90}{LIA}}}
    & \iaabs\xspace (Eq.~\eqref{eq:arith-grammar-1}) & $3.37 \cdot 10^6$ & 6 & 1.14 & 14 & 0.76 & 11 & 10.16 & 0.38 & $1.83 \cdot 10^5$ & 12.05 \\
    & & \iaabs\xspace (Eq.~\eqref{eq:arith-grammar-3}) & $\infty$ & 22 & 3.98 & 23 & 0.15 & 3 & 0.67 & 0.70 & $\infty$ & 4.80 \\
\cmidrule{2-13}
& \parbox[t]{2mm}{\multirow{5}{*}{\rotatebox[origin=c]{90}{List}}}
    & \lappend & $4.97 \cdot 10^8$ & 29 & 46.62 & 68 & 81.50 & 43 & 91.46 & 57.59 & $5.95 \cdot 10^8$ & 219.58 \\    
    & & \ldeleteall & $2.99 \cdot 10^6$ & 24 & 39.49 & 66 & 87.15 & 49 & 107.07 & 18.41 & $3.94 \cdot 10^6$ & 233.71 \\
    & & \lmin & $2.38 \cdot 10^5$ & 5 & 6.02 & 14 & 13.71 & 11 & 14.42 & 2.53 & $2.38 \cdot 10^5$ & 34.15 \\
    & & \lreverse & $1.73 \cdot 10^6$ & 6 & 7.28 & 16 & 18.11 & 12 & 16.84 & 7.83 & $1.96 \cdot 10^6$ & 42.23 \\
    & & \lreverse, \lreverse & $6.40 \cdot 10^4$ & 20 & 26.64 & 52 & 96.65 & 40 & 76.37 & 32.21 & $1.19 \cdot 10^5$ & 199.66 \\
\cmidrule{2-13}
& \parbox[t]{2mm}{\multirow{4}{*}{\rotatebox[origin=c]{90}{Stack}}}
    & \stempty & $1.25 \cdot 10^5$ & 2 & 2.12 & 6 & 5.61 & 5 & 5.89 & 2.05 & $1.17 \cdot 10^5$ & 13.63 \\
    & & \stpush & $1.73 \cdot 10^6$ & 4 & 4.9 & 10 & 11.16 & 7 & 9.73 & 2.50 & $1.91 \cdot 10^6$ & 25.79 \\
    & & \stpop & $1.73 \cdot 10^6$ & 4 & 4.84 & 11 & 12.16 & 8 & 10.97 & 12.02 & $1.93 \cdot 10^6$ & 27.98 \\
    & & \stpush, \stpop & $2.62 \cdot 10^5$ & 32 & 51.12 & 77 & 95.72 & 53 & 107.13 & 27.66 & $3.26 \cdot 10^5$ & 253.97 \\
\cmidrule{2-13}
& \parbox[t]{2mm}{\multirow{3}{*}{\rotatebox[origin=c]{90}{Queue}}}
    & \qempty & 64 & 2 & 2.18 & 6 & 5.86 & 5 & 9.81 & 2.07 & 62.51 & 17.85 \\
    & & \qenqueue & $5.93 \cdot 10^5$ & 4 & 5.4 & 11 & 16.82 & 8 & 270.26 & 200.15 & $9.06 \cdot 10^5$ & 292.48 \\
    & & \qdequeue & $5.93 \cdot 10^5$ & 4 & 5.51 & 9 & 12.51 & 6 & 290.28 & 192.02 & $8.24 \cdot 10^5$ & 308.30 \\
\cmidrule{2-13}
& \parbox[t]{2mm}{\multirow{4}{*}{\rotatebox[origin=c]{90}{Arithmetic}}}    
    & \ialinsum\xspace (Eq.~\eqref{eq:arith-grammar-1}) & $1.01 \cdot 10^7$ & 7 & 6.69 & 20 & 14.50 & 15 & 15.33 & 5.33 & $7.31 \cdot 10^6$ & 36.52 \\
    & & \ialinsum\xspace (Eq.~\eqref{eq:arith-grammar-2}) & $2.90 \cdot 10^{10}$ & 15 & 15.31 & 99 & 71.72 & 88 & 112.87 & 2.70 & $2.10 \cdot 10^{10}$ & 199.90 \\
    & & \ianonlinsum\xspace (Eq.~\eqref{eq:arith-grammar-1}) & $1.01 \cdot 10^7$ & 8 & 7.95 & 30 & 21.58 & 25 & 26.48 & 2.53 & $7.25 \cdot 10^6$ & 56.01 \\
    & & \ianonlinsum\xspace (Eq.~\eqref{eq:arith-grammar-2}) & $1.48 \cdot 10^{13}$ & 17 & 16.82 & 121 & 102.71 & 107 & 734.23 & 48.21 & $1.26 \cdot 10^{13}$ & 853.77 \\
\midrule[.2em]
\parbox[t]{3mm}{\multirow{6}{*}{\rotatebox[origin=c]{90}{Application 2}}} &
\parbox[t]{2mm}{\multirow{6}{*}{\rotatebox[origin=c]{90}{HashSet}}}
    & \setsize, \setempty & 3 & 2 & 0.62 & 6 & 1.80 & 5 & 1.62 & 0.62 & 0.90 & 4.04 \\ 
    & & \setsize, \setadd & 1,315 & 7 & 2.39 & 14 & 4.56 & 9 & 3.40 & 0.70 &  428.31 & 10.36 \\ 
    & & \setcontains, \setempty & 3 & 2 & 0.60 & 6 & 1.80 & 5 & 1.64 & 0.64 & 0.90 & 4.04 \\ 
    & & \setcontains, \setadd & 55 & 5 & 1.64 & 12 & 3.94 & 9 & 3.28 & 0.69 & 18.06 & 8.86 \\ 
    & & \setremove, \setempty & 5 & 1 & 0.34 & 5 & 1.05 & 5 & 1.78 & 0.55 & 1.05 & 3.17 \\ 
    & & \setremove, \setadd & 181 & 7 & 3.38 & 15 & 6.20 & 10 & 13.23 & 3.86 & 74.81 & 22.81 \\ 
\midrule[.2em]
\parbox[t]{3mm}{\multirow{9}{*}{\rotatebox[origin=c]{90}{Application 3}}} &
\parbox[t]{2mm}{\multirow{9}{*}{\rotatebox[origin=c]{90}{Sensitivity Edit Dist.}}}
    & \lappend & 16,385 & 43 & 29.18 & 82 & 118.52 & 43 & 55.27 & 7.95 & $2.37 \cdot 10^4$ & 202.97 \\
    & & \lcons & 769 & 6 & 2.49 & 10 & 4.18 & 4 & 2.02 & 8.69 & 321.44 & 8.69 \\
    & & \lconsdelete & 769 & 19 & 9.00 & 36 & 15.80 & 19 & 12.70 & 11.75 & 337.51 & 37.50 \\ 
    & & \ldelete & 769 & 33 & 17.9 & 68 & 115.43 & 42 & 31.53 & 1.31 & 1385.38 & 164.86 \\
    & & \ldeleteall & 769 & 26 & 13.42 & 55 & 190.48 & 33 & 23.55 & 3.64 & 2663.26 & 227.45 \\
    & & \lreverse & 257 & 13 & 6.00 & 27 & 61.43 & 16 & 10.11 & 3.32 & 594.24 & 77.53 \\
    & & \lsnoc & 769 & 29 & 15.44 & 61 & 37.21 & 36 & 25.72 & 5.80 & 700.07 & 78.38 \\
    & & \lstutter & 257 & 14 & 6.29 & 26 & 13.66 & 14 & 8.51 & 3.03 & 135.02 & 28.47 \\
    & & \ltail & 257 & 14 & 6.17 & 29 & 12.83 & 17 & 10.01 & 9.66 & 113.70 & 29.01 \\
\midrule[.2em]
\parbox[t]{3mm}{\multirow{5}{*}{\rotatebox[origin=c]{90}{Application 4}}} &
\parbox[t]{2mm}{\multirow{5}{*}{\rotatebox[origin=c]{90}{BV Polyhedra}}}
    & \bvsquare & $1.68 \cdot 10^7$ & 20 & 14.79 & 76 & 24.65 & 60 & 136.76 & 68.99 & $5.44 \cdot 10^6$  & 176.19 \\
    & & \bvhalf & $1.68 \cdot 10^7$ & 19 & 16.17 & 56 & 19.27 & 40 & 353.14 & $\ast$ & $5.77 \cdot 10^6$ & 388.57 \\
    & & \bvsquareineq & $1.68 \cdot 10^7$ & 45 & 105.24 & 70 & 20.24 & 28 & 95.41 & 98.74 & $4.85 \cdot 10^6$ & 220.89 \\
    & & \bvconj & $1.68 \cdot 10^7$ & 46 & 73.08 & 67 & 20.49 & 23 & 102.32 & 117.45 & $5.13 \cdot 10^6$ & 195.89 \\
    & & \bvdisj & $1.68 \cdot 10^7$ & 48 & 178.05 & 62 & 22.69 & 15 & 102.73 & 105.71 & $6.14 \cdot 10^6$ & 303.47 \\
\bottomrule[.1em]
\bottomrule[.1em]
\end{tabular}
}
\end{table}

For all but one benchmark, the \lconjunctions synthesized by \name were equivalent to the algebraic properties manually designed by the authors of \jlibsketch.
For the implementation of \texttt{HashMap} provided in \jlibsketch, for one specific grammar, \name synthesized an empty \lconjunction (i.e., the predicate \textit{true}) instead of the algebraic specification provided by the authors of \jlibsketch---i.e., $\kone = \ktwo \Rightarrow \mapget(\mapput(\varm, \kone, \varval), \ktwo) = \varval$.
Upon further inspection, we discovered that the implementation of \texttt{HashMap} used in \jlibsketch was incorrect and did not satisfy the specification the authors provided, due to an incorrect handling of hash collision!
After fixing the bug in the implementation of HashMap, we were able to synthesize the algebraic specification.
Because algebraic properties often involve multiple functions, we were not able to separately verify their correctness on all inputs using the Dafny verifier, but the fact that we obtained the same properties that the authors of \jlibsketch specified in their benchmarks is a strong signal that our properties are indeed sound.

\textbf{Finding:} \name can help automate modular synthesis by synthesizing precise algebraic specifications that existing synthesis tools can use to speed up modular synthesis.
Thanks to \name's provable guarantees, we were able to uncover a bug in one of \jlibsketch's module implementations.

\subsection{Application 3: Automating Sensitivity Analysis}
\label{se:sensitivity}

Automatically reasoning about quantitative properties such as differential privacy~\cite{DAntoniGAHP13} in programs requires one to analyze how changes to a program input affect the program output---e.g., differential privacy typically requires that bounded changes to a function's input cause bounded changes to its output.
A common and inexpensive approach to tackle this kind of problem is to use a compositional \emph{sensitivity analysis} (either in the form of an abstract interpretation or of a type system~\cite{DAntoniGAHP13}) in which one tracks how sensitive each operation in a program is to changes in its input.
For example, one can say that the function $f(x)=\texttt{abs}(2x)$, when given two inputs $x_1$ and $x_2$ that differ by $k$, produces two outputs that differ by at most $2k$.

While for the previous function, it was pretty easy to identify a precise sensitivity property, it is generally tricky to do so for functions involving data structures, such as lists, which are of interest in differential privacy when the list represents a database of individuals~\cite{dphamming}.
In this case study, we considered \numSensitivityBenchmarks list-manipulating functions ($\lappend$, $\lcons$, $\ldelete$, $\ldeleteall$, $\lreverse$, $\lsnoc$, $\lstutter$, $\ltail$, and $\lconsdelete$) and used \name[\sketch] to synthesize precise sensitivity properties describing how changes to the input lists affect the outputs.
The $\lconsdelete$ benchmark uses the query $\ldelete(\lcons(\x, \varl), \x)$ involving the composition of two functions.

For each function $f$, we used \name to synthesize a property of the form
$$guard(x_1,x_2) \wedge \textit{dist}(\lone, \ltwo) \leq d \Rightarrow \textit{dist}(f(\lone, \xone), f(\ltwo, \xtwo)) \leq \textit{exp}$$
where $guard$ can be the predicate true or an equality/inequality between $\xone$ and $\xtwo$ (the grammars vary across benchamrks), $\textit{dist}$ is the function computing the distance between two lists (we run experiments using both edit and Hamming distance), and the expression $\textit{exp}$ (the part to synthesize) can be any linear combinations of $\size(x)$, $\size(y)$, $d$, and constants in the range -1 to 2.
When considering all combinations of functions, guards, and distances, we obtained \numSensitivityBenchmarksTotal benchmarks.
All the DSLs designed for algebraic specifications synthesis were reused by modifying what function symbols could appear in the DSL.
The complete grammars are shown in Appendix~\ref{se:sensitivity-benchmark}.

\name terminated with a best \lconjunction for all the benchmarks (and grammars) in less than 250 seconds per benchmark. 
\name outperformed the enumerative baseline presented in Section~\ref{se:evalq1} for every problem (3.14$\times$ speedup for Hamming-distance sensitivity problems and 11.93$\times$ speedup for edit-distance sensitivity problems---geometric mean).

We observed that even for simple functions, sensitivity properties are fairly complicated and hard to reason about manually. 
For example, \name synthesizes the following sensitivity \lproperty for the function $\ldelete$ ($\edist$ denotes the edit distance):
$$\edist(\lone, \ltwo) \leq d \Rightarrow \edist(\ldelete(\lone, \xone), \ldelete(\ltwo, \xtwo)) \leq d + 2$$
However, if we add a condition that the element removed from the two lists is the same, \name can synthesize the following \lproperty that further bounds the edit distance on the output:
$$\xone = \xtwo \wedge \edist(\lone, \ltwo) \leq d \Rightarrow \edist(\ldelete(\lone, \xone), \ldelete(\ltwo, \xtwo)) \leq d + 1$$
When inspecting this property, we were initially confused because we had thought that the edit distance should not increase at all if identical elements are removed.
However, that is false as illustrated by the following tricky counterexample $\lone = [1; 2; 3]$, $\ltwo = [3; 2; 3]$ and $\xone = \xtwo = 3$.
Besides the \lproperty shown above with bound $d+1$, \name also synthesized (incomparable) best \lproperties with bounds $\size(\lone) - \size(\ltwo) + 2d$ and $\size(\ltwo) - \size(\lone) + 2d$ for the same query.
All combined, these \lproperties imply that the edit distance should not increase when $d = 0$.

Because of the complexity added by the programs that compute the edit and Hamming distances, by the use of unbounded data structures, and by the fact that sensitivity properties are hyperproperties, we were not able to separately verify the soundness of they synthesized \lconjunctions on all inputs using the Dafny verifier.
However, we believe that the synthesized properties are indeed sound given that they hold for lists up to length 7---i.e., the bound imposed by \sketch.

\textbf{Finding:} \name can synthesize precise sensitivity properties for functions involving lists; the synthesized function would be challenging for a human to handcraft.

\subsection{Application 4: Enabling New Abstract Domains}
\label{se:bitvector}

One of the most powerful relational abstract domains is the domain of \emph{convex polyhedra} \cite{DBLP:conf/popl/CousotH78,DBLP:journals/scp/BagnaraHZ08}.
While programs typically operate over \texttt{int}-valued program variables for which arithmetic is performed modulo a power of 2, such as $2^{16}$ or $2^{32}$,
existing implementations of polyhedra are based on conjunctions of linear inequalities with rational coefficients over rational-valued variables.
This disconnect prevents polyhedra from precisely modeling how values wrap around in \texttt{int}/bit-vector arithmetic when arithmetic operations overflow.


Heretofore, it has not been known how to create an analog of polyhedra that is appropriate for bit-vector arithmetic.
\citeauthor{DBLP:journals/pacmpl/YaoSHZ21} \cite{DBLP:journals/pacmpl/YaoSHZ21} recently defined two domains (\ref{It:BitVectorFormulas} and \ref{It:IntVectorFormulas} below), but have only devised algorithms to support \ref{It:IntVectorFormulas}.
\begin{enumerate}[label=Version \emph{\arabic*}, align=left]
  \item
    \label{It:BitVectorFormulas}
    (bit-vector-polyhedra domain):
    conjunctions of linear bit-vector inequalities
  \item
    \label{It:IntVectorFormulas}
    (integral-polyhedra domain):
    conjunctions of linear integer inequalities
\end{enumerate}

The case study described in this section shows that \name provides a way to enable precise polyhedra operations for \ref{It:BitVectorFormulas}.
The main reason why operations for bit-vector polyhedra have not been proposed previously is that it is challenging to work with relations over bit-vector-valued variables.
For example, let $x$ and $y$ be 4-bit bit-vectors.
Fig.\ \ref{Fi:BitVectorPolyhedraExamples}(a) depicts the satisfying assignments of the inequality $x + y + 4 \le 7$ interpreted over $4$-bit unsigned modular arithmetic.
As seen in the plot, the set of points that satisfy a single bit-vector inequality can be a non-contiguous region.

\begin{figure}[tb!]
\centering
\begin{tabular}{ccc}
\resizebox{0.28\textwidth}{!}{
  \begin{tabular}{r|*{16}{c}|}
    \hhline{~----------------}
    15&\cellcolor{dgreen}&\cellcolor{dgreen}&\cellcolor{dgreen}&\cellcolor{dgreen}&\cellcolor{dgreen}& & & & & & & & &\cellcolor{dgreen}&\cellcolor{dgreen}&\cellcolor{dgreen} \\
    14&\cellcolor{dgreen}&\cellcolor{dgreen}&\cellcolor{dgreen}&\cellcolor{dgreen}&\cellcolor{dgreen}&\cellcolor{dgreen}& & & & & & & & &\cellcolor{dgreen}&\cellcolor{dgreen} \\
    13&\cellcolor{dgreen}&\cellcolor{dgreen}&\cellcolor{dgreen}&\cellcolor{dgreen}&\cellcolor{dgreen}&\cellcolor{dgreen}&\cellcolor{dgreen}& & & & & & & & &\cellcolor{dgreen} \\
    12&\cellcolor{dgreen}&\cellcolor{dgreen}&\cellcolor{dgreen}&\cellcolor{dgreen}&\cellcolor{dgreen}&\cellcolor{dgreen}&\cellcolor{dgreen}&\cellcolor{dgreen}& & & & & & & &  \\
    11& &\cellcolor{dgreen}&\cellcolor{dgreen}&\cellcolor{dgreen}&\cellcolor{dgreen}&\cellcolor{dgreen}&\cellcolor{dgreen}&\cellcolor{dgreen}&\cellcolor{dgreen}& & & & & & &  \\
    10& & &\cellcolor{dgreen}&\cellcolor{dgreen}&\cellcolor{dgreen}&\cellcolor{dgreen}&\cellcolor{dgreen}&\cellcolor{dgreen}&\cellcolor{dgreen}&\cellcolor{dgreen}& & & & & &  \\
     9& & & &\cellcolor{dgreen}&\cellcolor{dgreen}&\cellcolor{dgreen}&\cellcolor{dgreen}&\cellcolor{dgreen}&\cellcolor{dgreen}&\cellcolor{dgreen}&\cellcolor{dgreen}& & & & &  \\
     8& & & & &\cellcolor{dgreen}&\cellcolor{dgreen}&\cellcolor{dgreen}&\cellcolor{dgreen}&\cellcolor{dgreen}&\cellcolor{dgreen}&\cellcolor{dgreen}&\cellcolor{dgreen}& & & &  \\
     7& & & & & &\cellcolor{dgreen}&\cellcolor{dgreen}&\cellcolor{dgreen}&\cellcolor{dgreen}&\cellcolor{dgreen}&\cellcolor{dgreen}&\cellcolor{dgreen}&\cellcolor{dgreen}& & &  \\
     6& & & & & & &\cellcolor{dgreen}&\cellcolor{dgreen}&\cellcolor{dgreen}&\cellcolor{dgreen}&\cellcolor{dgreen}&\cellcolor{dgreen}&\cellcolor{dgreen}&\cellcolor{dgreen}& &  \\
     5& & & & & & & &\cellcolor{dgreen}&\cellcolor{dgreen}&\cellcolor{dgreen}&\cellcolor{dgreen}&\cellcolor{dgreen}&\cellcolor{dgreen}&\cellcolor{dgreen}&\cellcolor{dgreen}&  \\
     4& & & & & & & & &\cellcolor{dgreen}&\cellcolor{dgreen}&\cellcolor{dgreen}&\cellcolor{dgreen}&\cellcolor{dgreen}&\cellcolor{dgreen}&\cellcolor{dgreen}&\cellcolor{dgreen} \\
     3&\cellcolor{dgreen}& & & & & & & & &\cellcolor{dgreen}&\cellcolor{dgreen}&\cellcolor{dgreen}&\cellcolor{dgreen}&\cellcolor{dgreen}&\cellcolor{dgreen}&\cellcolor{dgreen} \\
     2&\cellcolor{dgreen}&\cellcolor{dgreen}& & & & & & & & &\cellcolor{dgreen}&\cellcolor{dgreen}&\cellcolor{dgreen}&\cellcolor{dgreen}&\cellcolor{dgreen}&\cellcolor{dgreen} \\
     1&\cellcolor{dgreen}&\cellcolor{dgreen}&\cellcolor{dgreen}& & & & & & & & &\cellcolor{dgreen}&\cellcolor{dgreen}&\cellcolor{dgreen}&\cellcolor{dgreen}&\cellcolor{dgreen} \\
     0&\cellcolor{dgreen}&\cellcolor{dgreen}&\cellcolor{dgreen}&\cellcolor{dgreen}& & & & & & & & &\cellcolor{dgreen}&\cellcolor{dgreen}&\cellcolor{dgreen}&\cellcolor{dgreen} \\
    \hhline{~----------------}
    \multicolumn{1}{c}{ }&0&1&2&3&4&5&6&7&8&9&\hspace{-4pt}10\hspace{-4pt}&\hspace{-4pt}11\hspace{-4pt}&\hspace{-4pt}12\hspace{-4pt}&\hspace{-4pt}13\hspace{-4pt}&\hspace{-4pt}14\hspace{-4pt}&\multicolumn{1}{c}{\hspace{-4pt}15\hspace{-4pt}}
  \end{tabular}
 }
&
\resizebox{0.28\textwidth}{!}{
  \begin{tabular}{r|*{16}{c}|}
    \hhline{~----------------}
    15& & & & & & & & & & & & & & & &  \\
    14& & & & & & & & & & & & & & & &  \\
    13& & & & & & & & & & & & & & & &  \\
    12& & & & & & & & & & & & & & & &  \\
    11& & & & & & & & & & & & & & & &  \\
    10& & & & & & & & & & & & & & & &  \\
     9& & & &\cellcolor{dgreen}& &\cellcolor{dgreen}& & & & & &\cellcolor{dgreen}& &\cellcolor{dgreen}& &  \\
     8& & & & & & & & & & & & & & & &  \\
     7& & & & & & & & & & & & & & & &  \\
     6& & & & & & & & & & & & & & & &  \\
     5& & & & & & & & & & & & & & & &  \\
     4& & &\cellcolor{dgreen}& & & &\cellcolor{dgreen}& & & &\cellcolor{dgreen}& & & &\cellcolor{dgreen}&  \\
     3& & & & & & & & & & & & & & & &  \\
     2& & & & & & & & & & & & & & & &  \\
     1& &\cellcolor{dgreen}& & & & & &\cellcolor{dgreen}& &\cellcolor{dgreen}& & & & & &\cellcolor{dgreen} \\
     0&\cellcolor{dgreen}& & & &\cellcolor{dgreen}& & & &\cellcolor{dgreen}& & & &\cellcolor{dgreen}& & &  \\
    \hhline{~----------------}
    \multicolumn{1}{c}{ }&0&1&2&3&4&5&6&7&8&9&\hspace{-4pt}10\hspace{-4pt}&\hspace{-4pt}11\hspace{-4pt}&\hspace{-4pt}12\hspace{-4pt}&\hspace{-4pt}13\hspace{-4pt}&\hspace{-4pt}14\hspace{-4pt}&\multicolumn{1}{c}{\hspace{-4pt}15\hspace{-4pt}}
  \end{tabular}
}
&
\resizebox{0.28\textwidth}{!}{
  \begin{tabular}{r|*{16}{c}|}
    \hhline{~----------------}
    15& & & & & & & & & & & & & & & &  \\
    14& & & & & & & & & & & & & & & &  \\
    13& & & & & & & & & & & & & & & &  \\
    12& & & & & & & & & & & & & & & &  \\
    11& & & & & & & & & & & & & & & &  \\
    10& & & & & & & & & & & & & & & &  \\
     9& & & &\cellcolor{dgreen}& &\cellcolor{dgreen}& & & & & &\cellcolor{dgreen}& &\cellcolor{dgreen}& &  \\
     8& & & &\cellcolor{dgreen}& &\cellcolor{dgreen}& & & & & &\cellcolor{dgreen}& &\cellcolor{dgreen}& &  \\
     7& & & &\cellcolor{dgreen}& &\cellcolor{dgreen}& & & & & &\cellcolor{dgreen}& &\cellcolor{dgreen}& &  \\
     6& & & &\cellcolor{dgreen}& &\cellcolor{dgreen}& & & & & &\cellcolor{dgreen}& &\cellcolor{dgreen}& &  \\
     5& & & &\cellcolor{dgreen}& &\cellcolor{dgreen}& & & & & &\cellcolor{dgreen}& &\cellcolor{dgreen}& &  \\
     4& & &\cellcolor{dgreen}&\cellcolor{dgreen}& &\cellcolor{dgreen}&\cellcolor{dgreen}& & & &\cellcolor{dgreen}&\cellcolor{dgreen}& &\cellcolor{dgreen}&\cellcolor{dgreen}&  \\
     3&\cellcolor{dred}& &\cellcolor{dgreen}&\cellcolor{dgreen}&\cellcolor{dred}&\cellcolor{dgreen}&\cellcolor{dgreen}& &\cellcolor{dred}& &\cellcolor{dgreen}&\cellcolor{dgreen}&\cellcolor{dred}&\cellcolor{dgreen}&\cellcolor{dgreen}&  \\
     2&\cellcolor{dred}& &\cellcolor{dgreen}&\cellcolor{dgreen}&\cellcolor{dred}&\cellcolor{dgreen}&\cellcolor{dgreen}& &\cellcolor{dred}& &\cellcolor{dgreen}&\cellcolor{dgreen}&\cellcolor{dred}&\cellcolor{dgreen}&\cellcolor{dgreen}&  \\
     1&\cellcolor{dred}&\cellcolor{dgreen}&\cellcolor{dgreen}&\cellcolor{dgreen}&\cellcolor{dred}&\cellcolor{dgreen}&\cellcolor{dgreen}&\cellcolor{dgreen}&\cellcolor{dred}&\cellcolor{dgreen}&\cellcolor{dgreen}&\cellcolor{dgreen}&\cellcolor{dred}&\cellcolor{dgreen}&\cellcolor{dgreen}&\cellcolor{dgreen} \\
     0&\cellcolor{dgreen}&\cellcolor{dgreen}&\cellcolor{dgreen}&\cellcolor{dgreen}&\cellcolor{dgreen}&\cellcolor{dgreen}&\cellcolor{dgreen}&\cellcolor{dgreen}&\cellcolor{dgreen}&\cellcolor{dgreen}&\cellcolor{dgreen}&\cellcolor{dgreen}&\cellcolor{dgreen}&\cellcolor{dgreen}&\cellcolor{dgreen}&\cellcolor{dgreen} \\ 
    \hhline{~----------------}
    \multicolumn{1}{c}{ }&0&1&2&3&4&5&6&7&8&9&\hspace{-4pt}10\hspace{-4pt}&\hspace{-4pt}11\hspace{-4pt}&\hspace{-4pt}12\hspace{-4pt}&\hspace{-4pt}13\hspace{-4pt}&\hspace{-4pt}14\hspace{-4pt}&\multicolumn{1}{c}{\hspace{-4pt}15\hspace{-4pt}}
  \end{tabular}
}
\\
 {\small $x + y + 4 \le 7$} &
 {\small $\begin{array}{r@{\hspace{1.0ex}}c@{\hspace{1.0ex}}l}
  10x + 11y + 0 & \le & 2x + 3y + 7 \\
  12x + 5y + 7  & \le & 2x + 8y + 15 \\
  0x + 5y + 3   & \le & 8x + 2y + 4 \\
  12x + 7y + 1  & \le & 12x + 3y + 8 
\end{array}$}
&
{\small $\begin{array}{r@{\hspace{1.0ex}}c@{\hspace{1.0ex}}l}
   2x + 15y + 3 & \le & 0x + 15y + 15 \\
  10x + 0y + 4  & \le & 0x + 15y + 15 \\
   8x + 15y + 4 & \le & 8x + 0y + 4
\end{array}$}
\\
{\small (a) {\color{dgreen}$x + y + 4 \le 7$}} & {\small (b) {\color{dgreen}$y = x * x$}} & {\small (c) {\color{dgreen}$y \le x * x$}}
 \end{tabular}
\vspace{-1mm}
\caption{\label{Fi:BitVectorPolyhedraExamples}
Each subfigure illustrates a bit-vector formula (in {\color{dgreen}green}) and the most precise bit-vector polyhedron computed by \name (i.e., the inequalities above the formulas).
Each colored cell in the plots represents a solution in $4$-bit unsigned modular arithmetic of the conjunction of the inequalities found by \name: {\color{dgreen}green} cells represent solutions to the original formula, whereas {\color{dred}red} cells are points that are solution to the inequalities, but do not satisfy the original formula.
In (a) and (b), the conjunctive formula represents the original formula exactly (there are only {\color{dgreen}green} cells).
In (b), the twelve occurrences of {\color{dred}red} cells are points that do not satisfy the original formula, but are needed for a conjunctive formula to over-approximate the original formula.}
\end{figure}


In general, a bit-vector inequality over unsigned bit-vector variables $X = \{ x_1, \ldots, x_n \}$ has the form
$\sum_{i=1}^n p_i x_i + q \leq \sum_{j=1}^n r_j x_j + s$,
where $\{ p_i \} \cup \{ q \} \cup \{ r_j \} \cup \{ s \}$ are unsigned bit-vector constants, and $\leq$ is unsigned comparison.
Conjunctions of inequalities are a fragment of quantifier-free bit-vector logic ($\mathcal{L}_\BV$).
Let $\interp{\varphi}_\BV$ denote the set of assignments to $X$ that satisfy formula $\varphi \in \mathcal{L}_\BV$.

We instantiated \name to take as input a formula $\varphi \in \mathcal{L}_\BV$ and return a conjunction $\psi$ of bit-vector inequalities---i.e., the \emph{symbolic abstraction} \cite[\S5]{DBLP:conf/vmcai/RepsT16} of $\varphi$ in the conjunctive fragment $\mathcal{L}_\BVCONJ$.
Because \name computes best \lproperties, in this setting it computes the most-precise symbolic abstraction---i.e., the formula $\alphaHat$ computed by \name is one representation of the \emph{most-precise abstraction} of $\varphi$ that is expressible as a conjunction of bit-vector inequalities.

As known from the literature (\cite{DBLP:conf/vmcai/RepsSY04,DBLP:conf/cav/ThakurR12,DBLP:conf/sas/ThakurER12} and \cite[\S5]{DBLP:conf/vmcai/RepsT16}), operations needed for abstract interpretation, such as (i) the creation and/or application of abstract transformers, and (ii) taking the join of two abstract-domain elements, can be performed via an algorithm for symbolic abstraction.
For instance, if $\psi_a$ and $\psi_b$ are two formulas in $\mathcal{L}_\BVCONJ$, we can perform the join $\psi_a \sqcup \psi_b$ by $\alphaHat(\psi_a \lor \psi_b)$.
(Note that $\psi_a \lor \psi_b$ is not a formula of $\mathcal{L}_\BV$.)
For this reason, we say that \name \textit{enables} this new abstract domain.

In our experiments, we limited inequalities to two variables $x$ and $y$ on each side, and used $4$-bit unsigned arithmetic.
Our benchmarks were taken from an earlier study conducted by one of the authors, which on each example used brute force to consider all 16,762,320 non-tautologies of the 16,777,216 $4$-bit inequalities of the form $ax + by + c \le dx + ey + f$.
That study found that some example formulas had hundreds of thousands of inequalities as consequences.
We selected $\numBVBenchmarks$ interesting-looking formulas to use as benchmarks, including linear/nonlinear operations, equalities and inequalities, Boolean combinations, and one pair of formulas on which to perform the join operation.
(See \S\ref{se:BitVectorInequalities:q2} and Appendix~\ref{App:BitVectorPolyhedraBenchmarks}.)

\subsubsection{Quantitative Analysis}
\label{se:BitVectorInequalities:q1}

\name[\sketch] computed a sound best \lconjunction for \numBVSolved/\numBVBenchmarks formulas, and guaranteed that \numBVNotimeout/\numBVBenchmarks were best \lconjunctions.
For the query $y = x / 2$, \name[\sketch] timed out on a call to \cp, but the obtained \lconjunction was indeed a best \lconjunction because it defined the exact semantics of the query.
\name computed an \lconjunction for all the \numBVBenchmarks benchmarks in less than 400 seconds per benchmark, which is 2-5 orders of magnitude faster than the enumerative baseline presentd in Section~\ref{se:evalq1}.
%
Each output \lconjunction contained between 1 and 6 \lproperties.
For this domain, \sketch---and hence \name[\sketch]---is sound and precise because we are working with bit-vector arithmetic of fixed bit-width.

\name[\sketch] could not terminate for most of our benchmarks when considering 8-bit arithmetic.
Because our examples contain several multiplications, this limitation is not surprising because multiplication is one of the known weaknesses of \sketch and its underlying SAT solver.
Recently, there have been promising advances in SAT solving for multiplication circuits \cite{beame} that, if integrated with \sketch, we believe would help \name scale to larger bit-vectors.

\subsubsection{Qualitative Analysis}
\label{se:BitVectorInequalities:q2}
Examples of results obtained by \name are shown in Fig.\ \ref{Fi:BitVectorPolyhedraExamples}.
The result that the most-precise abstraction of these formulas could be expressed using only a small number of inequalities was surprising to the authors.
In Fig.\ \ref{Fi:BitVectorPolyhedraExamples}(c), the formula we are abstracting is $\varphi =_{\textit{df}} y \le x * x$).
Fig.\ \ref{Fi:BitVectorPolyhedraExamples}(c) shows that, in addition to the {\color{dgreen}green} points that satisfy the non-linear inequality $y \le x * x$, the bit-vector-polyhedral abstraction found for $\varphi$ includes twelve ``extra'' points, indicated by the {\color{dred}red} cells.
Because \name finds a
\emph{most-precise} sound bit-vector-polyhedral abstraction of $\varphi$, every sound bit-vector-polyhedral abstraction of $\varphi$ must also include those twelve points.
In an earlier study conducted by one of the authors, they used brute force to consider all 16,762,320 non-tautologies of the 16,777,216 $4$-bit inequalities of the form $ax + by + c \le dx + ey + f$.
That study found that the following numbers of inequalities over-approximated the original formula:
564 for Fig.~\ref{Fi:BitVectorPolyhedraExamples}(a), 109,008 for Fig.~\ref{Fi:BitVectorPolyhedraExamples}(b), and 456 for Fig.~\ref{Fi:BitVectorPolyhedraExamples}(c).
Thus, \name showed that a most-precise abstraction could be $2$-$4$ orders of magnitude smaller than one obtained by brute force.

\textbf{Finding:} \name can synthesize the most-precise sound bit-vector-polyhedral abstraction of a given bit-vector formula $\varphi$ over $4$-bit arithmetic.
Furthermore, \name surprised the authors by showing that the most-precise sound bit-vector-polyhedral abstraction for the presented examples could be precisely expressed with only a handful of bit-vector inequalities.

\subsection{Further Analysis of \name's Performance}
\label{se:furtheranalysis}

In the previous sections, we have shown that \name can synthesize best \lconjunctions for a variety of case studies.
In this section, we analyze what parameters affect \name's running time.


\subsubsection*{Q1: How do Different Primitives of the Algorithm Contribute to the Running Time?}
\label{se:evalq2}

On average, \name spends \ratioSynth\% of the time performing \synth, \ratioSound\% performing \cs, and \ratioPrecision\% performing \cp (details in Table~\ref{tab:benchmarks} in App.~\ref{App:data}).

It usually takes longer for \cs and \cp to show the nonexistence of an example---i.e., to return $\bot$---than to find an example.
\cs is one of the simplest queries, but occupies a large portion of the running time because it is expected to return $\bot$ many times, whereas \cp needs to return $\bot$ only once for each call to \synthproperty.
%
The last call to \cp (i.e., the one that returns $\bot$) often takes a significant amount of time to complete (on average \ratioLastCall\% of the time spent on each run of \synthproperty).

\textbf{Finding:} \name spends most of the time checking soundness and precision.

\subsubsection*{Q2: What Parts of the Input Affect the Running Time?}
\label{se:evalq3}
The number of \lproperties in the language $\lang$ has a large impact on the time taken by \synth (Fig.~\ref{fig:time-vs-size}) and \cp. 
%
%
%

The complexity of the code defining the semantics of various operators has a large impact on how long \cs takes.
\bstinsert\xspace, \bstdelete\xspace of BST and edit distance have relatively complicated implementations, and \cs takes longer for these problems.

The size and complexity of the example space also affect the running time. 
The biggest factor contributing to the size of the example space is the number of input and output variables used. 
The number of possible examples, i.e., variable assignments, increases exponentially with the number of variables.
The size of the example space affects not only the number of total queries but also the time that each query takes.
Specifically, the number of positive and negative examples affects the time taken by \synth or \cp (notice that \cs does not take the examples as input), as shown in Figure~\ref{fig:time-vs-example}.
%


\textbf{Finding:}
The running time of \name is affected by the sizes of
(i) the property search space, (ii) the programs that describe the semantics of the operators, and (iii) the example search space.

\begin{figure}[t]
\begin{subfigure}[valign=t]{0.325\linewidth}
\captionsetup{width=.9\linewidth}
\begin{tikzpicture}[every mark/.append style={mark size=1.5pt}]
	\begin{axis}[%
	legend style={nodes={scale=0.5, transform shape},at={(0.62, 0.18)},anchor=west},
	ylabel absolute, ylabel style={yshift=-4mm},
	xlabel absolute, xlabel style={yshift=1mm},
	xmode=log,
	ymode=log,
	xtick ={1e4, 1e8, 1e12, 1e16},
	ytick ={0.01, 1, 1e1, 1e2, 1e3},
	width=0.95\linewidth,
	height=0.95\linewidth,
	scatter/classes={%
		a={mark=square,draw=black,
			style={solid, fill=black},
			mark size=1.5pt},
		b={mark=star,draw=dgreen,
		    style={solid, fill=black},
		    mark size=2pt},
		c={mark=diamond,draw=violet,
			style={solid, fill=black},
			mark size=2pt}},
	xmin = 1,
	xmax = 1e17,
	ymin = 1,
	ymax = 1500,
	xlabel={\footnotesize{Grammar size}},
	ylabel={\footnotesize{Time (s)}}],
	\addplot[scatter,only marks,mark size=1pt,%
	scatter src=explicit symbolic]
	table[meta=label] {
		x y label
28991029248	49.975	a
10077696	18.26	a
14843406974976	284.59	a
10077696	18.67	a
2515456	33.69166667	b
125000	15.21	b
262144	27.395	b
2515456	29.08571429	b
496793088	54.895	b
2985984	25.825	b
2985984	33.38714286	b
474552000	58.67666667	b
4096	21.06	b
8741816	33.8625	b
1906624	30.335	b
238328	17.075	b
1815848	23.935	b
64000	24.9575	b
1728000	21.115	b
2985984	21.03714286	b
1728000	20.38	b
8489664	32.17666667	b
592704	308.3	a
64	17.85	a
592704	292.48	a
125000	13.63	a
1728000	27.98	a
262144	31.74625	a
1728000	25.79	a
216000	44.86125	b
216000	44.28	b
262144	52.362	b
567663552	376.598	b
262144	36.085	b
140608	14.73	b
3375000	6.025	c
3375000	13.168	c
56623104	64.27	c
157464	0.775	c
884736	24.79833333	c
	};
	\legend{Other, \synquid, \sygus}
	\end{axis}
\end{tikzpicture}
\caption{Time to synthesize a best \lproperty vs. grammar size}
\label{fig:time-vs-size}
\end{subfigure}
\begin{subfigure}[valign=t]{0.325\linewidth}
\captionsetup{width=.9\linewidth}
\vspace{2mm}
\begin{tikzpicture}[every mark/.append style={mark size=1.5pt}]
	\begin{axis}[%
	legend style={nodes={scale=0.5, transform shape},at={(0.05, 0.78)},anchor=west},
	ylabel absolute, ylabel style={yshift=-4mm},
	xlabel absolute, xlabel style={yshift=1mm},
	width=0.95\linewidth,
	height=0.95\linewidth,
	scatter/classes={%
		a={mark=*,draw=red,
			mark options={solid},
			style={solid, fill=white},
			mark size=2pt},
		c={mark=triangle,draw=orange,
			mark options={solid},
			style={solid, fill=white},
			mark size=2pt}},
	xlabel={\footnotesize{Examples}},
	ylabel={\footnotesize{\cp (s)}}],
	\addplot[scatter,only marks,mark size=1pt,%
	scatter src=explicit symbolic]
	table[meta=label] {
		x y label
0	1.622852802	a
1	1.693539143	a
4	1.848130941	a
5	1.933784008	a
6	2.002374887	a
7	2.2131598	a
9	2.253421068	a
11	2.302711248	a
12	2.464223862	a
13	2.417783976	a
14	2.656901836	a
14	2.901232004	a
6	2.009962797	a
7	2.229377985	a
8	2.090982914	a
9	2.218266964	a
10	2.319853067	a
11	2.572346687	a
12	2.656720877	a
13	2.751822233	a
14	2.552029133	a
14	2.459115028	a
13	2.808489084	a
14	2.81709981	a
15	2.913047075	a
17	3.249284029	a
18	3.067202806	a
18	3.364037037	a
20	3.716299057	a
22	4.118217945	a
23	4.932993889	a
23	4.23415184	a
24	4.234435081	a
27	4.625747919	a
28	4.624039173	a
31	5.562054873	a
32	5.608935833	a
35	6.820699692	a
35	7.874405146	a
44	7.733325005	a
45	8.969686985	a
46	9.212382078	a
47	9.288352013	a
47	10.24044609	a
44	8.027467012	a
46	8.359483004	a
48	8.914690971	a
49	10.92411017	a
50	10.38342619	a
50	10.14040089	a
47	8.324429035	a
49	8.930686235	a
49	11.17447901	a
48	8.882997036	a
49	9.142050982	a
51	9.806297064	a
53	10.25165081	a
56	10.962605	a
57	13.68678713	a
59	14.13837981	a
62	14.37795305	a
63	15.76964092	a
63	15.98415995	a
65	15.37469816	a
66	14.29191184	a
67	15.63994074	a
70	16.67337322	a
74	17.83886623	a
75	17.60687733	a
76	17.87613201	a
78	19.69236469	a
80	19.13138628	a
82	22.82573271	a
86	22.24416375	a
89	23.22312808	a
90	22.73412085	a
92	26.05378604	a
93	25.78535604	a
97	27.92346811	a
99	27.8820951	a
100	29.56081676	a
101	28.95943189	a
102	29.08925414	a
103	30.78888917	a
104	31.77130222	a
105	34.30780315	a
106	32.99144292	a
107	33.81336808	a
108	32.294065	a
109	34.33150005	a
110	39.03782916	a
111	36.06873512	a
113	40.57947707	a
114	36.60527086	a
115	44.15833497	a
116	36.90641165	a
117	38.19565129	a
118	46.47421694	a
120	39.39947534	a
121	49.42901611	a
122	42.51989412	a
123	74.17347789	a
124	45.86305571	a
125	59.12397909	a
126	47.2665391	a
127	54.56939912	a
128	72.29722524	a
129	55.421381	a
130	77.12007189	a
131	300.0064471	a
0	1.751425028	c
4	2.43398881	c
7	2.427947044	c
8	2.047446012	c
11	3.037776709	c
12	4.65690279	c
13	2.810189009	c
14	8.981709957	c
15	6.923151731	c
16	9.628410816	c
17	5.826683044	c
19	2.549660206	c
20	4.143041134	c
21	4.216393948	c
22	5.276009798	c
23	9.462685347	c
24	7.284902811	c
25	3.788697958	c
26	28.32117915	c
26	9.566752434	c
9	2.745138884	c
10	3.562884092	c
12	3.122756958	c
13	3.499140978	c
14	5.978044748	c
15	2.783586979	c
16	3.487811089	c
17	5.652858019	c
18	9.708377123	c
19	13.73902893	c
19	2.815600157	c
20	4.244663	c
21	3.065221071	c
22	4.383856058	c
23	8.030562162	c
10	2.526991129	c
11	3.826797009	c
12	5.537379265	c
14	10.96054602	c
15	10.83151174	c
16	27.92000699	c
17	72.9649632	c
17	3.300433159	c
18	3.323971987	c
19	3.779036999	c
20	2.492965221	c
21	12.39938188	c
15	45.93947005	c
	};
	\legend{\tbranch, \ianonlinsum}
	\end{axis}
\end{tikzpicture}
\vspace{-1mm}
\caption{\cp time vs. number of examples}
\label{fig:time-vs-example}
\end{subfigure}
\begin{subfigure}[valign=t]{0.325\linewidth}
\captionsetup{width=.9\linewidth}
\vspace{-3mm}
\begin{tikzpicture}[]
  \tikzset{mark options={mark size=1, opacity=0.3}}
  \begin{axis}[
	xmode=log,
	ymode=log,
    xlabel style={yshift=1mm},
    ylabel style={yshift=-2mm},
    height=0.95\linewidth,
    width=0.95\linewidth,
    xlabel= {\footnotesize{Without line~\ref{Li:MergeMayIntoMust} (s)}}, 
    ylabel= {\footnotesize{With line~\ref{Li:MergeMayIntoMust} (s)}},
    xmin = 0,
	xmax = 15000,
	ymin = 0,
	ymax = 15000,
	xtick ={1, 1e1, 1e2, 1e3, 1e4},
	ytick ={1, 1e1, 1e2, 1e3, 1e4},
  ]
  \addplot+[only marks]
table {
152.6	199.9
50.8	36.52
584.8	853.77
67.95	56.01
260.35	202.15
15.79	15.21
67.54	54.79
223.33	203.6
223.43	219.58
258.54	206.6
278.68	233.71
190.56	176.03
41.8	42.12
162.87	135.45
141.59	121.34
36	34.15
55.01	47.87
216.85	199.66
46.85	42.23
159.06	147.26
63.27	61.14
100.01	96.53
409.7	308.3
18.34	17.85
322.7	292.48
13.66	13.63
28.29	27.98
265.62	253.97
30.66	25.79
378.08	358.89
349.03	354.24
274.3	261.81
4134.47	3765.98
66.2	72.17
14.63	14.73
12.81	12.05
4.82	4.8
68.71	65.84
172.7	257.08
1.47	1.55
157.64	148.79
15.76	16.06
3.13	3.12
5.86	5.88
5.7	5.6
55.09	55.26
700.29	704.32
4.16	4.04
11.46	10.36
4.12	4.04
9.77	8.86
3.25	3.17
27.18	22.81
59.47	56.58
28.37	27.84
22.36	20
24.06	21.48
22.73	24.24
17.22	16.39
38.52	32.39
13.91	13.79
14.92	14.52
234.58	202.97
8.63	8.69
38.67	37.5
187.34	164.86
267.9	227.45
77.21	77.53
83.92	78.38
28.29	28.47
30.65	29.01
177.29	176.19
172.38	232.14
353.77	388.57
119.65	220.89
143.97	155.43
218.68	195.89
165.69	238.84
106.2	88.53
249.47	303.47
    };
\addplot[mark=none, red] coordinates {(1,1) (12000,12000)};
\end{axis}
\end{tikzpicture}
\vspace{-1mm}
\caption{Time with/without line~\ref{Li:MergeMayIntoMust}.}
\label{fig:freeze_comparison}
\end{subfigure}

\caption{Evaluation of the running time of \name for different input sizes and optimizations.}
\end{figure} 

\subsubsection*{Q3: How effective is line \ref{Li:MergeMayIntoMust} in Algorithm~\ref{alg:SynthesizeProperty}?}
\label{se:evalq4}

We compared the running times of \name with and without line~\ref{Li:MergeMayIntoMust} (Figure \ref{fig:freeze_comparison}).
%
When line~\ref{Li:MergeMayIntoMust} is present, \name is \ratioFreeze\% faster (geometric mean) than when line~\ref{Li:MergeMayIntoMust} is absent, but both versions can solve the same problems.
The optimization is only effective when the language has many incomparable properties that do not imply each other and cause \synth to often return $\bot$, thus triggering Line~\ref{Li:RevertLastSound} in Algorithm~\ref{alg:SynthesizeProperty}---e.g., in all \sygus benchmarks the language $\lang$ is such that the optimization is not used.

\textbf{Finding:} Freezing negative examples is slightly effective.

\section{Related Work}
\label{se:related-work}


\mypar{Abstract-interpretation techniques}
Many static program-analysis techniques are pitched as tools for checking safety properties, but behind the scenes they construct an artifact that abstracts the behavior of a program (in the sense of abstract interpretation \cite{DBLP:conf/popl/CousotC77}).
One such kind of artifact is a procedure summary \cite{DBLP:conf/popl/CousotH78,Chapter:CC78,Chapter:SP81,DBLP:conf/cav/GopanR07}, which abstracts a procedure's transition relation with an abstract value from an abstract domain, the elements of which denote transition relations.
%

%

Our problem is an instance of the \emph{strongest-consequence problem} \cite{DBLP:conf/vmcai/RepsT16}.
Existing techniques for solving this problem rely on properties of the language $\lang$ that are typical of abstract interpretation.
Some techniques work from ``below,'' identifying a chain of successively weaker implicants, until one is a consequence of $\varphi$ \cite{DBLP:conf/vmcai/RepsSY04}.
Other techniques work from ``above,'' identifying a chain of successively stronger implicates, until no further strengthening is possible \cite{DBLP:conf/cav/ThakurR12,DBLP:conf/sas/ThakurER12}.
\citet{DBLP:conf/vmcai/OzeriPRS17} explored a different approach, which works from above by repeatedly applying a semantic-reduction operation \cite{DBLP:conf/popl/CousotC77}. (A semantic reduction operation finds a less-complicated description of a given set of states if one exists.)
Our work differs from methods that use abstract interpretation in several aspects.
First, our algorithm is the first to use both positive and negative examples to achieve precision.
Second, while our work supports a variety of DSLs specified via a grammar, existing methods require that certain operations can be performed on concrete states and elements of the language $\lang$ (e.g., joins~\cite{DBLP:conf/cav/ThakurR12}),
thus limiting the language that can serve as the DSL.

\mypar{Type inference}
Liquid type inference~\cite{RondonKJ08,Hashimoto2015RefinementTI,vazou2014refinement} can infer a weakest precondition from a given postcondition or a strongest postcondition from a given precondition.
To make the problem tractable, properties must be specified in a user-given restricted set of predicates that are closed under Boolean operations.
Although our work shares some similarities---e.g., looking for properties over a restricted DSL---we tackled a fundamentally different problem because no pre- or post-condition is given as an input, and our algorithm instead looks for best \lproperties.
Furthermore, our work is not restricted to functional languages.

\mypar{Invariant inference}
Several data-driven, CEGIS-style algorithms can synthesize program invariants.
These techniques look for any invariant that is satisfactory for a client verification problem, whereas we do not assume there is a client for whom the properties are synthesized.
Without a client, ``true'' is a sound (and also weakest) but  useless specification.
The absence of a client requires synthesized specifications to be precise, therefore requiring our new \cp primitive.

A closely related system is Elrond \cite{DBLP:journals/pacmpl/ZhouDDJ21}, which synthesizes weakest library specifications that make verification possible in a client program.
Elrond allows one to specify a set of target predicates of interest, and finds quantified Boolean formulas with equalities over the variables and the predicates.
Because soundness (with respect to the library function) is only checked on a set of inputs, Elrond tries to synthesize weakest specifications using an iterative weakening approach.
Their algorithm takes advantage of the structure of supported formulas (they can contain disjunctions), but has some limitations (they can only contain equalities).

The key differences between Elrond and our work are:
\rone
    Our work supports a user-supplied DSL, which enables more generality, but prevents the use of techniques that rely on access to arbitrary Boolean operations.
\rtwo
    Our work uses a parametric DSL that can contain complex user-given functions, whereas Elrond only allows parametric atoms (i.e., user-defined Boolean function), equality over variables, and Boolean combinations of them.
    Our tool \name can synthesize the \lproperty ``$2 \varo = 0 \vee - \x + 2\varo - \x^2 = 0$,'' whereas Elrond does not consider arithmetic predicates.
\rthree
    The specifications generated by \name are most precise with respect to the DSL $\lang$, allowing for their reuse in multiple problem instances that use the same function. Such reuse is possible because \name ensures soundness of the synthesized properties. In contrast, Elrond operates in a closed-box setting and uses a random sampler for soundness, intentionally weakening the synthesized specifications to enhance the likelihood of soundness.  
\rfour
Our work can synthesize arithmetic properties efficiently (e.g., the ones considered in Sections \ref{se:sensitivity} and \ref{se:bitvector}) as well as complex algebraic properties (e.g., the ones in Section \ref{se:algspec}). In theory, Elrond can describe all the necessary components to express algebraic properties such as the property $(\vars, \x) = \stpop(\stpush(\vars, \x))$, but one would need to provide a function that combines $\stpop$ and $\stpush$. This approach is feasible if one is interested in only a few possible ways of combining functions, but becomes infeasible once more combinations are possible (e.g., for an arithmetic formula).
In short, the two approaches have different goals.

The presence of a user-supplied DSL and absence of a client distinguish our work from other prior work, e.g., abductive inference \cite{dillig2012automated}, ICE-learning \cite{DBLP:conf/cav/0001LMN14}, LoopInvGen \cite{DBLP:conf/pldi/PadhiSM16}, Hanoi~\cite{hanoi}, and Data-Driven CHC Solving \cite{DBLP:conf/pldi/ZhuMJ18}.
%



\mypar{Dynamic techniques}
Daikon \cite{DBLP:journals/tse/ErnstCGN01,DBLP:journals/scp/ErnstPGMPTX07} is a system for identifying likely invariants by inspecting program traces.
Invariants generally involve at most two program quantities and are checked at procedure entry and exit points (i.e., invariants form precondition/postcondition pairs).
In Daikon, the default is to check 75 different forms of invariants, instantiated for the program variables of interest.
The language of invariants can be extended by the user.

\name differs from Daikon (and follow-up work, e.g., ~\cite{DBLP:journals/tse/BeckmanNRSTT10}) in two ways: \rone The language $\lang$ is not limited to a set of predicates, and \name scales to languages containing millions of properties; \rtwo the properties that \name synthesizes are sound and provably \textit{best} \lproperties.
Furthermore, while Daikon's dynamic approach can scale to large programs, we could not find a way to encode our case studies as instances that Daikon could receive as input.
%

A similar tool to Daikon is QuickSpec \cite{smallbone2017quick}, which generates equational properties of Haskell programs from random tests.
\name differs from QuickSpec in two ways: \rone The language $\lang$ is not limited to equational properties; \rtwo the properties that \name synthesizes are sound and provably \textit{best} \lproperties.

\citet{DBLP:journals/pacmpl/AstorgaSDWMX21} synthesize contracts that are sound with respect to positive examples generated by a test generator.
We see two main differences between that work and ours:
\rone They do not use negative examples, whereas we do. 
Negative examples are the key to synthesizing best \lproperties.
\rtwo Their work does not allow a parametrized DSL and their notion of ``tight'' is with respect to a syntactic restriction on the logic in which the contract is to be specified.


\mypar{Synthesis of best $\lang$-transformers}
The paper that inspired our work synthesizes most-precise abstract transformers in a user-given DSL \cite{DBLP:journals/pacmpl/KalitaMDRR22}.
We realized that their basic insight---use both positive and negative examples; treat positive examples as hard constraints and negative examples as ``maybe'' constraints---had broader applicability than just creating abstract transformers.


%
Our work differs in two key ways.
First, because our goal is to obtain a formula rather than a piece of code (their setting), we could take advantage of the structure of formulas---in particular, conjunctions---to decompose the problem into \rone an ``inner search'' to find a best \lproperty that is an individual conjunct (Alg.~\ref{alg:SynthesizeProperty}), and \rtwo an ``outer search'' to accumulate best \lproperties to form a best \lconjunction (Alg.~\ref{alg:SynthesizeAllProperties}).
Second, Alg.~\ref{alg:SynthesizeProperty} exploits monotonicity---i.e., once a sound \lproperty is found, there must exist a best \lproperty that implies it (Lemma~\ref{lem:monotonicity}).
This observation
allows us to use a simplified set of primitives:
our algorithm uses a \synth primitive, whereas theirs requires a \maxsynth primitive---i.e.,
one that synthesizes a program that accept all the positive examples and rejects as many negative examples as possible.
%
Our ideas could be back-ported to provide improvements in their setting as well:
if the abstract domain supports meet ($\sqcap$),
they could run their algorithm multiple times to create a kind of ``conjunctive'' transformer, which would run multiple, incomparable best $\lang$-transformers, 
and then take the meet of the results.

\section{Conclusion}
\label{se:conclusion}

This paper presents a formal framework for the problem of synthesizing a best \lconjunction---i.e., a conjunctive specification of a program with respect to a user-defined logic $\lang$---and an algorithm for automatically synthesizing a  best \lconjunction.
The innovations in the algorithm are three-fold:
(i) it identifies individual conjuncts that are themselves strongest consequences;
(ii) it balances negative examples that \emph{must} be rejected by the \lproperty being synthesized and ones that \emph{may} be rejected;
and
(iii) it guarantees progress via monotonic constraint hardening.
%

Our work opens up many avenues for further study.
One is to harness other kinds of synthesis engines to implement \synth, \cs, and \cp. 
Recent work on Semantics-Guided Synthesis (\semgus) \cite{DBLP:journals/pacmpl/KimHDR21} provides an expressive synthesis framework for expressing complex synthesis problems like the ones discussed in this paper. 
\semgus solvers, such as \messy \cite{MessySemGuSTool-2022}, are able to produce two-sided answers to a problem:
either synthesizing a solution, or proving that the problem is unrealizable---i.e., has no solution---exactly what is needed for \cp.
On the theoretical side, while we have used first-order logic, it would be interesting to try other logics, such as separation logic \cite{DBLP:conf/lics/Reynolds02} or effectively propositional logic \cite{Thesis:Itzhaky14,Thesis:Padon18}.
On the practical side, our work could find applications in invariant generation~\cite{DBLP:journals/pacmpl/PadonWKMA22} and code deobfuscation~\cite{203640}.

\subsection*{Acknowledgement}

Supported, in part, by
a Microsoft Faculty Fellowship;
a gift from Rajiv and Ritu Batra;
NSF under grants CCF-\{1750965,1763871,1918211,2023222,2211968,2212558\};
and ONR under grant N00014-17-1-2889.
Any opinions, findings, and conclusions or recommendations expressed in this publication are those of the authors, and do not necessarily reflect the views of the sponsoring entities.

\subsection*{Data-Availability Statement}

We provide a comprehensive Docker image on Zenodo that containing the source code of \name, Dafny proofs, and all necessary dependencies for the experiments~\cite{kanghee20238327699}.


\bibliographystyle{ACM-Reference-Format}
\bibliography{main.bib}


\begin{thebibliography}{48}


\ifx \showCODEN    \undefined \def \showCODEN     #1{\unskip}     \fi
\ifx \showDOI      \undefined \def \showDOI       #1{#1}\fi
\ifx \showISBNx    \undefined \def \showISBNx     #1{\unskip}     \fi
\ifx \showISBNxiii \undefined \def \showISBNxiii  #1{\unskip}     \fi
\ifx \showISSN     \undefined \def \showISSN      #1{\unskip}     \fi
\ifx \showLCCN     \undefined \def \showLCCN      #1{\unskip}     \fi
\ifx \shownote     \undefined \def \shownote      #1{#1}          \fi
\ifx \showarticletitle \undefined \def \showarticletitle #1{#1}   \fi
\ifx \showURL      \undefined \def \showURL       {\relax}        \fi
\providecommand\bibfield[2]{#2}
\providecommand\bibinfo[2]{#2}
\providecommand\natexlab[1]{#1}
\providecommand\showeprint[2][]{arXiv:#2}

\bibitem[Alur et~al\mbox{.}(2019)]%
        {https://doi.org/10.48550/arxiv.1904.07146}
\bibfield{author}{\bibinfo{person}{Rajeev Alur}, \bibinfo{person}{Dana Fisman},
  \bibinfo{person}{Saswat Padhi}, \bibinfo{person}{Rishabh Singh}, {and}
  \bibinfo{person}{Abhishek Udupa}.} \bibinfo{year}{2019}\natexlab{}.
\newblock \bibinfo{title}{SyGuS-Comp 2018: Results and Analysis}.
\newblock
\newblock
\urldef\tempurl%
\url{https://doi.org/10.48550/ARXIV.1904.07146}
\showDOI{\tempurl}


\bibitem[Astorga et~al\mbox{.}(2021)]%
        {DBLP:journals/pacmpl/AstorgaSDWMX21}
\bibfield{author}{\bibinfo{person}{Angello Astorga},
  \bibinfo{person}{Shambwaditya Saha}, \bibinfo{person}{Ahmad Dinkins},
  \bibinfo{person}{Felicia Wang}, \bibinfo{person}{P. Madhusudan}, {and}
  \bibinfo{person}{Tao Xie}.} \bibinfo{year}{2021}\natexlab{}.
\newblock \showarticletitle{Synthesizing contracts correct modulo a test
  generator}.
\newblock \bibinfo{journal}{\emph{Proc. {ACM} Program. Lang.}}
  \bibinfo{volume}{5}, \bibinfo{number}{{OOPSLA}} (\bibinfo{year}{2021}),
  \bibinfo{pages}{1--27}.
\newblock
\urldef\tempurl%
\url{https://doi.org/10.1145/3485481}
\showDOI{\tempurl}


\bibitem[Bagnara et~al\mbox{.}(2008)]%
        {DBLP:journals/scp/BagnaraHZ08}
\bibfield{author}{\bibinfo{person}{Roberto Bagnara},
  \bibinfo{person}{Patricia~M. Hill}, {and} \bibinfo{person}{Enea Zaffanella}.}
  \bibinfo{year}{2008}\natexlab{}.
\newblock \showarticletitle{The Parma Polyhedra Library: Toward a complete set
  of numerical abstractions for the analysis and verification of hardware and
  software systems}.
\newblock \bibinfo{journal}{\emph{Sci. Comput. Program.}} \bibinfo{volume}{72},
  \bibinfo{number}{1-2} (\bibinfo{year}{2008}), \bibinfo{pages}{3--21}.
\newblock
\urldef\tempurl%
\url{https://doi.org/10.1016/j.scico.2007.08.001}
\showDOI{\tempurl}


\bibitem[Barbosa et~al\mbox{.}(2022)]%
        {cvc5}
\bibfield{author}{\bibinfo{person}{Haniel Barbosa}, \bibinfo{person}{Clark~W.
  Barrett}, \bibinfo{person}{Martin Brain}, \bibinfo{person}{Gereon Kremer},
  \bibinfo{person}{Hanna Lachnitt}, \bibinfo{person}{Makai Mann},
  \bibinfo{person}{Abdalrhman Mohamed}, \bibinfo{person}{Mudathir Mohamed},
  \bibinfo{person}{Aina Niemetz}, \bibinfo{person}{Andres N{\"{o}}tzli},
  \bibinfo{person}{Alex Ozdemir}, \bibinfo{person}{Mathias Preiner},
  \bibinfo{person}{Andrew Reynolds}, \bibinfo{person}{Ying Sheng},
  \bibinfo{person}{Cesare Tinelli}, {and} \bibinfo{person}{Yoni Zohar}.}
  \bibinfo{year}{2022}\natexlab{}.
\newblock \showarticletitle{cvc5: {A} Versatile and Industrial-Strength {SMT}
  Solver}. In \bibinfo{booktitle}{\emph{Tools and Algorithms for the
  Construction and Analysis of Systems - 28th International Conference, {TACAS}
  2022, Held as Part of the European Joint Conferences on Theory and Practice
  of Software, {ETAPS} 2022, Munich, Germany, April 2-7, 2022, Proceedings,
  Part {I}}} \emph{(\bibinfo{series}{Lecture Notes in Computer Science},
  Vol.~\bibinfo{volume}{13243})}, \bibfield{editor}{\bibinfo{person}{Dana
  Fisman} {and} \bibinfo{person}{Grigore Rosu}} (Eds.).
  \bibinfo{publisher}{Springer}, \bibinfo{pages}{415--442}.
\newblock
\urldef\tempurl%
\url{https://doi.org/10.1007/978-3-030-99524-9\_24}
\showDOI{\tempurl}


\bibitem[Beckman et~al\mbox{.}(2010)]%
        {DBLP:journals/tse/BeckmanNRSTT10}
\bibfield{author}{\bibinfo{person}{Nels~E. Beckman}, \bibinfo{person}{Aditya~V.
  Nori}, \bibinfo{person}{Sriram~K. Rajamani}, \bibinfo{person}{Robert~J.
  Simmons}, \bibinfo{person}{SaiDeep Tetali}, {and} \bibinfo{person}{Aditya~V.
  Thakur}.} \bibinfo{year}{2010}\natexlab{}.
\newblock \showarticletitle{Proofs from Tests}.
\newblock \bibinfo{journal}{\emph{{IEEE} Trans. Software Eng.}}
  \bibinfo{volume}{36}, \bibinfo{number}{4} (\bibinfo{year}{2010}),
  \bibinfo{pages}{495--508}.
\newblock
\urldef\tempurl%
\url{https://doi.org/10.1109/TSE.2010.49}
\showDOI{\tempurl}


\bibitem[Blazytko et~al\mbox{.}(2017)]%
        {203640}
\bibfield{author}{\bibinfo{person}{Tim Blazytko}, \bibinfo{person}{Moritz
  Contag}, \bibinfo{person}{Cornelius Aschermann}, {and}
  \bibinfo{person}{Thorsten Holz}.} \bibinfo{year}{2017}\natexlab{}.
\newblock \showarticletitle{Syntia: Synthesizing the Semantics of Obfuscated
  Code}. In \bibinfo{booktitle}{\emph{26th USENIX Security Symposium (USENIX
  Security 17)}}. \bibinfo{publisher}{USENIX Association},
  \bibinfo{address}{Vancouver, BC}, \bibinfo{pages}{643--659}.
\newblock
\showISBNx{978-1-931971-40-9}
\urldef\tempurl%
\url{https://www.usenix.org/conference/usenixsecurity17/technical-sessions/presentation/blazytko}
\showURL{%
\tempurl}


\bibitem[Cousot and Cousot(1977)]%
        {DBLP:conf/popl/CousotC77}
\bibfield{author}{\bibinfo{person}{Patrick Cousot} {and}
  \bibinfo{person}{Radhia Cousot}.} \bibinfo{year}{1977}\natexlab{}.
\newblock \showarticletitle{Abstract Interpretation: {A} Unified Lattice Model
  for Static Analysis of Programs by Construction or Approximation of
  Fixpoints}. In \bibinfo{booktitle}{\emph{Conference Record of the Fourth
  {ACM} Symposium on Principles of Programming Languages, Los Angeles,
  California, USA, January 1977}}, \bibfield{editor}{\bibinfo{person}{Robert~M.
  Graham}, \bibinfo{person}{Michael~A. Harrison}, {and} \bibinfo{person}{Ravi
  Sethi}} (Eds.). \bibinfo{publisher}{{ACM}}, \bibinfo{pages}{238--252}.
\newblock
\urldef\tempurl%
\url{https://doi.org/10.1145/512950.512973}
\showDOI{\tempurl}


\bibitem[Cousot and Cousot(1978)]%
        {Chapter:CC78}
\bibfield{author}{\bibinfo{person}{Patrick Cousot} {and}
  \bibinfo{person}{Radhia Cousot}.} \bibinfo{year}{1978}\natexlab{}.
\newblock \showarticletitle{Static Determination of Dynamic Properties of
  Recursive Procedures}.
\newblock In \bibinfo{booktitle}{\emph{Formal Descriptions of Programming
  Concepts}}. \bibinfo{publisher}{North-Holland}.
\newblock


\bibitem[Cousot and Halbwachs(1978)]%
        {DBLP:conf/popl/CousotH78}
\bibfield{author}{\bibinfo{person}{Patrick Cousot} {and}
  \bibinfo{person}{Nicolas Halbwachs}.} \bibinfo{year}{1978}\natexlab{}.
\newblock \showarticletitle{Automatic Discovery of Linear Restraints Among
  Variables of a Program}. In \bibinfo{booktitle}{\emph{Conference Record of
  the Fifth Annual {ACM} Symposium on Principles of Programming Languages,
  Tucson, Arizona, USA, January 1978}},
  \bibfield{editor}{\bibinfo{person}{Alfred~V. Aho},
  \bibinfo{person}{Stephen~N. Zilles}, {and} \bibinfo{person}{Thomas~G.
  Szymanski}} (Eds.). \bibinfo{publisher}{{ACM} Press},
  \bibinfo{pages}{84--96}.
\newblock
\urldef\tempurl%
\url{https://doi.org/10.1145/512760.512770}
\showDOI{\tempurl}


\bibitem[D'Antoni et~al\mbox{.}(2013)]%
        {DAntoniGAHP13}
\bibfield{author}{\bibinfo{person}{Loris D'Antoni}, \bibinfo{person}{Marco
  Gaboardi}, \bibinfo{person}{Emilio Jes{\'{u}}s~Gallego Arias},
  \bibinfo{person}{Andreas Haeberlen}, {and} \bibinfo{person}{Benjamin~C.
  Pierce}.} \bibinfo{year}{2013}\natexlab{}.
\newblock \showarticletitle{Sensitivity analysis using type-based constraints}.
  In \bibinfo{booktitle}{\emph{Proceedings of the 1st annual workshop on
  Functional programming concepts in domain-specific languages, FPCDSL@ICFP
  2013, Boston, Massachusetts, USA, September 22, 2013}},
  \bibfield{editor}{\bibinfo{person}{Richard Lazarus},
  \bibinfo{person}{Assaf~J. Kfoury}, {and} \bibinfo{person}{Jacob Beal}}
  (Eds.). \bibinfo{publisher}{{ACM}}, \bibinfo{pages}{43--50}.
\newblock
\urldef\tempurl%
\url{https://doi.org/10.1145/2505351.2505353}
\showDOI{\tempurl}


\bibitem[Dillig et~al\mbox{.}(2012)]%
        {dillig2012automated}
\bibfield{author}{\bibinfo{person}{Isil Dillig}, \bibinfo{person}{Thomas
  Dillig}, {and} \bibinfo{person}{Alex Aiken}.}
  \bibinfo{year}{2012}\natexlab{}.
\newblock \showarticletitle{Automated error diagnosis using abductive
  inference}.
\newblock \bibinfo{journal}{\emph{ACM SIGPLAN Notices}} \bibinfo{volume}{47},
  \bibinfo{number}{6} (\bibinfo{year}{2012}), \bibinfo{pages}{181--192}.
\newblock


\bibitem[Ernst et~al\mbox{.}(2001)]%
        {DBLP:journals/tse/ErnstCGN01}
\bibfield{author}{\bibinfo{person}{Michael~D. Ernst}, \bibinfo{person}{Jake
  Cockrell}, \bibinfo{person}{William~G. Griswold}, {and}
  \bibinfo{person}{David Notkin}.} \bibinfo{year}{2001}\natexlab{}.
\newblock \showarticletitle{Dynamically Discovering Likely Program Invariants
  to Support Program Evolution}.
\newblock \bibinfo{journal}{\emph{{IEEE} Trans. Software Eng.}}
  \bibinfo{volume}{27}, \bibinfo{number}{2} (\bibinfo{year}{2001}),
  \bibinfo{pages}{99--123}.
\newblock
\urldef\tempurl%
\url{https://doi.org/10.1109/32.908957}
\showDOI{\tempurl}


\bibitem[Ernst et~al\mbox{.}(2007)]%
        {DBLP:journals/scp/ErnstPGMPTX07}
\bibfield{author}{\bibinfo{person}{Michael~D. Ernst}, \bibinfo{person}{Jeff~H.
  Perkins}, \bibinfo{person}{Philip~J. Guo}, \bibinfo{person}{Stephen
  McCamant}, \bibinfo{person}{Carlos Pacheco}, \bibinfo{person}{Matthew~S.
  Tschantz}, {and} \bibinfo{person}{Chen Xiao}.}
  \bibinfo{year}{2007}\natexlab{}.
\newblock \showarticletitle{The Daikon system for dynamic detection of likely
  invariants}.
\newblock \bibinfo{journal}{\emph{Sci. Comput. Program.}} \bibinfo{volume}{69},
  \bibinfo{number}{1-3} (\bibinfo{year}{2007}), \bibinfo{pages}{35--45}.
\newblock
\urldef\tempurl%
\url{https://doi.org/10.1016/j.scico.2007.01.015}
\showDOI{\tempurl}


\bibitem[Garg et~al\mbox{.}(2014)]%
        {DBLP:conf/cav/0001LMN14}
\bibfield{author}{\bibinfo{person}{Pranav Garg}, \bibinfo{person}{Christof
  L{\"{o}}ding}, \bibinfo{person}{P. Madhusudan}, {and} \bibinfo{person}{Daniel
  Neider}.} \bibinfo{year}{2014}\natexlab{}.
\newblock \showarticletitle{{ICE:} {A} Robust Framework for Learning
  Invariants}. In \bibinfo{booktitle}{\emph{Computer Aided Verification - 26th
  International Conference, {CAV} 2014, Held as Part of the Vienna Summer of
  Logic, {VSL} 2014, Vienna, Austria, July 18-22, 2014. Proceedings}}
  \emph{(\bibinfo{series}{Lecture Notes in Computer Science},
  Vol.~\bibinfo{volume}{8559})}, \bibfield{editor}{\bibinfo{person}{Armin
  Biere} {and} \bibinfo{person}{Roderick Bloem}} (Eds.).
  \bibinfo{publisher}{Springer}, \bibinfo{pages}{69--87}.
\newblock
\urldef\tempurl%
\url{https://doi.org/10.1007/978-3-319-08867-9\_5}
\showDOI{\tempurl}


\bibitem[Gopan and Reps(2007)]%
        {DBLP:conf/cav/GopanR07}
\bibfield{author}{\bibinfo{person}{Denis Gopan} {and}
  \bibinfo{person}{Thomas~W. Reps}.} \bibinfo{year}{2007}\natexlab{}.
\newblock \showarticletitle{Low-Level Library Analysis and Summarization}. In
  \bibinfo{booktitle}{\emph{Computer Aided Verification, 19th International
  Conference, {CAV} 2007, Berlin, Germany, July 3-7, 2007, Proceedings}}
  \emph{(\bibinfo{series}{Lecture Notes in Computer Science},
  Vol.~\bibinfo{volume}{4590})}, \bibfield{editor}{\bibinfo{person}{Werner
  Damm} {and} \bibinfo{person}{Holger Hermanns}} (Eds.).
  \bibinfo{publisher}{Springer}, \bibinfo{pages}{68--81}.
\newblock
\urldef\tempurl%
\url{https://doi.org/10.1007/978-3-540-73368-3\_10}
\showDOI{\tempurl}


\bibitem[Hashimoto and Unno(2015)]%
        {Hashimoto2015RefinementTI}
\bibfield{author}{\bibinfo{person}{Kodai Hashimoto} {and}
  \bibinfo{person}{Hiroshi Unno}.} \bibinfo{year}{2015}\natexlab{}.
\newblock \showarticletitle{Refinement Type Inference via Horn Constraint
  Optimization}. In \bibinfo{booktitle}{\emph{SAS}}.
\newblock


\bibitem[Hu et~al\mbox{.}(2020)]%
        {HuCDR20}
\bibfield{author}{\bibinfo{person}{Qinheping Hu}, \bibinfo{person}{John
  Cyphert}, \bibinfo{person}{Loris D'Antoni}, {and} \bibinfo{person}{Thomas~W.
  Reps}.} \bibinfo{year}{2020}\natexlab{}.
\newblock \showarticletitle{Exact and approximate methods for proving
  unrealizability of syntax-guided synthesis problems}. In
  \bibinfo{booktitle}{\emph{Proceedings of the 41st {ACM} {SIGPLAN}
  International Conference on Programming Language Design and Implementation,
  {PLDI} 2020, London, UK, June 15-20, 2020}},
  \bibfield{editor}{\bibinfo{person}{Alastair~F. Donaldson} {and}
  \bibinfo{person}{Emina Torlak}} (Eds.). \bibinfo{publisher}{{ACM}},
  \bibinfo{pages}{1128--1142}.
\newblock
\urldef\tempurl%
\url{https://doi.org/10.1145/3385412.3385979}
\showDOI{\tempurl}


\bibitem[Itzhaky(2014)]%
        {Thesis:Itzhaky14}
\bibfield{author}{\bibinfo{person}{Shachar Itzhaky}.}
  \bibinfo{year}{2014}\natexlab{}.
\newblock \emph{\bibinfo{title}{Automatic Reasoning for Pointer Programs Using
  Decidable Logics}}.
\newblock \bibinfo{thesistype}{Ph.\,D. Dissertation}. \bibinfo{school}{Tel Aviv
  University}.
\newblock


\bibitem[Kalita et~al\mbox{.}(2022)]%
        {DBLP:journals/pacmpl/KalitaMDRR22}
\bibfield{author}{\bibinfo{person}{Pankaj~Kumar Kalita},
  \bibinfo{person}{Sujit~Kumar Muduli}, \bibinfo{person}{Loris D'Antoni},
  \bibinfo{person}{Thomas~W. Reps}, {and} \bibinfo{person}{Subhajit Roy}.}
  \bibinfo{year}{2022}\natexlab{}.
\newblock \showarticletitle{Synthesizing Abstract Transformers}.
\newblock \bibinfo{journal}{\emph{Proc. {ACM} Program. Lang.}}
  \bibinfo{number}{{OOPSLA}} (\bibinfo{year}{2022}).
\newblock
\urldef\tempurl%
\url{https://doi.org/10.1145/3563334}
\showDOI{\tempurl}


\bibitem[Kaufmann et~al\mbox{.}(2022)]%
        {beame}
\bibfield{author}{\bibinfo{person}{Daniela Kaufmann}, \bibinfo{person}{Paul
  Beame}, \bibinfo{person}{Armin Biere}, {and} \bibinfo{person}{Jakob
  Nordstrom}.} \bibinfo{year}{2022}\natexlab{}.
\newblock \showarticletitle{Adding Dual Variables to Algebraic Reasoning for
  Gate-Level Multiplier Verification}. In \bibinfo{booktitle}{\emph{Proceedings
  of the 2022 Conference on Design, Automation Test in Europe}} (Antwerp,
  Belgium) \emph{(\bibinfo{series}{DATE '22})}. \bibinfo{publisher}{European
  Design and Automation Association}, \bibinfo{address}{Leuven, BEL},
  \bibinfo{pages}{1431–1436}.
\newblock
\showISBNx{9783981926361}


\bibitem[Kim(2022)]%
        {MessySemGuSTool-2022}
\bibfield{author}{\bibinfo{person}{Jinwoo Kim}.}
  \bibinfo{year}{2022}\natexlab{}.
\newblock \bibinfo{title}{Messy-Release}.
\newblock
  \bibinfo{howpublished}{\url{https://github.com/kjw227/Messy-Release}}.
\newblock


\bibitem[Kim et~al\mbox{.}(2021)]%
        {DBLP:journals/pacmpl/KimHDR21}
\bibfield{author}{\bibinfo{person}{Jinwoo Kim}, \bibinfo{person}{Qinheping Hu},
  \bibinfo{person}{Loris D'Antoni}, {and} \bibinfo{person}{Thomas~W. Reps}.}
  \bibinfo{year}{2021}\natexlab{}.
\newblock \showarticletitle{Semantics-guided synthesis}.
\newblock \bibinfo{journal}{\emph{Proc. {ACM} Program. Lang.}}
  \bibinfo{volume}{5}, \bibinfo{number}{{POPL}} (\bibinfo{year}{2021}),
  \bibinfo{pages}{1--32}.
\newblock
\urldef\tempurl%
\url{https://doi.org/10.1145/3434311}
\showDOI{\tempurl}


\bibitem[Leino and W{\"{u}}stholz(2014)]%
        {LeinoW14}
\bibfield{author}{\bibinfo{person}{K.~Rustan~M. Leino} {and}
  \bibinfo{person}{Valentin W{\"{u}}stholz}.} \bibinfo{year}{2014}\natexlab{}.
\newblock \showarticletitle{The Dafny Integrated Development Environment}. In
  \bibinfo{booktitle}{\emph{Proceedings 1st Workshop on Formal Integrated
  Development Environment, {F-IDE} 2014, Grenoble, France, April 6, 2014}}
  \emph{(\bibinfo{series}{{EPTCS}}, Vol.~\bibinfo{volume}{149})},
  \bibfield{editor}{\bibinfo{person}{Catherine Dubois},
  \bibinfo{person}{Dimitra Giannakopoulou}, {and} \bibinfo{person}{Dominique
  M{\'{e}}ry}} (Eds.). \bibinfo{pages}{3--15}.
\newblock
\urldef\tempurl%
\url{https://doi.org/10.4204/EPTCS.149.2}
\showDOI{\tempurl}


\bibitem[Lo et~al\mbox{.}(2017)]%
        {BOOK:MSSSA17}
\bibfield{editor}{\bibinfo{person}{David Lo}, \bibinfo{person}{Siau-Cheng
  Khoo}, \bibinfo{person}{Jiawei Han}, {and} \bibinfo{person}{Chao Liu}}
  (Eds.). \bibinfo{year}{2017}\natexlab{}.
\newblock \bibinfo{booktitle}{\emph{Mining Software Specifications:
  {M}ethodologies and Applications}}.
\newblock \bibinfo{publisher}{{Chapman \& Hall}}.
\newblock


\bibitem[Mariano et~al\mbox{.}(2019)]%
        {mariano2019algspecs}
\bibfield{author}{\bibinfo{person}{Benjamin Mariano}, \bibinfo{person}{Josh
  Reese}, \bibinfo{person}{Siyuan Xu}, \bibinfo{person}{ThanhVu Nguyen},
  \bibinfo{person}{Xiaokang Qiu}, \bibinfo{person}{Jeffrey~S. Foster}, {and}
  \bibinfo{person}{Armando Solar-Lezama}.} \bibinfo{year}{2019}\natexlab{}.
\newblock \showarticletitle{Program Synthesis with Algebraic Library
  Specifications}.
\newblock \bibinfo{journal}{\emph{Proc. ACM Program. Lang.}}
  \bibinfo{volume}{3}, \bibinfo{number}{OOPSLA}, Article
  \bibinfo{articleno}{132} (\bibinfo{date}{Oct.} \bibinfo{year}{2019}),
  \bibinfo{numpages}{25}~pages.
\newblock
\urldef\tempurl%
\url{https://doi.org/10.1145/3360558}
\showDOI{\tempurl}


\bibitem[Miltner et~al\mbox{.}(2020)]%
        {hanoi}
\bibfield{author}{\bibinfo{person}{Anders Miltner}, \bibinfo{person}{Saswat
  Padhi}, \bibinfo{person}{Todd Millstein}, {and} \bibinfo{person}{David
  Walker}.} \bibinfo{year}{2020}\natexlab{}.
\newblock \showarticletitle{Data-Driven Inference of Representation
  Invariants}. In \bibinfo{booktitle}{\emph{Proceedings of the 41st ACM SIGPLAN
  Conference on Programming Language Design and Implementation}} (London, UK)
  \emph{(\bibinfo{series}{PLDI 2020})}. \bibinfo{publisher}{Association for
  Computing Machinery}, \bibinfo{address}{New York, NY, USA},
  \bibinfo{pages}{1–15}.
\newblock
\showISBNx{9781450376136}
\urldef\tempurl%
\url{https://doi.org/10.1145/3385412.3385967}
\showDOI{\tempurl}


\bibitem[Mitchell(1997)]%
        {Book:Mitchell97}
\bibfield{author}{\bibinfo{person}{Tom~M. Mitchell}.}
  \bibinfo{year}{1997}\natexlab{}.
\newblock \bibinfo{booktitle}{\emph{Machine Learning}}.
\newblock \bibinfo{publisher}{McGraw-Hill}.
\newblock


\bibitem[Ozeri et~al\mbox{.}(2017)]%
        {DBLP:conf/vmcai/OzeriPRS17}
\bibfield{author}{\bibinfo{person}{Or Ozeri}, \bibinfo{person}{Oded Padon},
  \bibinfo{person}{Noam Rinetzky}, {and} \bibinfo{person}{Mooly Sagiv}.}
  \bibinfo{year}{2017}\natexlab{}.
\newblock \showarticletitle{Conjunctive Abstract Interpretation Using
  Paramodulation}. In \bibinfo{booktitle}{\emph{Verification, Model Checking,
  and Abstract Interpretation - 18th International Conference, {VMCAI} 2017,
  Paris, France, January 15-17, 2017, Proceedings}}
  \emph{(\bibinfo{series}{Lecture Notes in Computer Science},
  Vol.~\bibinfo{volume}{10145})}, \bibfield{editor}{\bibinfo{person}{Ahmed
  Bouajjani} {and} \bibinfo{person}{David Monniaux}} (Eds.).
  \bibinfo{publisher}{Springer}, \bibinfo{pages}{442--461}.
\newblock
\urldef\tempurl%
\url{https://doi.org/10.1007/978-3-319-52234-0\_24}
\showDOI{\tempurl}


\bibitem[Padhi et~al\mbox{.}(2016)]%
        {DBLP:conf/pldi/PadhiSM16}
\bibfield{author}{\bibinfo{person}{Saswat Padhi}, \bibinfo{person}{Rahul
  Sharma}, {and} \bibinfo{person}{Todd~D. Millstein}.}
  \bibinfo{year}{2016}\natexlab{}.
\newblock \showarticletitle{Data-driven precondition inference with learned
  features}. In \bibinfo{booktitle}{\emph{Proceedings of the 37th {ACM}
  {SIGPLAN} Conference on Programming Language Design and Implementation,
  {PLDI} 2016, Santa Barbara, CA, USA, June 13-17, 2016}},
  \bibfield{editor}{\bibinfo{person}{Chandra Krintz} {and}
  \bibinfo{person}{Emery~D. Berger}} (Eds.). \bibinfo{publisher}{{ACM}},
  \bibinfo{pages}{42--56}.
\newblock
\urldef\tempurl%
\url{https://doi.org/10.1145/2908080.2908099}
\showDOI{\tempurl}


\bibitem[Padon(2018)]%
        {Thesis:Padon18}
\bibfield{author}{\bibinfo{person}{Oded Padon}.}
  \bibinfo{year}{2018}\natexlab{}.
\newblock \emph{\bibinfo{title}{Deductive Verification of Distributed Protocols
  in First-Order Logic}}.
\newblock \bibinfo{thesistype}{Ph.\,D. Dissertation}. \bibinfo{school}{Tel Aviv
  University}.
\newblock


\bibitem[Padon et~al\mbox{.}(2022)]%
        {DBLP:journals/pacmpl/PadonWKMA22}
\bibfield{author}{\bibinfo{person}{Oded Padon}, \bibinfo{person}{James~R.
  Wilcox}, \bibinfo{person}{Jason~R. Koenig}, \bibinfo{person}{Kenneth~L.
  McMillan}, {and} \bibinfo{person}{Alex Aiken}.}
  \bibinfo{year}{2022}\natexlab{}.
\newblock \showarticletitle{Induction duality: primal-dual search for
  invariants}.
\newblock \bibinfo{journal}{\emph{Proc. {ACM} Program. Lang.}}
  \bibinfo{volume}{6}, \bibinfo{number}{{POPL}} (\bibinfo{year}{2022}),
  \bibinfo{pages}{1--29}.
\newblock
\urldef\tempurl%
\url{https://doi.org/10.1145/3498712}
\showDOI{\tempurl}


\bibitem[Park et~al\mbox{.}(2023a)]%
        {park2023synthesizing}
\bibfield{author}{\bibinfo{person}{Kanghee Park}, \bibinfo{person}{Loris
  D'Antoni}, {and} \bibinfo{person}{Thomas Reps}.}
  \bibinfo{year}{2023}\natexlab{a}.
\newblock \bibinfo{title}{Synthesizing Specifications}.
\newblock
\newblock
\showeprint[arxiv]{2301.11117}~[cs.PL]


\bibitem[Park et~al\mbox{.}(2023b)]%
        {kanghee20238327699}
\bibfield{author}{\bibinfo{person}{Kanghee Park}, \bibinfo{person}{Loris
  D'Antoni}, {and} \bibinfo{person}{Thomas Reps}.}
  \bibinfo{year}{2023}\natexlab{b}.
\newblock \bibinfo{title}{Synthesizing Specifications}.
\newblock
\newblock
\urldef\tempurl%
\url{https://doi.org/10.5281/zenodo.8327699}
\showDOI{\tempurl}


\bibitem[Polikarpova et~al\mbox{.}(2016)]%
        {polikarpova2016program}
\bibfield{author}{\bibinfo{person}{Nadia Polikarpova}, \bibinfo{person}{Ivan
  Kuraj}, {and} \bibinfo{person}{Armando Solar-Lezama}.}
  \bibinfo{year}{2016}\natexlab{}.
\newblock \showarticletitle{Program synthesis from polymorphic refinement
  types}.
\newblock \bibinfo{journal}{\emph{ACM SIGPLAN Notices}} \bibinfo{volume}{51},
  \bibinfo{number}{6} (\bibinfo{year}{2016}), \bibinfo{pages}{522--538}.
\newblock


\bibitem[Reps et~al\mbox{.}(2004)]%
        {DBLP:conf/vmcai/RepsSY04}
\bibfield{author}{\bibinfo{person}{Thomas~W. Reps}, \bibinfo{person}{Shmuel
  Sagiv}, {and} \bibinfo{person}{Greta Yorsh}.}
  \bibinfo{year}{2004}\natexlab{}.
\newblock \showarticletitle{Symbolic Implementation of the Best Transformer}.
  In \bibinfo{booktitle}{\emph{Verification, Model Checking, and Abstract
  Interpretation, 5th International Conference, {VMCAI} 2004, Venice, Italy,
  January 11-13, 2004, Proceedings}}. \bibinfo{pages}{252--266}.
\newblock
\urldef\tempurl%
\url{https://doi.org/10.1007/978-3-540-24622-0\_21}
\showDOI{\tempurl}


\bibitem[Reps and Thakur(2016)]%
        {DBLP:conf/vmcai/RepsT16}
\bibfield{author}{\bibinfo{person}{Thomas~W. Reps} {and}
  \bibinfo{person}{Aditya~V. Thakur}.} \bibinfo{year}{2016}\natexlab{}.
\newblock \showarticletitle{Automating Abstract Interpretation}. In
  \bibinfo{booktitle}{\emph{Verification, Model Checking, and Abstract
  Interpretation - 17th International Conference, {VMCAI} 2016, St. Petersburg,
  FL, USA, January 17-19, 2016. Proceedings}} \emph{(\bibinfo{series}{Lecture
  Notes in Computer Science}, Vol.~\bibinfo{volume}{9583})},
  \bibfield{editor}{\bibinfo{person}{Barbara Jobstmann} {and}
  \bibinfo{person}{K.~Rustan~M. Leino}} (Eds.). \bibinfo{publisher}{Springer},
  \bibinfo{pages}{3--40}.
\newblock
\urldef\tempurl%
\url{https://doi.org/10.1007/978-3-662-49122-5\_1}
\showDOI{\tempurl}


\bibitem[Reynolds(2002)]%
        {DBLP:conf/lics/Reynolds02}
\bibfield{author}{\bibinfo{person}{John~C. Reynolds}.}
  \bibinfo{year}{2002}\natexlab{}.
\newblock \showarticletitle{Separation Logic: {A} Logic for Shared Mutable Data
  Structures}. In \bibinfo{booktitle}{\emph{17th {IEEE} Symposium on Logic in
  Computer Science {(LICS} 2002), 22-25 July 2002, Copenhagen, Denmark,
  Proceedings}}. \bibinfo{publisher}{{IEEE} Computer Society},
  \bibinfo{pages}{55--74}.
\newblock
\urldef\tempurl%
\url{https://doi.org/10.1109/LICS.2002.1029817}
\showDOI{\tempurl}


\bibitem[Rondon et~al\mbox{.}(2008)]%
        {RondonKJ08}
\bibfield{author}{\bibinfo{person}{Patrick~Maxim Rondon}, \bibinfo{person}{Ming
  Kawaguchi}, {and} \bibinfo{person}{Ranjit Jhala}.}
  \bibinfo{year}{2008}\natexlab{}.
\newblock \showarticletitle{Liquid types}. In
  \bibinfo{booktitle}{\emph{Proceedings of the {ACM} {SIGPLAN} 2008 Conference
  on Programming Language Design and Implementation, Tucson, AZ, USA, June
  7-13, 2008}}, \bibfield{editor}{\bibinfo{person}{Rajiv Gupta} {and}
  \bibinfo{person}{Saman~P. Amarasinghe}} (Eds.). \bibinfo{publisher}{{ACM}},
  \bibinfo{pages}{159--169}.
\newblock
\urldef\tempurl%
\url{https://doi.org/10.1145/1375581.1375602}
\showDOI{\tempurl}


\bibitem[Sharir and Pnueli(1981)]%
        {Chapter:SP81}
\bibfield{author}{\bibinfo{person}{Micha Sharir} {and} \bibinfo{person}{Amir
  Pnueli}.} \bibinfo{year}{1981}\natexlab{}.
\newblock \showarticletitle{Two Approaches to Interprocedural Data Flow
  Analysis}.
\newblock In \bibinfo{booktitle}{\emph{Program Flow Analysis: {T}heory and
  Applications}}. \bibinfo{publisher}{Prentice-Hall}.
\newblock


\bibitem[Smallbone et~al\mbox{.}(2017)]%
        {smallbone2017quick}
\bibfield{author}{\bibinfo{person}{Nicholas Smallbone}, \bibinfo{person}{Moa
  Johansson}, \bibinfo{person}{Koen Claessen}, {and}
  \bibinfo{person}{Maximilian Algehed}.} \bibinfo{year}{2017}\natexlab{}.
\newblock \showarticletitle{Quick specifications for the busy programmer}.
\newblock \bibinfo{journal}{\emph{Journal of Functional Programming}}
  \bibinfo{volume}{27} (\bibinfo{year}{2017}), \bibinfo{pages}{e18}.
\newblock


\bibitem[Solar{-}Lezama(2013)]%
        {DBLP:journals/sttt/Solar-Lezama13}
\bibfield{author}{\bibinfo{person}{Armando Solar{-}Lezama}.}
  \bibinfo{year}{2013}\natexlab{}.
\newblock \showarticletitle{Program sketching}.
\newblock \bibinfo{journal}{\emph{Int. J. Softw. Tools Technol. Transf.}}
  \bibinfo{volume}{15}, \bibinfo{number}{5-6} (\bibinfo{year}{2013}),
  \bibinfo{pages}{475--495}.
\newblock
\urldef\tempurl%
\url{https://doi.org/10.1007/s10009-012-0249-7}
\showDOI{\tempurl}


\bibitem[Thakur et~al\mbox{.}(2012)]%
        {DBLP:conf/sas/ThakurER12}
\bibfield{author}{\bibinfo{person}{Aditya~V. Thakur}, \bibinfo{person}{Matt
  Elder}, {and} \bibinfo{person}{Thomas~W. Reps}.}
  \bibinfo{year}{2012}\natexlab{}.
\newblock \showarticletitle{Bilateral Algorithms for Symbolic Abstraction}. In
  \bibinfo{booktitle}{\emph{Static Analysis - 19th International Symposium,
  {SAS} 2012, Deauville, France, September 11-13, 2012. Proceedings}}.
  \bibinfo{pages}{111--128}.
\newblock
\urldef\tempurl%
\url{https://doi.org/10.1007/978-3-642-33125-1\_10}
\showDOI{\tempurl}


\bibitem[Thakur and Reps(2012)]%
        {DBLP:conf/cav/ThakurR12}
\bibfield{author}{\bibinfo{person}{Aditya~V. Thakur} {and}
  \bibinfo{person}{Thomas~W. Reps}.} \bibinfo{year}{2012}\natexlab{}.
\newblock \showarticletitle{A Method for Symbolic Computation of Abstract
  Operations}. In \bibinfo{booktitle}{\emph{Computer Aided Verification - 24th
  International Conference, {CAV} 2012, Berkeley, CA, USA, July 7-13, 2012
  Proceedings}}. \bibinfo{pages}{174--192}.
\newblock
\urldef\tempurl%
\url{https://doi.org/10.1007/978-3-642-31424-7\_17}
\showDOI{\tempurl}


\bibitem[Vazou et~al\mbox{.}(2014)]%
        {vazou2014refinement}
\bibfield{author}{\bibinfo{person}{Niki Vazou}, \bibinfo{person}{Eric~L
  Seidel}, \bibinfo{person}{Ranjit Jhala}, \bibinfo{person}{Dimitrios
  Vytiniotis}, {and} \bibinfo{person}{Simon Peyton-Jones}.}
  \bibinfo{year}{2014}\natexlab{}.
\newblock \showarticletitle{Refinement types for Haskell}. In
  \bibinfo{booktitle}{\emph{Proceedings of the 19th ACM SIGPLAN international
  conference on Functional programming}}. \bibinfo{pages}{269--282}.
\newblock


\bibitem[Wang et~al\mbox{.}(2016)]%
        {dphamming}
\bibfield{author}{\bibinfo{person}{Yu-Xiang Wang}, \bibinfo{person}{Jing Lei},
  {and} \bibinfo{person}{Stephen~E. Fienberg}.}
  \bibinfo{year}{2016}\natexlab{}.
\newblock \showarticletitle{Learning with Differential Privacy: Stability,
  Learnability and the Sufficiency and Necessity of ERM Principle}.
\newblock \bibinfo{journal}{\emph{J. Mach. Learn. Res.}} \bibinfo{volume}{17},
  \bibinfo{number}{1} (\bibinfo{date}{jan} \bibinfo{year}{2016}),
  \bibinfo{pages}{6353–6392}.
\newblock
\showISSN{1532-4435}


\bibitem[Yao et~al\mbox{.}(2021)]%
        {DBLP:journals/pacmpl/YaoSHZ21}
\bibfield{author}{\bibinfo{person}{Peisen Yao}, \bibinfo{person}{Qingkai Shi},
  \bibinfo{person}{Heqing Huang}, {and} \bibinfo{person}{Charles Zhang}.}
  \bibinfo{year}{2021}\natexlab{}.
\newblock \showarticletitle{Program analysis via efficient symbolic
  abstraction}.
\newblock \bibinfo{journal}{\emph{Proc. {ACM} Program. Lang.}}
  \bibinfo{volume}{5}, \bibinfo{number}{{OOPSLA}} (\bibinfo{year}{2021}),
  \bibinfo{pages}{1--32}.
\newblock
\urldef\tempurl%
\url{https://doi.org/10.1145/3485495}
\showDOI{\tempurl}


\bibitem[Zhou et~al\mbox{.}(2021)]%
        {DBLP:journals/pacmpl/ZhouDDJ21}
\bibfield{author}{\bibinfo{person}{Zhe Zhou}, \bibinfo{person}{Robert
  Dickerson}, \bibinfo{person}{Benjamin Delaware}, {and}
  \bibinfo{person}{Suresh Jagannathan}.} \bibinfo{year}{2021}\natexlab{}.
\newblock \showarticletitle{Data-driven abductive inference of library
  specifications}.
\newblock \bibinfo{journal}{\emph{Proc. {ACM} Program. Lang.}}
  \bibinfo{volume}{5}, \bibinfo{number}{{OOPSLA}} (\bibinfo{year}{2021}),
  \bibinfo{pages}{1--29}.
\newblock
\urldef\tempurl%
\url{https://doi.org/10.1145/3485493}
\showDOI{\tempurl}


\bibitem[Zhu et~al\mbox{.}(2018)]%
        {DBLP:conf/pldi/ZhuMJ18}
\bibfield{author}{\bibinfo{person}{He Zhu}, \bibinfo{person}{Stephen Magill},
  {and} \bibinfo{person}{Suresh Jagannathan}.} \bibinfo{year}{2018}\natexlab{}.
\newblock \showarticletitle{A data-driven {CHC} solver}. In
  \bibinfo{booktitle}{\emph{Proceedings of the 39th {ACM} {SIGPLAN} Conference
  on Programming Language Design and Implementation, {PLDI} 2018, Philadelphia,
  PA, USA, June 18-22, 2018}}, \bibfield{editor}{\bibinfo{person}{Jeffrey~S.
  Foster} {and} \bibinfo{person}{Dan Grossman}} (Eds.).
  \bibinfo{publisher}{{ACM}}, \bibinfo{pages}{707--721}.
\newblock
\urldef\tempurl%
\url{https://doi.org/10.1145/3192366.3192416}
\showDOI{\tempurl}


\end{thebibliography}

\newpage
\appendix
\section{Proofs}
\label{App:Proofs}

\bestconjunction*
\begin{proof}
We first show the $\subseteq$ direction. The interpretation of the conjunction of all possible best properties must include $\phiand$ because $\phiand$ is semantically minimal---i.e., for every best \lproperty $\phil \in \mathcal{P}(\signatures)$ we have $\interp{\phiand} \subseteq \interp{\phil}$.
The $\supseteq$ direction holds because each conjunct of $\phiand$ is a best \lproperty. 
\end{proof}

\monotonicity*
\begin{proof}
For any $\varphi\in\lang$, because $\Rightarrow$ is a well-quasi order on the set of \lproperties, the set of distinct properties $\Pi=\{\varphi'\mid \varphi' \Rightarrow\varphi \wedge \varphi\neq \varphi'\}$ is finite. 
Then, either $\varphi$ is a best \lproperty or at least one of the properties in $\Pi$ is.
\end{proof}

\soundprecise*
\begin{proof}
The fact that $\phil$ accepts all the examples in $\eplus$ and rejects all the examples in $\eminusmay \cup \eminusmust$ follows from Invariant 1. The condition for $\phiinit$ guarantees that Invariant 1 holds on
the first iteration, and the definitions of \synth and \cp guarantee that Invariant 1 holds whenever $\phil$ is updated.
The \lproperty $\phil$ at line \ref{Li:SynthPropReturn} must be sound and precise with respect to $\psi$ because both \cs and \cp have returned $\bot$.
Returning $\bot$ from $\cs(\phil, \phiprog)$ (line~\ref{Li:CallCheckSoundness}) guarantees that $\phil$ is sound.
Returning $\bot$ from $\cp(\phil, \psi, \eplus, \eminusmust)$ (line~\ref{Li:CallCheckPrecision}) guarantees that $\phil$ is precise for $\phiprog$ with respect to $\psi$.
\end{proof}

\soundnessconj*
\begin{proof}
Notice that $\Pi$ is a best \lconjunction if Invariant 4 holds and $\Pi$ is semantically minimal at line \ref{Li:SynthPropertiesReturn}.
We can give an inductive argument to show that Invariant 4 always holds.
At the beginning, Invariant 4 holds trivially just after line \ref{Li:InitEminusMayEmpty}.
For the inductive step, assuming that Invariant 4 holds at line \ref{Li:CallSynthPropertyTwo} of \synthproperties, we show that it still holds after $\Pi \gets \Pi \cup \{ \phil \}$ in line \ref{Li:UpdatePhi}.
First, Lemma \ref{lem:soundprecise} guarantees that each $\phil$ returned from \synthproperty at line \ref{Li:CallSynthPropertyTwo} of \synthproperties is a best \lproperty for $\signatures$.
Also, $\eminusmust$ must be non-empty for the control to reach line \ref{Li:CallSynthPropertyTwo} of \synthproperties.
The examples in $\eminusmust$ must have been discovered by the call to \synthproperty in line \ref{Li:CallSynthPropertyOne}, where the  domain from which examples can be drawn is $\phiand$ (the second argument of the call).
Thus, $\phil$ rejects all of the negative examples in $\eminusmust$ while $\phiand$ accepts them, which establishes that $\phil$ is not weaker than each $\phil' \in \Pi$.
Consequently, Invariant 4 holds for $\Pi \cup \{ \phil \}$.

We now prove that $\Pi$ is semantically minimal at line \ref{Li:SynthPropertiesReturn}---i.e., when \synthproperties returns, there is no
best (most importantly, no sound) \lproperty $\phil'$ such that $\interp{\phiand \wedge \phil'} \subset \interp{\phiand}$.
For \synthproperties to return, \textsc{IsSat($\phiand \wedge \neg \phil$)} at line \ref{Li:CheckForImprovement} must return $\bot$, which would imply that $\interp{\phiand} = \interp{\phiand \wedge \phil}$.
If $\Pi$ is not semantically minimal, then there would exist a sound \lproperty $\phil'$ such that $\interp{\phiand \wedge \phil'} \subset \interp{\phiand} = \interp{\phiand \wedge \phil}$, but this is not possible because the property $\phil$  returned by \synthproperty at line~\ref{Li:CallSynthPropertyOne} is precise with respect to $\phiand$ (Lemma \ref{lem:soundprecise}).
\end{proof}

\completeness*
\begin{proof}
Intuitively, the theorem states that no infinite sequence of sound (resp.\ unsound) properties exists, thus guaranteeing that in finitely many steps \synthproperty will find a best \lproperty and \synthproperties will find a best \lconjunction.
For instance, because each conjunct of a best \lconjunction must be incomparable via $\Rightarrow$, conjuncts of an infinite best \lconjunction would have to form an infinite anti-chain of \lproperties, which contradicts the assumption that $\Rightarrow$ is a well-quasi order on sound \lproperties.

A key property of \synthproperty is that any property stronger than a property that fails \cs at line \ref{Li:CallCheckSoundness} is never considered again, because positive examples in $\eplus$ are never removed.
Consequently, the sequence of unsound \lproperties in an execution of \synthproperty is non-strengthening.
Because $\Leftarrow$ is a well-quasi order on the set of unsound \lproperties, any non-strengthening sequence of unsound \lproperties can only be finite.

We can now move to the case where an \lproperty $\phil$ reaches \cp at line \ref{Li:CallCheckPrecision}.
In this case, we have two possibilities (illustrated in Fig.~\ref{fig:precision}): 
\rone \cp returns an unsound property $\phil''$ (and a negative example $\negex_1$);
\rtwo \cp returns a sound property $\phil''$ (and a negative example $\negex_2$).
Case \rone can only happen finitely many times because only finitely many unsound \lproperties can be found during an execution of \synthproperty.
For case \rtwo, \cs will return $\bot$ in the next iteration, and the negative example $\negex_2$ will be added to $\eminusmust$ (from $\eminusmay$); therefore, any property weaker than $\phil$ will never be considered during this execution of \synthproperty.
Therefore, the sequence $\phil_1=\phil'', \phil_2, \ldots, \phil_n$ of sound \lproperties obtained by \cp is non-weakening.
Because $\Rightarrow$ is a well-quasi order on the set of sound \lproperties, any non-weakening sequence of sound \lproperties must be finite.

The algorithm is also complete if we remove line \ref{Li:MergeMayIntoMust}, but in this case the proof is a little more involved.
Again, we need to worry about case \rtwo.
Because $\phil''$ is sound, in the next iteration, \synthproperty will call \cp again on $\phil''$.
If $\cp$ returns another property (i.e., $\phil''$ is not precise), it may look like the algorithm can get stuck in a loop.
However, that cannot happen for the following reason.
Let $\phil_1=\phil'', \phil_2, \ldots, \phil_n$ be a sequence of \lproperties obtained by consecutive calls to \cp, each of which has returned a precision example $\negex_1,\ldots,\negex_{n}$.
Note that for every set of examples $\{\negex_1, \ldots, \negex_{i}\}$, every $\phil_j$ (for $j> i$) rejects those examples, while every $\phil_j$ (for $j\leq i$) does not. 
Because $\Rightarrow$ is a well-quasi order on the set of sound \lproperties, such a sequence of \lproperties obtained by consecutive calls to \cp must be finite.
The sequence can be terminated by either finding an unsound property (see point \rone for that case), or by \cp returning $\bot$, in which case the algorithm terminates.

Turning to \synthproperties, each best \lproperty is a sound \lproperty;
thus, if there were an infinite sequence of best \lproperties, there would be an infinite sequence of sound \lproperties, violating the assumption that $\Rightarrow$ is a well-quasi order on the set of sound \lproperties.
Consequently, because $\Pi$ in \synthproperties is a set of incomparable sound \lproperties (Invariant 4), the loop on lines \ref{Li:BeginWhileLoop}--\ref{Li:UpdateConjunction} cannot repeat indefinitely.
Moreover, the fact that each call on \synthproperty terminates guarantees that \synthproperties always terminates.
\end{proof}

\section{The Encoding of our Primitives in \sketch}
\label{app:sketch}

We discuss below how \name uses the features of \sketch to implement the DSL $\lang$ specified by the given grammar, and the operations \synth, \cs, and \cp.

The grammar of the DSL $\lang$ is compiled to a \sketch generator.
A generator is a function with holes that allows one to build complex programs via recursion.
Intuitively, holes are used to allow the synthesizer to make choices about what to output.
In our case, holes are used to select which productions are expanded at each point in a derivation tree, and recursion is used to move down the derivation tree.
\sketch allows one to provide a recursion bound to make synthesis tractable, and \name's grammar syntax allows one to provide such bounds for each of the nonterminals in the grammar.
Fig.~\ref{fig:DSLencoding} shows how the generator \texttt{genAP} for the productions of nonterminal $\nonap$ in Equation~\ref{eq:grammar-rev2} can be specified. 

\lstset{%
    language=C,
    basewidth=0.5em,
    xleftmargin=5mm,
    commentstyle=\color{dgreen}\ttfamily,
    basicstyle=\footnotesize\ttfamily,%
    numbers=left, numbersep=5pt,%
    emph={%
    assert, harness, void, if, return, generator, int, list, boolean, minimize%
    },emphstyle={\color{blue}}%
    }%

\begin{figure}
\begin{subfigure}{0.9\textwidth}
\begin{lstlisting}
generator boolean genAP(list l1, list l2, list ol1, list ol2) {
    int t = ??;
    list var_L_1 = genL(l1, l2, ol1, ol2);
    if (t == 0) return isEmpty(var_L_1);
    if (t == 1) return !isEmpty(var_L_1);
    ...
    int var_S_1 = genS(l1, l2, ol1, ol2);
    int var_S_2 = genS(l1, l2, ol1, ol2);
    return {|var_S_1 (< | <= | > | >= | == | !=) var_S_2 + ??(1)|};
}
\end{lstlisting}
\end{subfigure}
\caption{
\sketch generator \texttt{genAP} for the productions of nonterminal $\nonap$ in the DSL in Equation~\ref{eq:grammar-rev2}. Lines 3, 7, and 8 are calls to generators of other nonterminals, and lines 4, 5, and 11 mimic productions in the grammar. 
A production is selected based on the value synthesized by \sketch for \texttt{t} at line 2. (\sketch can synthesize a different value for each recursive call.)
}
\label{fig:DSLencoding}
\end{figure}

\sketch uses assertions (which are often written within so-called harnesses) to impose constraints on the programs produced by a generator. 
In our setting, assertions can be used to produce properties that accept all positive examples and reject all must negative examples.
    
% ...

\synth is implemented using a call to \sketch that finds an \lproperty produced by the generator that satisfies all the assertions.
%

\begin{figure}
\begin{subfigure}{0.45\textwidth}
\begin{lstlisting}[numbers=left,basicstyle=\footnotesize\ttfamily,basewidth=0.5em]
harness void checkPrecision() {
    // find a negative example in psi 
    list lnegex = list_constructor();
    boolean out;
    psi(lnegex, ..., out); 
    assert out;
    // lnegex accepted by current property
    phi(lnegex, ..., out);
    assert out;
    // lnegex not accepted by new property
    phiprime(lnegex, ..., out);
    assert !out;
}
\end{lstlisting}
\end{subfigure}
\quad
\begin{subfigure}{0.41\textwidth}
\begin{lstlisting}
harness void pos_example_1() {
    ... // example description
    boolean out;
    phiprime(lposex1, ..., out);
    assert out;
}
\end{lstlisting}

\begin{lstlisting}
harness void neg_example_1() {
    ... // example description
    boolean out;
    phiprime(lnegex1, ..., out);
    assert !out;
}
\end{lstlisting}
\end{subfigure}
\caption{
\sketch code for \cp with positive and negative examples. \texttt{lnegex} is the negative example $\negex$ generated by the generator \texttt{list\_constructor} that would witness imprecision, \texttt{phi} is a function describing the current property $\phil$, \texttt{psi} is 
a function describing the predicate $\psi$,
and \texttt{phiprime} is a function calling the generator \texttt{genD} to find an \lproperty $\phil'$ that together with \texttt{lnegex} would form a witness for imprecision.}
\label{fig:precisionencoding}
\end{figure}

We also use generators to write data-structure constructors, which are functions that can generate data structures that satisfy desired invariants---e.g., \texttt{list\_constructor()} can generate any valid list up to a certain length.
We bound the size of the inputs in our experiments.

Given a logical specification of a concrete function, we use \sketch as a satisfiability solver.
For a given candidate \lproperty, \cs uses \sketch to attempt to synthesize a soundness counterexample.

\cp asks \sketch to find a property that satisfies all the assertions, together with a new negative example that the newly synthesized property can reject and the previous property accepts (see Figure~\ref{fig:precisionencoding}).

\noindent\textit{Term minimization:}
  To avoid generating unnecessarily complex formulas, we use the minimization feature of \sketch to look for terms of smallest size (both for \synth and \cp).
  Because we are only interested in the final formula, it is enough to apply minimization only in the second call to \synthproperty in Alg.\ \ref{alg:SynthesizeAllProperties} (line \ref{Li:CallSynthPropertyTwo}).

\section{Evaluation Details}
\label{App:data}

In this section, we first describe the benchmarks in detail and then present detailed evaluation metrics.

\subsection{Specification-Mining Benchmarks}
\label{App:SpecificationMiningBenchmarks}

\subsubsection{SyGuS}
We ``inverted'' the roles from the original benchmarks.
We selected a few representative reference implementations of functions from the \sygus repository, and applied \name to synthesize specifications for them.
In this case, we have a ``ground-truth'' property against which we can compare, namely, the specification that a \sygus solver would start from to synthesize the function.
%
%
We bounded the size of integers to 5 bits.

\mypar{Max} 
\maxtwo, \maxthree, and \maxfour\xspace are problems where the goal is to find a function that computes the maximum of 2, 3, or 4 integers, respectively.
For each problem, we provided the reference implementation given in the \sygus repository, and synthesized properties expressed in the grammar 
\begin{equation}
\begin{array}{rcl}
    \nonprop  & :{=} & \top \mid \nonap \mid \nonap \vee \nonap \mid \nonap \vee \nonap \vee \nonap \mid \ldots \\
    \nonap & :{=} & \nonint ~\{ {\leq} \mid {<} \mid {\geq} \mid {>} \mid {=} \mid {\neq}  \}~ \nonint\\
    \nonint  & :{=} & \xone \mid \xtwo \mid \ldots \mid \varo
\end{array}
\label{eq:max-grammar}
\end{equation}
The nonterminal $\nonint$ derives all input and output variables appearing in the
given query.
The grammar for \maxtwo\xspace and \maxthree\xspace allows up to three disjuncts in $\nonprop$, the grammar for \maxfour\xspace allows four disjuncts, and the grammar for \maxfive\xspace allows five disjuncts.

\mypar{Diff} $\diff$ is a \sygus problem where the goal is to find a commutative function that returns the difference of two input variables---i.e. $\varo {=} \x {-} \y \vee \varo {=} \y {-} \x$ whenever $\varo {=} \diff(\x, \y)$.

First, the $\diffone$ benchmark synthesizes consequences of the query $\varo {=} \diff(\x, \y)$ using the following grammar:
\begin{equation}
\begin{array}{rcl}
    \nonprop  & :{=} & \top \mid \nonap \mid \nonap \vee \nonap \mid \nonap \vee \nonap \vee \nonap \\
    \nonap & :{=} & \nonint ~\{ {\leq} \mid {<} \mid {\geq} \mid {>} \mid {=} \mid {\neq}  \}~ \nonint {+} \nonint\\
    \nonint  & :{=} & \x \mid \y \mid \varo \mid 0
\end{array}
\label{eq:diff1-grammar}
\end{equation}
The nonterminal $\nonint$ derives all input and output variables appearing in the
given query,
and also the constant 0, which appears in the implementation of the function.

Second, the $\difftwo$ benchmark synthesizes consequences of the query $\oone {=} \diff(\xone, \yone) \land \otwo {=} \diff(\xtwo, \ytwo)$, using the following grammar:
\begin{equation}
\begin{array}{rcl}
    \nonprop  & :{=} & \top \mid \nonap \mid \nonap \vee \nonap \mid \nonap \vee \nonap \vee \nonap \\
    \nonap & :{=} & \nonint ~\{ {=} \mid {\neq}  \}~ \nonint\\
    \nonint  & :{=} & \xone \mid \yone \mid \oone \mid \xtwo \mid \ytwo \mid \otwo \mid 0
\end{array}
\label{eq:diff2-grammar}
\end{equation}
Using a query in which a function symbol appears twice allows one to express properties such as commutativity.

\mypar{Array}
$\arraytwo$ and $\arraythree$ are \sygus problems where the goal is to synthesize an expression that finds the index of an element in a sorted tuple of size 2 and 3, respectively. 
For example, $\varo = \arraythree(\xone, \xtwo, \xthree, \vark)$ implies that $\varo$ is the index of the value $\vark$ in the sorted tuple $(\xone, \xtwo, \xthree)$.
For each problem, we synthesize properties expressed in the following grammar:
\begin{equation}
\begin{array}{rcl}
    \nonprop  & :{=} & \top \mid \nonap \mid \nonap \vee \nonap \mid \nonap \vee \nonap \vee \nonap \mid \ldots \\
    \nonap & :{=} & \nonint ~\{ {\leq} \mid {<} \mid {\geq} \mid {>} \mid {=} \mid {\neq}  \}~ \nonint \mid 
        \nonsize ~\{ {\leq} \mid {<} \mid {\geq} \mid {>} \mid {=} \mid {\neq}  \}~ \nonsize\\
    \nonint  & :{=} & \xone \mid \xtwo \mid \ldots \mid \vark\\
    \nonsize  & :{=} & 0 \mid 1 \mid 2 \mid \ldots \mid \varo
\end{array}
\label{eq:array-grammar}
\end{equation}
The nonterminal $\nonint$ derives all input variables, 
while the nonterminal $\nonsize$ derives the output and all index constants appearing in the implementation.
The grammar for $\arraytwo$ (resp. $\arraythree$) allows up to three (resp. four) disjuncts.

\subsubsection{Synquid}

We collected all the problems from the \synquid paper \cite{polikarpova2016program} that involve Lists, Binary Trees, and Binary Search Trees (BSTs), and that do not use other datatypes or higher-order functions. 
In these cases, we have a ``ground-truth'' property provided by the polymorphic refinement type used to synthesize the function.
The reference code that we used for each problem is a version of the code synthesized by \synquid, rewritten in a C-style language.
We modified \lelemidx\xspace to return -1 instead of raising an exception when there is no matched element in the list, and added a new problem \ldelete\xspace that deletes only the first occurrence of an element.

\mypar{List}
We considered 13 different methods for Lists constructed by $\exname{Nil}$ and $\exname{Cons}$ (we bounded the length of lists generated by list constructors to 10):
$\lappend$, $\ldelete$, $\ldeleteall$, $\ldrop$, $\lelem$, $\lelemidx$, $\lith$, $\lmin$, $\lreplicate$, $\lreverse$, $\lsnoc$, $\lstutter$, and $\ltake$.
Given the reference implementation from the \synquid benchmark, we synthesized properties in the grammar:
\begin{equation}
\begin{array}{rcl}
    \nonprop  & :{=} & \top \mid \nonap \mid \nonap \vee \nonap \mid \nonap \vee \nonap \vee \nonap \\
    \nonap & :{=} & \isempty(\nonlist) \mid \neg \isempty(\nonlist) \mid \equal(\nonlist, \nonlist) \mid \neg \equal(\nonlist, \nonlist)\\
        &   & \mid \nonsize ~\{ {\leq} \mid {<} \mid {\geq} \mid {>} \mid {=} \mid {\neq}  \}~ \nonsize {+} \ldots {+} \nonsize {+} \{0 \mid 1\}\\
        &   & \mid \{ \forallm \mid \existsm \} \x {\in} \nonlist.~ \x ~\{ {\leq} \mid {<} \mid {\geq} \mid {>} \mid {=} \mid {\neq}  \}~ \nonint \\
        &   & \mid \inputvars \mid \outputvars \\
    \nonint  & :{=} & \inputvars \mid \outputvars \\
    \nonsize  & :{=} & 0 \mid \size(\nonlist) \mid \inputvars \mid \outputvars \mid \progconsts \\
    \nonlist & :{=} & \inputvars \mid \outputvars \\
\end{array}
\label{eq:list-grammar}
\end{equation}

The nonterminal $\nonint$ derives all input and output variables denoting list elements, 
the nonterminal $\nonsize$ derives all integer input and output variables and constants, and 
the nonterminal $\nonlist$ derives all list input and output variables in the implementation.
%
%
%
Size comparisons are used only when the output variable is a list, index, or size. 
The number of occurrences of the nonterminal $\nonsize$ on the right-hand side of a comparison is equal to the number of input variables of type list, index, or size.

In addition to the above 13 problems, 
$\lreversetwice$ synthesizes properties for the query, $\loneout {=} \lreverse(\lone) \land \ltwoout {=} \lreverse(\ltwo)$,
using the following grammar:\footnote{Notice that in the example described 
  in Section~\ref{se:problem-definition} we used a more complex grammar (Eq.~\eqref{eq:grammar-rev2}) that encompasses both the grammars above and can directly synthesize all the properties that are true for a single invocation of \lreverse\xspace in the first grammar as well as for a double invocation.
  We choose to use the grammar in Eq.~\eqref{eq:reverse-twice-grammar} because we are interested in detecting properties that relate two invocations of \lreverse.
}
\begin{equation}
\begin{array}{rcl}
    \nonprop  & :{=} & \top \mid \nonap \mid \nonap \vee \nonap \mid \nonap \vee \nonap \vee \nonap \\
    \nonap & :{=} & \isempty(\nonlist) \mid \neg \isempty(\nonlist)  \mid \equal(\nonlist, \nonlist) \mid \neg \equal(\nonlist, \nonlist)\\
    \nonlist & :{=} & \lone \mid \ltwo \mid \loneout \mid \ltwoout \\
\end{array}
\label{eq:reverse-twice-grammar}
\end{equation}

\mypar{Binary Tree}
The first set of Binary Tree problems consists of the method $\telem$ and the constructors $\tempty$ and $\tbranch$.
We provided the reference implementations from the \synquid benchmark, and synthesized properties expressed in the following grammar:
\begin{equation}
\begin{array}{rcl}
    \nonprop  & :{=} & \top \mid \nonap \mid \nonap \vee \nonap \mid \nonap \vee \nonap \vee \nonap \\
    \nonap & :{=} & \isempty(\nontree) \mid \neg \isempty(\nontree) \mid \equal(\nontree, \nontree) \mid \neg \equal(\nontree, \nontree)\\
        &   & \mid \nonsize ~\{ {\leq} \mid {<} \mid {\geq} \mid {>} \mid {=} \mid {\neq}  \}~ \nonsize {+} \ldots {+} \nonsize {+} \{0 \mid 1\}\\
        &   & \mid \{ \forallm \mid \existsm \} x {\in} \nontree.~ x ~\{ {\leq} \mid {<} \mid {\geq} \mid {>} \mid {=} \mid {\neq}  \}~ \nonint \\
        &   & \mid \inputvars \mid \outputvars \\
    \nonint  & :{=} & \inputvars \mid \outputvars \\
    \nonsize  & :{=} & 0 \mid \size(\nontree) \mid \inputvars \mid \outputvars \mid \progconsts \\
    \nontree & :{=} & \inputvars \mid \outputvars \\
\end{array}
\label{eq:tree-grammar}
\end{equation}

The nonterminal $\nonint$ derives all input and output variables for tree elements, 
the nonterminal $\nonsize$ derives all integer input and output variables and constants, and the nonterminal $\nontree$ derives all tree input and output variables that appear in the implementation.
%
%

In the second set of Binary Tree problems we synthesize properties that capture relationships between the constructor $\tbranch$ and the three methods $\trootval$, $\tleft$, and $\tright$, which return the respective elements of $\tbranch$.
Properties are expressed in the following grammar:
\begin{equation}
\begin{array}{rcl}
    \nonprop  & :{=} & \top \mid \nonap \mid \nonap \vee \nonap \mid \nonap \vee \nonap \vee \nonap \\
    \nonap & :{=} & \isempty(\nontree) \mid \neg \isempty(\nontree) \mid \equal(\nontree, \nontree) \mid \neg \equal(\nontree, \nontree) \mid \nonint ~\{ {=} \mid {\neq}  \}~ \nonint\\
    \nonint  & :{=} & \inputvars \mid \outputvars \\
    \nontree & :{=} & \inputvars \mid \outputvars \\
\end{array}
\label{eq:tree-grammar-2}
\end{equation}
The nonterminal $\nonint$ derives all input and output variables for tree elements, 
and the nonterminal $\nontree$ derives all input and output variables of type Tree.

\mypar{Binary Search Tree}
We considered four different methods for BSTs: $\bstempty$, $\bstinsert$, $\bstdelete$, and $\bstfind$.
We provided the reference implementations from the \synquid benchmarks, and synthesized properties expressed in the grammar in Eq.~\eqref{eq:tree-grammar}.
%
We provided a data-structure constructor that can generate all valid BSTs up to size 4 using the following operations ($n$ stands for an integer):
\begin{equation}
\begin{array}{rcl}
    BST  & :{=} & \bstempty() \mid \bstinsert(BST, \varn)\\
\end{array}
\label{eq:bst-constructor}
\end{equation}

\subsubsection{Other}

We designed {\numOtherBenchmarks} problems to cover missing categories.  
\rone Stack and Queue are data structures implemented using other data structures (i.e., Lists), \rtwo Integer Arithmetic cover more integer arithmetic problems, and in particular, cases where the input program uses loops.

\mypar{Stack}
We considered three Stack methods: $\stempty$, $\stpush$ and $\stpop$ and  synthesized properties in the following grammar:
\begin{equation}
\begin{array}{rcl}
    \nonprop  & :{=} & \top \mid \nonap \mid \nonap \vee \nonap \mid \nonap \vee \nonap \vee \nonap \\
    \nonap & :{=} & \isempty(\nonstack) \mid \neg \isempty(\nonstack) \mid \equal(\nonstack, \nonstack) \mid \neg \equal(\nonstack, \nonstack)\\
        &   & \mid \nonsize ~\{ {\leq} \mid {<} \mid {\geq} \mid {>} \mid {=} \mid {\neq}  \}~ \nonsize {+} \{0 \mid 1\}\\
    \nonsize  & :{=} & 0 \mid \size(\nonstack) \mid \inputvars \mid \outputvars \\
    \nonstack & :{=} & \inputvars \mid \outputvars \\
\end{array}
\label{eq:stack-grammar}
\end{equation}
The nonterminal $\nonsize$ derives all input and output variables for indices or size,
and the nonterminal $\nonstack$ derives all input and output variables of Stack type.

The data-structure constructor generates all stacks up to size 10 using the following operations:
\begin{equation}
\begin{array}{rcl}
    ST  & :{=} & \stempty() \mid \stpush(ST, \varn) \\
\end{array}
\label{eq:stack-constructor}
\end{equation}

Finally, the problem $\stpushpop$ synthesizes properties for the query $\soneout {=} \stpush(\sone, \xone) \land (\stwoout, \xtwo) {=} \stpop(\stwo)$ using the following grammar:
\begin{equation}
\begin{array}{rcl}
    \nonprop  & :{=} & \top \mid \nonap \mid \nonap \vee \nonap \mid \nonap \vee \nonap \vee \nonap \\
    \nonap & :{=} & \nonint ~\{ {=} \mid {\neq}  \}~ \nonint \mid \isempty(\nonstack) \mid \neg \isempty(\nonstack)  \mid \equal(\nonstack, \nonstack) \mid \neg \equal(\nonstack, \nonstack)\\
    \nonint & :{=} & \xone \mid \xtwo \\
    \nonlist & :{=} & \sone \mid \stwo \mid \soneout \mid \stwoout \\
\end{array}
\label{eq:stack-grammar-2}
\end{equation}

\mypar{Queue} 
We considered three Queue methods: $\qempty$, $\qenqueue$ and $\qdequeue$.
%
%
We defined a functional-style Queue using two Lists, and provided implementations of \qempty, \qenqueue\xspace, and \qdequeue\xspace using List methods.
We also provided a function $\qtolist$ that converts a Queue to a single List in which elements of Queue are in the intended order, and synthesized properties expressed in the following grammar:
\begin{equation}
\begin{array}{rcl}
    \nonprop  & :{=} & \top \mid \nonap \mid \nonap \vee \nonap \mid \nonap \vee \nonap \vee \nonap \\
    \nonap & :{=} & \isempty(\nonlist) \mid \neg \isempty(\nonlist)  \mid \equal(\nonlist, \nonlist) \mid \neg \equal(\nonlist, \nonlist)\\
    \nonlist & :{=} & \qtolist(\nonqueue) \mid \lcons(\nonint, \qtolist(\nonqueue)) \mid \lsnoc(\qtolist(\nonqueue), \nonint) \\
    \nonint & :{=} & \inputvars \mid \outputvars \\
    \nonqueue & :{=} & \inputvars \mid \outputvars \\
\end{array}
\label{queue-grammar}
\end{equation}
The nonterminal $\nonint$ derives all input and output variables for Queue elements,
and the nonterminal $\nonqueue$ derives all input and output variables of type Queue.
We provided a data-structure constructor that can generate all valid queues up to size 4 using the following operations:
\begin{equation}
\begin{array}{rcl}
    Q  & :{=} & \qempty() \mid \qenqueue(Q, \varn) \mid \qdequeue(Q) \\
\end{array}
\label{eq:queue-constructor}
\end{equation}

\mypar{Integer Arithmetic} 
We considered three additional integer-manipulating functions:
\rone $\iaabs(\x)$ computes the absolute value of the integer $\x$ (the semantics is expressible in SMT-Lib);
\rtwo $\ialinsum(\x)$ outputs the value of $\x$ using a loop that counts up to $\x$; 
\rthree $\ianonlinsum(\x)$ computes the sum of all integers from $0$ to the value of $\x$ using a loop.

First, we synthesize properties expressed in the following grammar where the nonterminal $\nonint$ derives all input and output variables, constants that appear in the code, and their negations:
\begin{equation}
\begin{array}{rcl}
    \nonprop  & :{=} & \top \mid \nonap \mid \nonap \vee \nonap \mid \nonap \vee \nonap \vee \nonap \\
    \nonap & :{=} & \nonint ~\{ {\leq} \mid {<} \mid {\geq} \mid {>} \mid {=} \mid {\neq}  \}~ \nonint\\
    \nonint  & :{=} & \x \mid {-} \x \mid \varo \mid {-} \varo \mid \progconsts \\
\end{array}
\label{eq:arith-grammar-1}
\end{equation}

Then we synthesize properties in the following grammar, which can express more general properties involving linear combinations of the integer terms of interest:
\begin{equation}
\begin{array}{rcl}
    \nonprop  & :{=} & \top \mid \nonap \mid \nonap \vee \nonap \mid \nonap \vee \nonap \vee \nonap \\
    \nonap & :{=} & \noncoeff \tone {+} \cdots {+}  \noncoeff \tn ~\{ {\leq} \mid {<} \mid {\geq} \mid {>} \mid {=} \mid {\neq}  \}~ 0\\
    \noncoeff  & :{=} & -3 \mid -2 \mid \ldots \mid 3 \mid 4 \\
\end{array}
\label{eq:arith-grammar-2}
\end{equation}
The terms from $\tone$ to $\tn$ are the input and output variables and constants that appear in the implementation.
In the case of $\ianonlinsum$, the quadratic term $\x \cdot \x$ was also included.

Finally, for the case of $\iaabs(\x)$, we also considered the following grammar, which was introduced in Section~\ref{subsub:monotonicity} to illustrate a case in which the given language forms a well-quasi-order but is not finite.
\begin{equation}
\begin{array}{rcl}
    \nonap & :{=} & -20\leq \x \leq 10 \Rightarrow \varo \leq \noncoeff \\
    \noncoeff  & :{=} & 0\mid 1+\noncoeff \\
\end{array}
\label{eq:arith-grammar-3}
\end{equation}

\subsection{Algebraic-Specification Benchmarks}
\label{App:AlgebraicSpecificationBenchmarks}

For each module, we considered queries that were relevant to synthesize the algebraic properties described in the \jlibsketch benchmarks.
We considered fairly small grammars as the module's implementations contained many lines of code and were otherwise hard to perform synthesis for.
For example, the implementation of \texttt{HashMap} contained 150 lines of code and involved arrays.

\subsubsection{ArrayList}
For the query
$\valout = \arrayget(\arrayadd(\arr, \varval), \vark)$, we considered the following grammar:
\begin{equation}
\begin{array}{rcl}
    \nonprop  & :{=} & \tru \mid \nonguard \Rightarrow \valout = \nonint \\
    \nonguard & :{=} & \tru \mid \nonsize ~\{ {\leq} \mid {<} \mid {\geq} \mid {>} \mid {=} \mid {\neq}  \}~ \nonsize \\
    \nonlist & :{=} & \arrayempty \mid \arr \\
    \nonsize & :{=} & \size(\nonlist) \mid \size(\nonlist) {+} 1 \mid \vark \mid \vark {+} 1 \mid 0 \mid 1 \\
    \nonint & :{=} & \varval \mid \arrayget(\arr, \vark) \\
\end{array}
\label{eq:array-grammar-1}
\end{equation}
In this case, \name synthesized the algebraic properties
\begin{equation}
\begin{array}{c}
    \size(\arr) = \vark \Rightarrow \valout = \varval \\
    \size(\arr) \neq \vark \Rightarrow \valout = \arrayget(\arr, \vark)
\end{array}
\label{eq:array-prop-1}
\end{equation}

For the query $\varn = \size(\arrayempty)$, we considered the following grammar:
\begin{equation}
\begin{array}{rcl}
    \nonprop  & :{=} & \tru \mid \varn = \nonsize \\
    \nonsize & :{=} & 0 \mid 1 \\
\end{array}
\label{eq:array-grammar-2}
\end{equation}
In this case, \name synthesized the algebraic property
\begin{equation}
\begin{array}{c}
    \varn = 0
\end{array}
\label{eq:array-prop-2}
\end{equation}

For the query $\varn = \size(\arrayadd(\arr, \varval))$, we considered the following grammar:
\begin{equation}
\begin{array}{rcl}
    \nonprop  & :{=} & \tru \mid \varn = \nonsize \\
    \nonsize & :{=} & \setsize(\nonlist) \mid \setsize(\nonlist) + 1 \mid 0 \mid 1 \\
    \nonlist & :{=} & \arrayempty \mid \arr \\
\end{array}
\label{eq:array-grammar-3}
\end{equation}
In this case, \name synthesized the algebraic property
\begin{equation}
\begin{array}{c}
    \varn = \size(\arr) + 1
\end{array}
\label{eq:array-prop-3}
\end{equation}

\subsubsection{HashMap}
We implemented the \texttt{KeyNotFound} exception of $\mapget$ using an additional output flag $\error$---i.e.,
$\mapget(\varm, \vark)$ refers the output value when the error flag is not set to true, and $\error$ is true when any of the functions that called has thrown a \texttt{KeyNotFound} exception. 
As mentioned in Section~\ref{se:algspec}, the initial mock implementation for HashMap provided in \jlibsketch was incorrect.
In the rest of this section, we assume the correct mock implementation we fixed before synthesizing properties.

For the query $\valout = \mapget(\mapempty, \vark)$, we considered the following grammar:
\begin{equation}
\begin{array}{rcl}
    \nonprop  & :{=} & \tru \mid \error
\end{array}
\label{eq:map-grammar-1}
\end{equation}
In this case, \name synthesized the algebraic property
\begin{equation}
\begin{array}{c}
    \error
\end{array}
\label{eq:map-prop-1}
\end{equation}

For the query $\valout = \mapget(\mapput(\varm, \kone, \varval), \ktwo)$, we considered the following grammar:
\begin{equation}
\begin{array}{rcl}
    \nonprop  & :{=} & \tru \mid \error \mid \valout = \nonvalue \\
    \nonguard & :{=} & \tru \mid \nonkey ~\{ {=} \mid {\neq}  \} ~\nonkey \\
    \nonmap & :{=} & \mapempty \mid \varm \\
    \nonkey & :{=} & \kone \mid \ktwo \\
    \nonvalue & :{=} & \varval \mid \mapget(\nonmap, \nonkey)  \\
\end{array}
\label{eq:map-grammar-2}
\end{equation}
In this case, \name synthesized the algebraic properties
\begin{equation}
\begin{array}{c}
    \kone = \ktwo \Rightarrow \valout = \varval \\
    \kone \neq \ktwo \Rightarrow \valout = \mapget(\varm, \ktwo)
\end{array}
\label{eq:map-prop-2}
\end{equation}

For the query $\mout = \mapput(\mapput(\varm, \kone, \valone), \ktwo, \valtwo)$, we considered the following grammar:
\begin{equation}
\begin{array}{rcl}
    \nonprop  & :{=} & \tru \mid \nonguard \Rightarrow \equal(\mout, \nonmap) \\
    \nonguard & :{=} & \tru \mid \nonkey ~\{ {=} \mid {\neq}  \} ~\nonkey \mid \nonvalue ~\{ {=} \mid {\neq}  \} ~\nonvalue\\
    \nonmap & :{=} & \mapempty \mid \varm \mid \mapput(\varm, \nonkey, \nonvalue) \\
    \nonkey & :{=} & \kone \mid \ktwo \\
    \nonvalue & :{=} & \valone \mid \valtwo  \\
\end{array}
\label{eq:map-grammar-3}
\end{equation}
In this case, \name synthesized the algebraic property
\begin{equation}
\begin{array}{c}
    \kone = \ktwo \Rightarrow \equal(\mout, \mapput(\varm, \kone, \valtwo)) \\
\end{array}
\label{eq:map-prop-3}
\end{equation}

\subsubsection{HashSet}
For the query $\varn = \setsize(\setempty)$, we considered the following grammar:
\begin{equation}
\begin{array}{rcl}
    \nonprop  & :{=} & \tru \mid \varn = \nonsize \\
    \nonsize & :{=} & 0 \mid 1 \\
\end{array}
\label{eq:hashset-grammar-1}
\end{equation}
In this case, \name synthesized the algebraic property
\begin{equation}
\begin{array}{c}
    \varn = 0
\end{array}
\label{eq:set-prop-1}
\end{equation}

For the query $\varn = \setsize(\setadd(\vars, \x))$, we considered the following grammar:
\begin{equation}
\begin{array}{rcl}
    \nonprop  & :{=} & \tru \mid \nonguard \Rightarrow \varn = \nonsize \\
    \nonguard & :{=} & \tru \mid \setcontains(\vars, \x) \mid \neg \setcontains(\vars, \x) \\
    \nonsize & :{=} & \setsize(\vars) \mid \setsize(\vars) + 1 \mid 0 \mid 1 \\
\end{array}
\label{eq:hashset-grammar-2}
\end{equation}
In this case, \name synthesized the algebraic properties
\begin{equation}
\begin{array}{c}
    \setcontains(\vars, \x) \Rightarrow \varn = \setsize(\vars) \\
    \neg \setcontains(\vars, \x) \Rightarrow \varn = \setsize(\vars) + 1 \\
\end{array}
\label{eq:set-prop-2}
\end{equation}

For the query $\varb = \setcontains(\setempty, \x)$, we considered the following grammar:
\begin{equation}
\begin{array}{rcl}
    \nonprop  & :{=} & \tru \mid \varb = \nonap \\
    \nonap & :{=} & \tru \mid \fls \\
\end{array}
\label{eq:hashset-grammar-3}
\end{equation}
In this case, \name synthesized the algebraic property
\begin{equation}
\begin{array}{c}
    \varb = \fls \\
\end{array}
\label{eq:set-prop-3}
\end{equation}

For the query $\varb =\setcontains(\setadd(\vars, \xone), \xtwo)$, we considered the following grammar:
\begin{equation}
\begin{array}{rcl}
    \nonprop  & :{=} & \tru \mid \nonguard \Rightarrow \varb = \nonap \\
    \nonguard & :{=} & \tru \mid \nonint ~\{ {=} \mid {\neq}  \}~ \nonint \\
    \nonint & :{=} & \xone \mid \xtwo \\
    \nonap & :{=} & \setcontains(\vars, \nonint) \mid \tru \mid \fls \\
\end{array}
\label{eq:hashset-grammar-4}
\end{equation}
In this case, \name synthesized the algebraic properties
\begin{equation}
\begin{array}{c}
    \xone = \xtwo \Rightarrow \varb = \tru \\
    \xone \neq \xtwo \Rightarrow \varb = \setcontains(\vars, \xtwo) \\
\end{array}
\label{eq:set-prop-4}
\end{equation}

For the query $\sout = \setremove(\setempty, \x)$, we considered the following grammar:
\begin{equation}
\begin{array}{rcl}
    \nonprop  & :{=} & \tru \mid \sout = \nonhash \\
    \nonhash & :{=} & \setempty \mid \setadd(\nonhash, \x) \\
\end{array}
\label{eq:hashset-grammar-5}
\end{equation}
In this case, \name synthesized the algebraic property
\begin{equation}
\begin{array}{c}
    \sout = \setempty
\end{array}
\label{eq:set-prop-5}
\end{equation}

For the query $\sout =\setremove(\setadd(\vars, \xone), \xtwo)$, we considered the following grammar:
\begin{equation}
\begin{array}{rcl}
    \nonprop  & :{=} & \tru \mid \nonguard \Rightarrow \sout = \nonhash \\
    \nonguard & :{=} & \tru \mid \nonint ~\{ {=} \mid {\neq}  \}~ \nonint \\
    \nonint & :{=} & \xone \mid \xtwo \\
    \nonhash & :{=} & \setempty \mid \vars \mid \setadd(\nonhash, \nonint) \mid \setremove(s, \nonint) \\
\end{array}
\label{eq:hashset-grammar-6}
\end{equation}
In this case, \name synthesized the algebraic properties
\begin{equation}
\begin{array}{c}
    \xone = \xtwo \Rightarrow \sout = \setremove(\vars, \xtwo) \\
    \xone \neq \xtwo \Rightarrow \sout = \setadd(\setremove(\vars, \xtwo), \xone) \\
\end{array}
\label{eq:set-prop-6}
\end{equation}

\subsection{Sensitivity-Analysis Benchmarks}
\label{se:sensitivity-benchmark}
We divide our benchmarks into ones that proved sensitivity properties with respect to the Hamming distance versus the edit distance over input and output lists.

\subsubsection{Hamming Distance}

In the following, $\hdist$ denotes the Hamming distance between two lists.
For queries of the form $\loneout = f(\lone, \xone) \wedge \ltwoout = f(\ltwo, \xtwo)$, where $f \in \{\lcons, \ldelete, \ldeleteall, \lsnoc\}$, we considered the following grammar:
\begin{equation}
\begin{array}{rcl}
    \nonprop  & :{=} & \tru \mid \nonguard \wedge \size(\lone) {=} \size(\ltwo)
                \wedge \hdist(\lone, \ltwo) {\leq} \vard \Rightarrow \hdist(\loneout, \ltwoout) {\leq} \nonsize \\
    \nonguard & :{=} & \top \mid \xone ~\{ {=} \mid {\neq} \}~ \xtwo \\
    \nonsize & :{=} &  \noncoeff*\size(\lone) {+} \noncoeff* \vard {+}  \noncoeff\\
    \noncoeff  & :{=} & -1 \mid 0 \mid 1 \mid 2 \\
\end{array}
\label{eq:hamming-grammar-1}
\end{equation}
In this case, \name synthesized the following sensitivity properties for $\lcons$:
\begin{equation}
\begin{array}{c}
    \size(\lone) {=} \size(\ltwo) \wedge \hdist(\lone, \ltwo) {\leq} \vard \Rightarrow \hdist(\loneout, \ltwoout) {\leq} \size(\lone) {+} 1 \\
    \size(\lone) {=} \size(\ltwo) \wedge \hdist(\lone, \ltwo) {\leq} \vard \Rightarrow \hdist(\loneout, \ltwoout) {\leq} \vard {+} 1 \\
    \xone {=} \xtwo \wedge \size(\lone) {=} \size(\ltwo) \wedge \hdist(\lone, \ltwo) {\leq} \vard \Rightarrow \hdist(\loneout, \ltwoout) {\leq} \size(\lone) \\
    \xone {=} \xtwo \wedge \size(\lone) {=} \size(\ltwo) \wedge \hdist(\lone, \ltwo) {\leq} \vard \Rightarrow \hdist(\loneout, \ltwoout) {\leq} \vard \\
\end{array}
\label{eq:hamming-prop-1}
\end{equation}
the following properties for $\ldelete$,
\begin{equation}
\begin{array}{c}
    \size(\lone) {=} \size(\ltwo) \wedge \hdist(\lone, \ltwo) {\leq} \vard \Rightarrow \hdist(\loneout, \ltwoout) {\leq} \size(\lone) \\
\end{array}
\label{eq:hamming-prop-2}
\end{equation}
the following properties for $\ldeleteall$, 
\begin{equation}
\begin{array}{c}
    \size(\lone) {=} \size(\ltwo) \wedge \hdist(\lone, \ltwo) {\leq} \vard \Rightarrow \hdist(\loneout, \ltwoout) {\leq} \size(\lone) \\
\end{array}
\label{eq:hamming-prop-3}
\end{equation}
and the following properties for $\lsnoc$:
\begin{equation}
\begin{array}{c}
    \size(\lone) {=} \size(\ltwo) \wedge \hdist(\lone, \ltwo) {\leq} \vard \Rightarrow \hdist(\loneout, \ltwoout) {\leq} \size(\lone) {+} 1 \\
    \size(\lone) {=} \size(\ltwo) \wedge \hdist(\lone, \ltwo) {\leq} \vard \Rightarrow \hdist(\loneout, \ltwoout) {\leq} \vard {+} 1 \\
    \xone {=} \xtwo \wedge \size(\lone) {=} \size(\ltwo) \wedge \hdist(\lone, \ltwo) {\leq} \vard \Rightarrow \hdist(\loneout, \ltwoout) {\leq} \size(\lone) \\
    \xone {=} \xtwo \wedge \size(\lone) {=} \size(\ltwo) \wedge \hdist(\lone, \ltwo) {\leq} \vard \Rightarrow \hdist(\loneout, \ltwoout) {\leq} \vard \\
\end{array}
\label{eq:hamming-prop-4}
\end{equation}

For queries of the form $\loneout = f(\lone) \wedge \ltwoout = f(\ltwo)$, with $f \in \{\lreverse, \lstutter, \ltail\}$, we considered the following grammar:
\begin{equation}
\begin{array}{rcl}
    \nonprop  & :{=} & \tru \mid  \size(\lone) {=} \size(\ltwo)
                \wedge \hdist(\lone, \ltwo) {\leq} \vard \Rightarrow \hdist(\loneout, \ltwoout) {\leq} \nonsize \\
    \nonsize & :{=} &  \noncoeff *\size(\lone) {+} \noncoeff *\vard {+}  \noncoeff\\
    \noncoeff  & :{=} & -1 \mid 0 \mid 1 \mid 2 \\
\end{array}
\label{eq:hamming-grammar-2}
\end{equation}
In this case, \name synthesized the following sensitivity properties for $\lreverse$,
\begin{equation}
\begin{array}{c}
    \size(\lone) {=} \size(\ltwo) \wedge \hdist(\lone, \ltwo) {\leq} \vard \Rightarrow \hdist(\loneout, \ltwoout) {\leq} \size(\lone) \\
    \size(\lone) {=} \size(\ltwo) \wedge \hdist(\lone, \ltwo) {\leq} \vard \Rightarrow \hdist(\loneout, \ltwoout) {\leq} d \\
\end{array}
\label{eq:hamming-prop-5}
\end{equation}
the following properties for $\lstutter$,
\begin{equation}
\begin{array}{c}
    \size(\lone) {=} \size(\ltwo) \wedge \hdist(\lone, \ltwo) {\leq} \vard \Rightarrow \hdist(\loneout, \ltwoout) {\leq} 2\size(\lone) \\
    \size(\lone) {=} \size(\ltwo) \wedge \hdist(\lone, \ltwo) {\leq} \vard \Rightarrow \hdist(\loneout, \ltwoout) {\leq} 2d \\
\end{array}
\label{eq:hamming-prop-6}
\end{equation}
and the following properties for $\ltail$: 
\begin{equation}
\begin{array}{c}
    \size(\lone) {=} \size(\ltwo) \wedge \hdist(\lone, \ltwo) {\leq} \vard \Rightarrow \hdist(\loneout, \ltwoout) {\leq} \size(\lone) {-} 1\\
    \size(\lone) {=} \size(\ltwo) \wedge \hdist(\lone, \ltwo) {\leq} \vard \Rightarrow \hdist(\loneout, \ltwoout) {\leq} d \\
\end{array}
\label{eq:hamming-prop-7}
\end{equation}

Properties for $\loneout = \ldelete(\lcons(\lone, \x), \x) \wedge \ltwoout = \ldelete(\lcons(\ltwo, \x), \x)$ are also expressed in the grammar Eq.~\eqref{eq:hamming-grammar-1}.
In this case, \name synthesized the following sensitivity properties: 
\begin{equation}
\begin{array}{c}
    \size(\lone) {=} \size(\ltwo) \wedge \hdist(\lone, \ltwo) {\leq} \vard \Rightarrow \hdist(\loneout, \ltwoout) {\leq} \size(\lone)\\
    \size(\lone) {=} \size(\ltwo) \wedge \hdist(\lone, \ltwo) {\leq} \vard \Rightarrow \hdist(\loneout, \ltwoout) {\leq} d \\
\end{array}
\label{eq:hamming-prop-8}
\end{equation}

For queries of the form  $\loneout = \lappend(\varl_{11}, \varl_{12})\wedge \ltwoout = \lappend(\varl_{21}, \varl_{22})$, we considered the following grammar:
\begin{equation}
\begin{array}{rcl}
    \nonprop  & :{=} & \tru \mid \size(\varl_{11}) {=} \size(\varl_{21}) \wedge \size(\varl_{12}) {=} \size(\varl_{22})
                \wedge \hdist(\varl_{11}, \varl_{21}) {\leq} \done  \\
                & & \qquad \wedge\, \hdist(\varl_{11}, \varl_{21}) {\leq} \dtwo \Rightarrow \hdist(\loneout, \ltwoout) {\leq} \nonsize \\
    \nonsize & :{=} &  \noncoeff *\size(\varl_{11}) {+} \noncoeff *\size(\varl_{12}) {+} \noncoeff* \done {+} \noncoeff* \dtwo {+}  \noncoeff\\
    \noncoeff  & :{=} & -1 \mid 0 \mid 1 \mid 2 \\
\end{array}
\label{eq:hamming-grammar-3}
\end{equation}
In this case, \name synthesized the following sensitivity properties:
\begin{equation}
\begin{array}{c}
\psi \wedge\, \hdist(\varl_{11}, \varl_{21}) {\leq} \dtwo \Rightarrow \hdist(\loneout, \ltwoout) {\leq} \size(\varl_{11}) {+} \dtwo \\
\psi \wedge\, \hdist(\varl_{11}, \varl_{21}) {\leq} \dtwo \Rightarrow \hdist(\loneout, \ltwoout) {\leq} \size(\varl_{21}) {+} \done \\
\psi \wedge\, \hdist(\varl_{11}, \varl_{21}) {\leq} \dtwo \Rightarrow \hdist(\loneout, \ltwoout) {\leq} \size(\varl_{11}) {+} \size(\varl_{21}) \\
\psi \wedge\, \hdist(\varl_{11}, \varl_{21}) {\leq} \dtwo \Rightarrow \hdist(\loneout, \ltwoout) {\leq} \done {+} \dtwo \\
\end{array}
\label{eq:hamming-prop-9}
\end{equation}
where $\psi := \size(\varl_{11}) {=} \size(\varl_{21}) \wedge \size(\varl_{12}) {=} \size(\varl_{22}) \wedge \hdist(\varl_{11}, \varl_{21}) {\leq} \done \wedge \hdist(\varl_{11}, \varl_{21}) {\leq} \dtwo$.

\subsubsection{Edit Distance}

For queries of the form $\loneout = f(\lone, \xone)$ and $\ltwoout = f(\ltwo, \xtwo)$, with $f \in \{\lcons, \ldelete, \ldeleteall, \lsnoc\}$, we considered the following grammar:
\begin{equation}
\begin{array}{rcl}
    \nonprop  & :{=} & \tru \mid \nonguard \wedge \edist(\lone, \ltwo) {\leq} \vard \Rightarrow \edist(\loneout, \ltwoout) {\leq} \nonsize \\
    \nonguard & :{=} & \top \mid \xone ~\{ {=} \mid {\neq} \}~ \xtwo \\
    \nonsize & :{=} &  \noncoeff *\size(\lone) {+} \noncoeff *\size(\ltwo) {+} \noncoeff *\vard {+}  \noncoeff\\
    \noncoeff  & :{=} & -1 \mid 0 \mid 1 \mid 2 \\
\end{array}
\label{eq:edit-grammar-1}
\end{equation}
In this case, \name synthesized the following sensitivity properties for $\lcons$:
\begin{equation}
\begin{array}{c}
    \tru
\end{array}
\label{eq:edit-prop-1}
\end{equation}
the following properties for $\ldelete$,
\begin{equation}
\begin{array}{c}
    \edist(\lone, \ltwo) {\leq} \vard \Rightarrow \edist(\loneout, \ltwoout) {\leq} \size(\lone) + \size(\ltwo) \\
    \edist(\lone, \ltwo) {\leq} \vard \Rightarrow \edist(\loneout, \ltwoout) {\leq} \size(\lone) + d \\
    \edist(\lone, \ltwo) {\leq} \vard \Rightarrow \edist(\loneout, \ltwoout) {\leq} \size(\ltwo) + d \\
    \edist(\lone, \ltwo) {\leq} \vard \Rightarrow \edist(\loneout, \ltwoout) {\leq} d + 2 \\
    \xone = \xtwo \wedge \edist(\lone, \ltwo) {\leq} \vard \Rightarrow \edist(\loneout, \ltwoout) {\leq} \size(\lone) - \size(\ltwo) + 2d \\
    \xone = \xtwo \wedge \edist(\lone, \ltwo) {\leq} \vard \Rightarrow \edist(\loneout, \ltwoout) {\leq} - \size(\lone) + \size(\ltwo) + 2d \\
    \xone = \xtwo \wedge \edist(\lone, \ltwo) {\leq} \vard \Rightarrow \edist(\loneout, \ltwoout) {\leq} d + 1 \\
\end{array}
\label{eq:edit-prop-2}
\end{equation}
the following properties for $\ldeleteall$, 
\begin{equation}
\begin{array}{c}
    \edist(\lone, \ltwo) {\leq} \vard \Rightarrow \edist(\loneout, \ltwoout) {\leq} \size(\lone) + \size(\ltwo) \\
    \edist(\lone, \ltwo) {\leq} \vard \Rightarrow \edist(\loneout, \ltwoout) {\leq} \size(\lone) + d \\
    \edist(\lone, \ltwo) {\leq} \vard \Rightarrow \edist(\loneout, \ltwoout) {\leq} \size(\ltwo) + d \\
    \xone = \xtwo \wedge \edist(\lone, \ltwo) {\leq} \vard \Rightarrow \edist(\loneout, \ltwoout) {\leq} d \\
\end{array}
\label{eq:edit-prop-3}
\end{equation}
and the following properties for $\lsnoc$:
\begin{equation}
\begin{array}{c}
    \xone = \xtwo \wedge \edist(\lone, \ltwo) {\leq} \vard \Rightarrow \edist(\loneout, \ltwoout) {\leq} \vard \\
    \xone = \xtwo \wedge \edist(\lone, \ltwo) {\leq} \vard \Rightarrow \edist(\loneout, \ltwoout) {\leq} \size(\lone) + \size(\ltwo) \\
    \edist(\lone, \ltwo) {\leq} \vard \Rightarrow \edist(\loneout, \ltwoout) {\leq} \vard + 1 \\
    \edist(\lone, \ltwo) {\leq} \vard \Rightarrow \edist(\loneout, \ltwoout) {\leq} \size(\lone) + \size(\ltwo) + 1 \\
\end{array}
\label{eq:edit-prop-4}
\end{equation}

For queries of the form  $\loneout = f(\lone)\wedge\ltwoout = f(\ltwo)$, with $f \in \{\lreverse, \lstutter, \ltail\}$, we considered the following grammar:
\begin{equation}
\begin{array}{rcl}
    \nonprop  & :{=} & \tru \mid \edist(\lone, \ltwo) {\leq} \vard \Rightarrow \edist(\loneout, \ltwoout) {\leq} \nonsize \\
    \nonsize & :{=} &  \noncoeff *\size(\lone) {+} \noncoeff *\size(\ltwo) {+} \noncoeff *\vard {+}  \noncoeff\\
    \noncoeff  & :{=} & -1 \mid 0 \mid 1 \mid 2 \\
\end{array}
\label{eq:edit-grammar-2}
\end{equation}
In this case, \name synthesized the following sensitivity properties for $\lreverse$,
\begin{equation}
\begin{array}{c}
    \edist(\lone, \ltwo) {\leq} \vard \Rightarrow \edist(\loneout, \ltwoout) {\leq} \size(\lone) + \size(\ltwo) \\
    \edist(\lone, \ltwo) {\leq} \vard \Rightarrow \edist(\loneout, \ltwoout) {\leq} d \\
\end{array}
\label{eq:edit-prop-5}
\end{equation}
the following properties for $\lstutter$,
\begin{equation}
\begin{array}{c}
    \edist(\lone, \ltwo) {\leq} \vard \Rightarrow \edist(\loneout, \ltwoout) {\leq} 2 \size(\lone) + 2 \size(\ltwo) \\
    \edist(\lone, \ltwo) {\leq} \vard \Rightarrow \edist(\loneout, \ltwoout) {\leq} 2 d \\
\end{array}
\label{eq:edit-prop-6}
\end{equation}
and the following properties for $\ltail$: 
\begin{equation}
\begin{array}{c}
    \edist(\lone, \ltwo) {\leq} \vard \Rightarrow \edist(\loneout, \ltwoout) {\leq} \size(\lone) + \size(\ltwo) - 1 \\
    \edist(\lone, \ltwo) {\leq} \vard \Rightarrow \edist(\loneout, \ltwoout) {\leq} d \\
\end{array}
\label{eq:edit-prop-7}
\end{equation}

Properties for $\loneout = \ldelete(\lcons(\lone, \x), \x)\wedge \ltwoout = \ldelete(\lcons(\ltwo, \x), \x)$ are also expressed in the grammar Eq.~\eqref{eq:edit-grammar-1}.
In this case, \name synthesized the following sensitivity properties:
\begin{equation}
\begin{array}{c}
    \edist(\lone, \ltwo) {\leq} \vard \Rightarrow \edist(\loneout, \ltwoout) {\leq} \size(\lone) + \size(\ltwo) \\
    \xone = \xtwo \wedge \edist(\lone, \ltwo) {\leq} \vard \Rightarrow \edist(\loneout, \ltwoout) {\leq} d \\
\end{array}
\label{eq:edit-prop-8}
\end{equation}

For queries of the form  $\loneout = \lappend(\varl_{11}, \varl_{12})\wedge\ltwoout = \lappend(\varl_{21}, \varl_{22})$, we considered the following grammar:
\begin{equation}
\begin{array}{rcl}
    \nonprop  & :{=} & \tru \mid \edist(\varl_{11}, \varl_{21}) {\leq} \done \wedge \edist(\varl_{12}, \varl_{22}) {\leq} \dtwo \Rightarrow \edist(\loneout, \ltwoout) {\leq} \nonsize \\
    \nonsize & :{=} &  \noncoeff* \size(\varl_{11}) {+} \noncoeff *\size(\varl_{12}) {+} \noncoeff *\size(\varl_{21}) {+} \noncoeff *\size(\varl_{22}) {+} \noncoeff *\done {+} \noncoeff *\dtwo {+}  \noncoeff\\
    \noncoeff  & :{=} & -1 \mid 0 \mid 1 \mid 2 \\
\end{array}
\label{eq:edit-grammar-3}
\end{equation}
In this case, \name synthesized the following sensitivity properties:
\begin{equation}
\begin{array}{c}
    \edist(\varl_{11}, \varl_{21}) {\leq} \done \wedge \edist(\varl_{11}, \varl_{21}) {\leq} \dtwo \Rightarrow \edist(\loneout, \ltwoout) {\leq} \done {+}  \dtwo \\
    \edist(\varl_{11}, \varl_{21}) {\leq} \done \wedge \edist(\varl_{11}, \varl_{21}) {\leq} \dtwo \Rightarrow \edist(\loneout, \ltwoout) {\leq} \size(\varl_{12}) {+}  \size(\varl_{22}) {+}  \done \\
    \edist(\varl_{11}, \varl_{21}) {\leq} \done \wedge \edist(\varl_{11}, \varl_{21}) {\leq} \dtwo \Rightarrow \edist(\loneout, \ltwoout) {\leq} \size(\varl_{11}) {+}  \size(\varl_{21}) {+}  \dtwo \\
    \edist(\varl_{11}, \varl_{21}) {\leq} \done \wedge \edist(\varl_{11}, \varl_{21}) {\leq} \dtwo \Rightarrow \edist(\loneout, \ltwoout) {\leq} \size(\varl_{11}) {+} \size(\varl_{12}) {+}  \size(\varl_{21}) {+}  \size(\varl_{22})  \\
\end{array}
\label{eq:edit-prop-9}
\end{equation}

\subsection{Bit-Vector Polyhedra Benchmarks}
\label{App:BitVectorPolyhedraBenchmarks}

We considered the following 9 queries, showcasing a variety of behaviors, when synthesizing \lconjunctions of linear bit-vector inequalities:
\begin{equation}
\hspace{-2mm}
\begin{array}{lcl}
    \bvsquare & {:} & y = x * x  \\
    \bvcube & {:} & y = x * x * x \\
    \bvhalf & {:} & y = x / 2 \\
    \bvsquareineq & {:} & y \leq x * x \\
    \bvconstneq & {:} & y \neq 7 \\
    \bvconj & {:} & (x + y + 4 \leq x + 15y + 7) \wedge (2x + 14y + 3 \leq x + 15y + 7) \\
    \bvsinglepoint & {:} & (x = 3 \wedge y=7) \\
    \bvfourpoints & {:} & (x = 3 \wedge y=7) \vee (x=9 \wedge y=10) \vee  (x=6 \wedge y=1) \vee (x=0 \wedge y=15) \\
    \bvdisj & {:} & (x + y + 4 \leq x + 15y + 7) \vee (2x + 14y + 3 \leq x + 15y + 7)
\end{array}
\label{eq:bv-problems}
\end{equation}
For each query, we used the following grammar:
\begin{equation}
\begin{array}{rcl}
    \nonprop  & :{=} & \noncoeff x {+} \noncoeff y {+} \noncoeff {\leq} \noncoeff x {+} \noncoeff y {+} \noncoeff  \\
    \noncoeff  & :{=} & 0 \mid 1 \mid \ldots \mid 15 \\
\end{array}
\label{eq:bv-grammar}
\end{equation}

\begin{figure}[tb!]
\centering
\begin{spacing}{1.2}
\begin{tabular}{cc}
\resizebox{0.4\textwidth}{!}{
  \begin{tabular}{r|*{16}{c}|}
    \hhline{~----------------}
    15&\cellcolor{dgreen}&\cellcolor{dgreen}&\cellcolor{dgreen}&\cellcolor{dgreen}& & &\cellcolor{dgreen}&\cellcolor{dgreen}& & & & & & &\cellcolor{dgreen}&\cellcolor{dgreen} \\
    14&\cellcolor{dgreen}&\cellcolor{dgreen}&\cellcolor{dgreen}& & &\cellcolor{dgreen}&\cellcolor{dgreen}& & & & & & &\cellcolor{dgreen}&\cellcolor{dgreen}&\cellcolor{dgreen} \\
    13&\cellcolor{dgreen}&\cellcolor{dgreen}& & &\cellcolor{dgreen}&\cellcolor{dgreen}& & & & & & &\cellcolor{dgreen}&\cellcolor{dgreen}&\cellcolor{dgreen}&\cellcolor{dgreen} \\
    12&\cellcolor{dgreen}& & &\cellcolor{dgreen}&\cellcolor{dgreen}& & & & & & &\cellcolor{dgreen}&\cellcolor{dgreen}&\cellcolor{dgreen}&\cellcolor{dgreen}&\cellcolor{dgreen} \\
    11& & &\cellcolor{dgreen}&\cellcolor{dgreen}& & & & & & &\cellcolor{dgreen}&\cellcolor{dgreen}&\cellcolor{dgreen}&\cellcolor{dgreen}&\cellcolor{dgreen}&\cellcolor{dgreen} \\
    10& &\cellcolor{dgreen}&\cellcolor{dgreen}& & & & & & &\cellcolor{dgreen}&\cellcolor{dgreen}&\cellcolor{dgreen}&\cellcolor{dgreen}&\cellcolor{dgreen}&\cellcolor{dgreen}&  \\
     9&\cellcolor{dgreen}&\cellcolor{dgreen}& & & & & & &\cellcolor{dgreen}&\cellcolor{dgreen}&\cellcolor{dgreen}&\cellcolor{dgreen}&\cellcolor{dgreen}&\cellcolor{dgreen}& &  \\
     8&\cellcolor{dgreen}& & & & & & &\cellcolor{dgreen}&\cellcolor{dgreen}&\cellcolor{dgreen}&\cellcolor{dgreen}&\cellcolor{dgreen}&\cellcolor{dgreen}& & &\cellcolor{dgreen} \\
     7& & & & & & &\cellcolor{dgreen}&\cellcolor{dgreen}&\cellcolor{dgreen}&\cellcolor{dgreen}&\cellcolor{dgreen}&\cellcolor{dgreen}& & &\cellcolor{dgreen}&\cellcolor{dgreen} \\
     6& & & & & &\cellcolor{dgreen}&\cellcolor{dgreen}&\cellcolor{dgreen}&\cellcolor{dgreen}&\cellcolor{dgreen}&\cellcolor{dgreen}& & &\cellcolor{dgreen}&\cellcolor{dgreen}&  \\
     5& & & & &\cellcolor{dgreen}&\cellcolor{dgreen}&\cellcolor{dgreen}&\cellcolor{dgreen}&\cellcolor{dgreen}&\cellcolor{dgreen}& & &\cellcolor{dgreen}&\cellcolor{dgreen}& &  \\
     4& & & &\cellcolor{dgreen}&\cellcolor{dgreen}&\cellcolor{dgreen}&\cellcolor{dgreen}&\cellcolor{dgreen}&\cellcolor{dgreen}& & &\cellcolor{dgreen}&\cellcolor{dgreen}& & &  \\
     3& & &\cellcolor{dgreen}&\cellcolor{dgreen}&\cellcolor{dgreen}&\cellcolor{dgreen}&\cellcolor{dgreen}&\cellcolor{dgreen}& & &\cellcolor{dgreen}&\cellcolor{dgreen}& & & &  \\
     2& &\cellcolor{dgreen}&\cellcolor{dgreen}&\cellcolor{dgreen}&\cellcolor{dgreen}&\cellcolor{dgreen}&\cellcolor{dgreen}& & &\cellcolor{dgreen}&\cellcolor{dgreen}& & & & &  \\
     1&\cellcolor{dgreen}&\cellcolor{dgreen}&\cellcolor{dgreen}&\cellcolor{dgreen}&\cellcolor{dgreen}&\cellcolor{dgreen}& & &\cellcolor{dgreen}&\cellcolor{dgreen}& & & & & &  \\
     0&\cellcolor{dgreen}&\cellcolor{dgreen}&\cellcolor{dgreen}&\cellcolor{dgreen}&\cellcolor{dgreen}& & &\cellcolor{dgreen}&\cellcolor{dgreen}& & & & & & &\cellcolor{dgreen} \\
    \hhline{~----------------}
    \multicolumn{1}{c}{ }&0&1&2&3&4&5&6&7&8&9&\hspace{-4pt}10\hspace{-4pt}&\hspace{-4pt}11\hspace{-4pt}&\hspace{-4pt}12\hspace{-4pt}&\hspace{-4pt}13\hspace{-4pt}&\hspace{-4pt}14\hspace{-4pt}&\multicolumn{1}{c}{\hspace{-4pt}15\hspace{-4pt}}
  \end{tabular}
}
&
\resizebox{0.4\textwidth}{!}{
  \begin{tabular}{r|*{16}{c}|}
    \hhline{~----------------}
    15&\cellcolor{dgreen}&\cellcolor{dgreen}&\cellcolor{dgreen}&\cellcolor{dgreen}&\cellcolor{dgreen}&\cellcolor{dgreen}&\cellcolor{dgreen}&\cellcolor{dgreen}& & & & & &\cellcolor{dgreen}&\cellcolor{dgreen}&\cellcolor{dgreen} \\
    14&\cellcolor{dgreen}&\cellcolor{dgreen}&\cellcolor{dgreen}&\cellcolor{dgreen}&\cellcolor{dgreen}&\cellcolor{dgreen}&\cellcolor{dgreen}& & & & & & & &\cellcolor{dgreen}&\cellcolor{dgreen} \\
    13&\cellcolor{dgreen}&\cellcolor{dgreen}&\cellcolor{dgreen}&\cellcolor{dgreen}&\cellcolor{dgreen}&\cellcolor{dgreen}& & & & & & & & & &\cellcolor{dgreen} \\
    12&\cellcolor{dgreen}&\cellcolor{dgreen}&\cellcolor{dgreen}&\cellcolor{dgreen}&\cellcolor{dgreen}& & & & & & & & & & &  \\
    11& &\cellcolor{dgreen}&\cellcolor{dgreen}&\cellcolor{dgreen}& & & & & & & & & & & &  \\
    10& & &\cellcolor{dgreen}& & & & & & & & & & & & &  \\
     9&\cellcolor{dgreen}&\cellcolor{dgreen}& &\cellcolor{dgreen}&\cellcolor{dgreen}&\cellcolor{dgreen}&\cellcolor{dgreen}&\cellcolor{dgreen}&\cellcolor{dgreen}&\cellcolor{dgreen}&\cellcolor{dgreen}&\cellcolor{dgreen}&\cellcolor{dgreen}&\cellcolor{dgreen}&\cellcolor{dgreen}&\cellcolor{dgreen} \\
     8&\cellcolor{dgreen}& & & &\cellcolor{dgreen}&\cellcolor{dgreen}&\cellcolor{dgreen}&\cellcolor{dgreen}&\cellcolor{dgreen}&\cellcolor{dgreen}&\cellcolor{dgreen}&\cellcolor{dgreen}&\cellcolor{dgreen}&\cellcolor{dgreen}&\cellcolor{dgreen}&\cellcolor{dgreen} \\
     7& & & & & &\cellcolor{dgreen}&\cellcolor{dgreen}&\cellcolor{dgreen}&\cellcolor{dgreen}&\cellcolor{dgreen}&\cellcolor{dgreen}&\cellcolor{dgreen}&\cellcolor{dgreen}&\cellcolor{dgreen}&\cellcolor{dgreen}&\cellcolor{dgreen} \\
     6& & & & & & &\cellcolor{dgreen}&\cellcolor{dgreen}&\cellcolor{dgreen}&\cellcolor{dgreen}&\cellcolor{dgreen}&\cellcolor{dgreen}&\cellcolor{dgreen}&\cellcolor{dgreen}&\cellcolor{dgreen}&  \\
     5& & & & & & & &\cellcolor{dgreen}&\cellcolor{dgreen}&\cellcolor{dgreen}&\cellcolor{dgreen}&\cellcolor{dgreen}&\cellcolor{dgreen}&\cellcolor{dgreen}& &  \\
     4& & & & & & & & &\cellcolor{dgreen}&\cellcolor{dgreen}&\cellcolor{dgreen}&\cellcolor{dgreen}&\cellcolor{dgreen}& & &  \\
     3& & & & & & & & & &\cellcolor{dgreen}&\cellcolor{dgreen}&\cellcolor{dgreen}& & & &  \\
     2& & & & & & & & & & &\cellcolor{dgreen}& & & & &  \\
     1&\cellcolor{dgreen}&\cellcolor{dgreen}&\cellcolor{dgreen}&\cellcolor{dgreen}&\cellcolor{dgreen}&\cellcolor{dgreen}&\cellcolor{dgreen}&\cellcolor{dgreen}&\cellcolor{dgreen}&\cellcolor{dgreen}& &\cellcolor{dgreen}&\cellcolor{dgreen}&\cellcolor{dgreen}&\cellcolor{dgreen}&\cellcolor{dgreen} \\
     0&\cellcolor{dgreen}&\cellcolor{dgreen}&\cellcolor{dgreen}&\cellcolor{dgreen}&\cellcolor{dgreen}&\cellcolor{dgreen}&\cellcolor{dgreen}&\cellcolor{dgreen}&\cellcolor{dgreen}& & & &\cellcolor{dgreen}&\cellcolor{dgreen}&\cellcolor{dgreen}&\cellcolor{dgreen} \\
    \hhline{~----------------}
    \multicolumn{1}{c}{ }&0&1&2&3&4&5&6&7&8&9&\hspace{-4pt}10\hspace{-4pt}&\hspace{-4pt}11\hspace{-4pt}&\hspace{-4pt}12\hspace{-4pt}&\hspace{-4pt}13\hspace{-4pt}&\hspace{-4pt}14\hspace{-4pt}&\multicolumn{1}{c}{\hspace{-4pt}15\hspace{-4pt}}
  \end{tabular}
}
\\
{\small $2 * x + 14 * y + 3 \le x + 15 * y + 7$} & {\small $x + y + 4 \le x + 15 * y + 7$} \\
{\small (a)} & {\small (b)}
\\
\\
\resizebox{0.4\textwidth}{!}{
  \begin{tabular}{r|*{16}{c}|}
    \hhline{~----------------}
    15&\cellcolor{dgreen}&\cellcolor{dgreen}&\cellcolor{dgreen}&\cellcolor{dgreen}&\cellcolor{dgreen}&\cellcolor{dgreen}&\cellcolor{dgreen}&\cellcolor{dgreen}& & & & & &\cellcolor{dgreen}&\cellcolor{dgreen}&\cellcolor{dgreen} \\
    14&\cellcolor{dgreen}&\cellcolor{dgreen}&\cellcolor{dgreen}&\cellcolor{dgreen}&\cellcolor{dgreen}&\cellcolor{dgreen}&\cellcolor{dgreen}& & & & & & &\cellcolor{dgreen}&\cellcolor{dgreen}&\cellcolor{dgreen} \\
    13&\cellcolor{dgreen}&\cellcolor{dgreen}&\cellcolor{dgreen}&\cellcolor{dgreen}&\cellcolor{dgreen}&\cellcolor{dgreen}& & & & & & &\cellcolor{dgreen}&\cellcolor{dgreen}&\cellcolor{dgreen}&\cellcolor{dgreen} \\
    12&\cellcolor{dgreen}&\cellcolor{dgreen}&\cellcolor{dgreen}&\cellcolor{dgreen}&\cellcolor{dgreen}& & & & & & &\cellcolor{dgreen}&\cellcolor{dgreen}&\cellcolor{dgreen}&\cellcolor{dgreen}&\cellcolor{dgreen} \\
    11& &\cellcolor{dgreen}&\cellcolor{dgreen}&\cellcolor{dgreen}& & & & & & &\cellcolor{dgreen}&\cellcolor{dgreen}&\cellcolor{dgreen}&\cellcolor{dgreen}&\cellcolor{dgreen}&\cellcolor{dgreen} \\
    10& &\cellcolor{dgreen}&\cellcolor{dgreen}& & & & & & &\cellcolor{dgreen}&\cellcolor{dgreen}&\cellcolor{dgreen}&\cellcolor{dgreen}&\cellcolor{dgreen}&\cellcolor{dgreen}&  \\
     9&\cellcolor{dgreen}&\cellcolor{dgreen}& &\cellcolor{dgreen}&\cellcolor{dgreen}&\cellcolor{dgreen}&\cellcolor{dgreen}&\cellcolor{dgreen}&\cellcolor{dgreen}&\cellcolor{dgreen}&\cellcolor{dgreen}&\cellcolor{dgreen}&\cellcolor{dgreen}&\cellcolor{dgreen}&\cellcolor{dgreen}&\cellcolor{dgreen} \\
     8&\cellcolor{dgreen}& & & &\cellcolor{dgreen}&\cellcolor{dgreen}&\cellcolor{dgreen}&\cellcolor{dgreen}&\cellcolor{dgreen}&\cellcolor{dgreen}&\cellcolor{dgreen}&\cellcolor{dgreen}&\cellcolor{dgreen}&\cellcolor{dgreen}&\cellcolor{dgreen}&\cellcolor{dgreen} \\
     7& & & & & &\cellcolor{dgreen}&\cellcolor{dgreen}&\cellcolor{dgreen}&\cellcolor{dgreen}&\cellcolor{dgreen}&\cellcolor{dgreen}&\cellcolor{dgreen}&\cellcolor{dgreen}&\cellcolor{dgreen}&\cellcolor{dgreen}&\cellcolor{dgreen} \\
     6& & & & & &\cellcolor{dgreen}&\cellcolor{dgreen}&\cellcolor{dgreen}&\cellcolor{dgreen}&\cellcolor{dgreen}&\cellcolor{dgreen}&\cellcolor{dgreen}&\cellcolor{dgreen}&\cellcolor{dgreen}&\cellcolor{dgreen}&  \\
     5& & & & &\cellcolor{dgreen}&\cellcolor{dgreen}&\cellcolor{dgreen}&\cellcolor{dgreen}&\cellcolor{dgreen}&\cellcolor{dgreen}&\cellcolor{dgreen}&\cellcolor{dgreen}&\cellcolor{dgreen}&\cellcolor{dgreen}& &  \\
     4& & & &\cellcolor{dgreen}&\cellcolor{dgreen}&\cellcolor{dgreen}&\cellcolor{dgreen}&\cellcolor{dgreen}&\cellcolor{dgreen}&\cellcolor{dgreen}&\cellcolor{dgreen}&\cellcolor{dgreen}&\cellcolor{dgreen}& & &  \\
     3& & &\cellcolor{dgreen}&\cellcolor{dgreen}&\cellcolor{dgreen}&\cellcolor{dgreen}&\cellcolor{dgreen}&\cellcolor{dgreen}& &\cellcolor{dgreen}&\cellcolor{dgreen}&\cellcolor{dgreen}& & & &  \\
     2& &\cellcolor{dgreen}&\cellcolor{dgreen}&\cellcolor{dgreen}&\cellcolor{dgreen}&\cellcolor{dgreen}&\cellcolor{dgreen}& & &\cellcolor{dgreen}&\cellcolor{dgreen}& & & & &  \\
     1&\cellcolor{dgreen}&\cellcolor{dgreen}&\cellcolor{dgreen}&\cellcolor{dgreen}&\cellcolor{dgreen}&\cellcolor{dgreen}&\cellcolor{dgreen}&\cellcolor{dgreen}&\cellcolor{dgreen}&\cellcolor{dgreen}& &\cellcolor{dgreen}&\cellcolor{dgreen}&\cellcolor{dgreen}&\cellcolor{dgreen}&\cellcolor{dgreen} \\
     0&\cellcolor{dgreen}&\cellcolor{dgreen}&\cellcolor{dgreen}&\cellcolor{dgreen}&\cellcolor{dgreen}&\cellcolor{dgreen}&\cellcolor{dgreen}&\cellcolor{dgreen}&\cellcolor{dgreen}& & & &\cellcolor{dgreen}&\cellcolor{dgreen}&\cellcolor{dgreen}&\cellcolor{dgreen} \\
    \hhline{~----------------}
    \multicolumn{1}{c}{ }&0&1&2&3&4&5&6&7&8&9&\hspace{-4pt}10\hspace{-4pt}&\hspace{-4pt}11\hspace{-4pt}&\hspace{-4pt}12\hspace{-4pt}&\hspace{-4pt}13\hspace{-4pt}&\hspace{-4pt}14\hspace{-4pt}&\multicolumn{1}{c}{\hspace{-4pt}15\hspace{-4pt}}
  \end{tabular}
}
&
\resizebox{0.4\textwidth}{!}{
  \begin{tabular}{r|*{16}{c}|}
    \hhline{~----------------}
    15&\cellcolor{dgreen}&\cellcolor{dgreen}&\cellcolor{dgreen}&\cellcolor{dgreen}&\cellcolor{dgreen}&\cellcolor{dgreen}&\cellcolor{dgreen}&\cellcolor{dgreen}& & \cellcolor{dred}& \cellcolor{dred}& \cellcolor{dred}& \cellcolor{dred}&\cellcolor{dgreen}&\cellcolor{dgreen}&\cellcolor{dgreen} \\
    14&\cellcolor{dgreen}&\cellcolor{dgreen}&\cellcolor{dgreen}&\cellcolor{dgreen}&\cellcolor{dgreen}&\cellcolor{dgreen}&\cellcolor{dgreen}& & \cellcolor{dred}& \cellcolor{dred}& \cellcolor{dred}& \cellcolor{dred}& \cellcolor{dred}&\cellcolor{dgreen}&\cellcolor{dgreen}&\cellcolor{dgreen} \\
    13&\cellcolor{dgreen}&\cellcolor{dgreen}&\cellcolor{dgreen}&\cellcolor{dgreen}&\cellcolor{dgreen}&\cellcolor{dgreen}& & \cellcolor{dred}& \cellcolor{dred}& \cellcolor{dred}& \cellcolor{dred}& \cellcolor{dred}&\cellcolor{dgreen}&\cellcolor{dgreen}&\cellcolor{dgreen}&\cellcolor{dgreen} \\
    12&\cellcolor{dgreen}&\cellcolor{dgreen}&\cellcolor{dgreen}&\cellcolor{dgreen}&\cellcolor{dgreen}& & \cellcolor{dred}& \cellcolor{dred}& \cellcolor{dred}& \cellcolor{dred}& \cellcolor{dred}&\cellcolor{dgreen}&\cellcolor{dgreen}&\cellcolor{dgreen}&\cellcolor{dgreen}&\cellcolor{dgreen} \\
    11&\cellcolor{dred}&\cellcolor{dgreen}&\cellcolor{dgreen}&\cellcolor{dgreen}& & \cellcolor{dred}& \cellcolor{dred}& \cellcolor{dred}& \cellcolor{dred}& \cellcolor{dred}&\cellcolor{dgreen}&\cellcolor{dgreen}&\cellcolor{dgreen}&\cellcolor{dgreen}&\cellcolor{dgreen}&\cellcolor{dgreen} \\
    10&\cellcolor{dred}&\cellcolor{dgreen}&\cellcolor{dgreen}& & \cellcolor{dred}& \cellcolor{dred}& \cellcolor{dred}& \cellcolor{dred}& \cellcolor{dred}&\cellcolor{dgreen}&\cellcolor{dgreen}&\cellcolor{dgreen}&\cellcolor{dgreen}&\cellcolor{dgreen}&\cellcolor{dgreen}&\cellcolor{dred} \\
     9&\cellcolor{dgreen}&\cellcolor{dgreen}& &\cellcolor{dgreen}&\cellcolor{dgreen}&\cellcolor{dgreen}&\cellcolor{dgreen}&\cellcolor{dgreen}&\cellcolor{dgreen}&\cellcolor{dgreen}&\cellcolor{dgreen}&\cellcolor{dgreen}&\cellcolor{dgreen}&\cellcolor{dgreen}&\cellcolor{dgreen}&\cellcolor{dgreen} \\
     8&\cellcolor{dgreen}& & \cellcolor{dred}& \cellcolor{dred}&\cellcolor{dgreen}&\cellcolor{dgreen}&\cellcolor{dgreen}&\cellcolor{dgreen}&\cellcolor{dgreen}&\cellcolor{dgreen}&\cellcolor{dgreen}&\cellcolor{dgreen}&\cellcolor{dgreen}&\cellcolor{dgreen}&\cellcolor{dgreen}&\cellcolor{dgreen} \\
     7& & \cellcolor{dred}& \cellcolor{dred}& \cellcolor{dred}& \cellcolor{dred}&\cellcolor{dgreen}&\cellcolor{dgreen}&\cellcolor{dgreen}&\cellcolor{dgreen}&\cellcolor{dgreen}&\cellcolor{dgreen}&\cellcolor{dgreen}&\cellcolor{dgreen}&\cellcolor{dgreen}&\cellcolor{dgreen}&\cellcolor{dgreen} \\
     6& \cellcolor{dred}& \cellcolor{dred}& \cellcolor{dred}& \cellcolor{dred}& \cellcolor{dred}&\cellcolor{dgreen}&\cellcolor{dgreen}&\cellcolor{dgreen}&\cellcolor{dgreen}&\cellcolor{dgreen}&\cellcolor{dgreen}&\cellcolor{dgreen}&\cellcolor{dgreen}&\cellcolor{dgreen}&\cellcolor{dgreen}&  \\
     5& \cellcolor{dred}& \cellcolor{dred}& \cellcolor{dred}& \cellcolor{dred}&\cellcolor{dgreen}&\cellcolor{dgreen}&\cellcolor{dgreen}&\cellcolor{dgreen}&\cellcolor{dgreen}&\cellcolor{dgreen}&\cellcolor{dgreen}&\cellcolor{dgreen}&\cellcolor{dgreen}&\cellcolor{dgreen}& & \cellcolor{dred} \\
     4& \cellcolor{dred}& \cellcolor{dred}& \cellcolor{dred}&\cellcolor{dgreen}&\cellcolor{dgreen}&\cellcolor{dgreen}&\cellcolor{dgreen}&\cellcolor{dgreen}&\cellcolor{dgreen}&\cellcolor{dgreen}&\cellcolor{dgreen}&\cellcolor{dgreen}&\cellcolor{dgreen}& & \cellcolor{dred}& \cellcolor{dred} \\
     3& \cellcolor{dred}& \cellcolor{dred}&\cellcolor{dgreen}&\cellcolor{dgreen}&\cellcolor{dgreen}&\cellcolor{dgreen}&\cellcolor{dgreen}&\cellcolor{dgreen}& \cellcolor{dred}&\cellcolor{dgreen}&\cellcolor{dgreen}&\cellcolor{dgreen}& & \cellcolor{dred}& \cellcolor{dred}& \cellcolor{dred} \\
     2& \cellcolor{dred}&\cellcolor{dgreen}&\cellcolor{dgreen}&\cellcolor{dgreen}&\cellcolor{dgreen}&\cellcolor{dgreen}&\cellcolor{dgreen}& \cellcolor{dred}& \cellcolor{dred}&\cellcolor{dgreen}&\cellcolor{dgreen}& & \cellcolor{dred}& \cellcolor{dred}& \cellcolor{dred}& \cellcolor{dred} \\
     1&\cellcolor{dgreen}&\cellcolor{dgreen}&\cellcolor{dgreen}&\cellcolor{dgreen}&\cellcolor{dgreen}&\cellcolor{dgreen}&\cellcolor{dgreen}&\cellcolor{dgreen}&\cellcolor{dgreen}&\cellcolor{dgreen}& &\cellcolor{dgreen}&\cellcolor{dgreen}&\cellcolor{dgreen}&\cellcolor{dgreen}&\cellcolor{dgreen} \\
     0&\cellcolor{dgreen}&\cellcolor{dgreen}&\cellcolor{dgreen}&\cellcolor{dgreen}&\cellcolor{dgreen}&\cellcolor{dgreen}&\cellcolor{dgreen}&\cellcolor{dgreen}&\cellcolor{dgreen}& & \cellcolor{dred}& \cellcolor{dred}&\cellcolor{dgreen}&\cellcolor{dgreen}&\cellcolor{dgreen}&\cellcolor{dgreen} \\
    \hhline{~----------------}
    \multicolumn{1}{c}{ }&0&1&2&3&4&5&6&7&8&9&\hspace{-4pt}10\hspace{-4pt}&\hspace{-4pt}11\hspace{-4pt}&\hspace{-4pt}12\hspace{-4pt}&\hspace{-4pt}13\hspace{-4pt}&\hspace{-4pt}14\hspace{-4pt}&\multicolumn{1}{c}{\hspace{-4pt}15\hspace{-4pt}}
  \end{tabular}
}
\\
{\small $\begin{array}{r@{\hspace{1.0ex}}l}
       & (2 * x + 14 * y + 3 \le x + 15 * y + 7) \\
  \lor & (x + y + 4 \le x + 15 * y + 7)
\end{array}$}
&
{\small $\begin{array}{r@{\hspace{1.0ex}}l}
         &  \left(\begin{array}{@{\hspace{0ex}}r@{\hspace{1.0ex}}l@{\hspace{0ex}}}
                     & (2 * x + 14 * y + 3 \le x + 15 * y + 7) \\
              \sqcup & (x + y + 4 \le x + 15 * y + 7)
            \end{array}\right) \\
       = & 0 * x + 0 * y + 1 \le 7 * x + 9 * y + 1
\end{array}$}
\\
{\small (c)} & {\small (d)}
\end{tabular}
\end{spacing}
\caption{\label{Fi:BitVectorPolyhedraJoinExample}
{\small Each {\color{dgreen} green} and {\color{dred} red} cell represents a solution in $4$-bit unsigned modular arithmetic of the indicated formula.
In (d), the occurrences of {\color{dred} red} cells are points that do not satisfy (c), but are needed for a conjunctive formula to over-approximate (c).
}
}
\end{figure}

\mypar{Additional experiment}
Fig.\ \ref{Fi:BitVectorPolyhedraJoinExample} illustrates the application of join ($\sqcup$) to two formulas in $\mathcal{L}_\BVCONJ$.
Let $\psi_a, \psi_b \in \mathcal{L}_\BVCONJ$ be the formulas from Fig.\ \ref{Fi:BitVectorPolyhedraJoinExample}(a) and Fig.\ \ref{Fi:BitVectorPolyhedraJoinExample}(b), respectively.
Fig.\ \ref{Fi:BitVectorPolyhedraJoinExample}(c) shows the satisfying solutions of the disjunction $\psi_c =_{\textit{df}} \psi_a \lor \psi_b$.
Fig.\ \ref{Fi:BitVectorPolyhedraJoinExample}(d) shows the $\mathcal{L}_\BVCONJ$ formula $\psi_d$ computed by \name for $\psi_a \sqcup \psi_b$, as well as the satisfying solutions of $\psi_d$.
The join computation was performed by applying \name to $\psi_c$.
In this case, \name showed that the result could be expressed with a single inequality, namely, $1 \le 7x + 9y + 1$.

\subsection{Detailed Evaluation for each Applications}

Table~\ref{tab:benchmarks} shows how many times each query is performed for each benchmark in our first application in Section~\ref{se:specmining} and the total running time spent performing each type of query.
Table~\ref{tab:benchmarks-2} shows the same metrics for the other three applications.
``Last iter.''\ denotes the time taken to perform the final call to \cp from the final call on \synthproperty at line \ref{Li:CallSynthPropertyOne} of Alg.\ \ref{alg:SynthesizeAllProperties}.
That call on \synthproperty is the most difficult one because it implicitly establishes that the synthesized \lconjunction is indeed a best one.

In Table~\ref{tab:benchmarks}, for the benchmarks above the thick line, the semantics of the involved operations are expressible in SMT-Lib, whereas for the other benchmarks the semantics are expressed in \sketch.

\begin{table}[tp]
\caption{Evaluation results of \name for specification-mining. 
The Enum. column reports the estimated time required to run \cs for all formulas in the DSL $\lang$. This estimation is achieved by multiplying the size of the grammar by the average running time of the \cs.
A (*) indicates a timeout when attempting to prove precision in the last iteration.
In that case, we report as total time the time at which \name timed out. 
\label{tab:benchmarks}}
{\footnotesize
\setlength{\tabcolsep}{2pt}
\begin{tabular}{ccrrrrrrrrrr} 
\toprule[.1em]
\multicolumn{2}{c}{\multirow{2}{*}[-0.4ex]{Problem}}
& \multicolumn{1}{c}{\multirow{2}{*}{$|\lang|$}}
    & \multicolumn{2}{c}{\synth} & \multicolumn{2}{c}{\cs} & \multicolumn{2}{c}{\cp} & Last Iter. & Enum. & Total \\
    \cmidrule{4-12}
    & & & Num & T(sec) & Num & T(sec) & Num & T(sec) & T(sec) & T(sec) & T(sec) \\
\midrule[.1em]
\parbox[t]{2mm}{\multirow{8}{*}{\rotatebox[origin=c]{90}{\sygus}}}
    & \maxtwo & $1.57 \cdot 10^5$ & 6 & 0.13 & 12 & 0.06 & 9 & 1.37 & 0.05 & 787.32 & 1.55 \\
    & \maxthree & $8.85 \cdot 10^5$ & 19 & 13.83 & 48 & 3.53 & 36 & 131.44 & 0.46 & $6.51 \cdot 10^4$ & 148.79 \\
    & \maxfour & $5.06 \cdot 10^8$ & - & - & - & - & - & - & - & - & - \\
    & \diff & $5.66 \cdot 10^7$ & 18 & 19.41 & 41 & 1.06 & 28 & 236.61 & 0.48 & $1.46 \cdot 10^6$ & 257.08 \\
    & \diff, \diff & $3.73 \cdot 10^5$ & - & - & - & - & - & - & - & - & - \\
    & \arraytwo & $3.37 \cdot 10^6$ & 16 & 3.44 & 41 & 1.44 & 31 & 60.97 & 0.28 & $1.19 \cdot 10^5$ & 65.84 \\
    & \arraythree & $3.66 \cdot 10^9$ & - & - & - & - & - & - & - & - & - \\
\cmidrule{1-12}
\parbox[t]{2mm}{\multirow{3}{*}{\rotatebox[origin=c]{90}{LIA}}}
    & \iaabs\xspace (Eq.~\eqref{eq:arith-grammar-1}) & $3.37 \cdot 10^6$ & 6 & 1.14 & 14 & 0.76 & 11 & 10.16 & 0.38 & $1.83 \cdot 10^5$ & 12.05 \\
    & \iaabs\xspace (Eq.~\eqref{eq:arith-grammar-2}) & $2.90 \cdot 10^{11}$ & - & - & - & - & - & - & - & - & - \\
    & \iaabs\xspace (Eq.~\eqref{eq:arith-grammar-3}) & $\infty$ & 22 & 3.98 & 23 & 0.15 & 3 & 0.67 & 0.70 & $\infty$ & 4.80 \\
\midrule[.2em]
\parbox[t]{2mm}{\multirow{14}{*}{\rotatebox[origin=c]{90}{List}}}
    & \lappend & $4.97 \cdot 10^8$ & 29 & 46.62 & 68 & 81.50 & 43 & 91.46 & 57.59 & $5.95 \cdot 10^8$ & 219.58 \\
    & \ldelete & $2.99 \cdot 10^6$ & 22 & 33.91 & 64 & 72.05 & 50 & 100.63 & 9.10 & $3.36 \cdot 10^6$ & 206.60 \\
    & \ldeleteall & $2.99 \cdot 10^6$ & 24 & 39.49 & 66 & 87.15 & 49 & 107.07 & 18.41 & $3.94 \cdot 10^6$ & 233.71 \\
    & \ldrop & $4.75 \cdot 10^8$ & 26 & 38.55 & 60 & 66.00 & 37 & 71.48 & 45.89 & $5.22 \cdot 10^8$ & 176.03 \\
    & \lelem & 4,096 & 8 & 9.05 & 18 & 17.36 & 12 & 15.71 & 7.04 & $3950.36$ & 42.12 \\
    & \lelemidx & $8.74 \cdot 10^6$ & 22 & 28.62 & 53 & 52.87 & 35 & 53.96 & 17.22 & $8.72 \cdot 10^6$ & 135.45 \\
    & \lith & $1.91 \cdot 10^6$ & 21 & 28.72 & 46 & 45.41 & 29 & 47.21 & 23.40 & $1.88 \cdot 10^6$ & 121.34 \\
    & \lmin & $2.38 \cdot 10^5$ & 5 & 6.02 & 14 & 13.71 & 11 & 14.42 & 2.53 & $2.38 \cdot 10^5$ & 34.15 \\
    & \lreplicate & $1.82 \cdot 10^6$ & 5 & 6.04 & 20 & 19.25 & 17 & 22.59 & 6.59 & $1.75 \cdot 10^6$ & 47.87 \\
    & \lreverse & $1.73 \cdot 10^6$ & 6 & 7.28 & 16 & 18.11 & 12 & 16.84 & 7.83 & $1.96 \cdot 10^6$ & 42.23 \\
    & \lreverse, \lreverse & $6.40 \cdot 10^4$ & 20 & 26.64 & 52 & 96.65 & 40 & 76.37 & 32.21 & $1.19 \cdot 10^5$ & 199.66 \\
    & \lsnoc & $2.99 \cdot 10^6$ & 17 & 25.03 & 47 & 53.42 & 37 & 68.81 & 19.72 & $3.39 \cdot 10^6$ & 147.26 \\
    & \lstutter & $1.73 \cdot 10^6$ & 6 & 7.69 & 23 & 24.80 & 20 & 28.65 & 2.44 & $1.86 \cdot 10^6$ & 61.14 \\
    & \ltake & $8.49 \cdot 10^6$ & 12 & 15.1 & 35 & 40.08 & 26 & 41.35 & 7.91 & $9.72 \cdot 10^6$ & 96.53 \\
\cmidrule{1-12}
\parbox[t]{2mm}{\multirow{6}{*}{\rotatebox[origin=c]{90}{Binary Tree}}}
    & \tempty & $1.41 \cdot 10^5$ & 2 & 2.1 & 6 & 6.41 & 5 & 6.21 & 2.23 & $1.50 \cdot 10^5$ & 14.73 \\
    & \tbranch & $5.68 \cdot 10^8$ & 134 & 1,487.03 & 242 & 340.10 & 118 & 1,938.85 & $\ast$ & $8.00 \cdot 10^9$ & 3,765.98 \\
    & \telem & $2.62 \cdot 10^5$ & 12 & 15.15 & 27 & 30.79 & 17 & 26.22 & 13.69 & $2.99 \cdot 10^5$ & 72.17 \\
    & \tbranch, \tleft & $2.16 \cdot 10^5$ & 38 & 67.7 & 91 & 134.30 & 61 & 156.89 & 70.24 & $3.19 \cdot 10^5$ & 358.89 \\
    & \tbranch, \tright & $2.16 \cdot 10^5$ & 40 & 69.85 & 95 & 134.88 & 63 & 149.51 & 20.75 & $3.07 \cdot 10^5$ & 354.24 \\
    & \tbranch, \trootval & $2.62 \cdot 10^5$ & 36 & 67.19 & 72 & 101.16 & 41 & 93.45 & 84.36 & $3.68 \cdot 10^5$ & 261.81 \\
\cmidrule{1-12}
\parbox[t]{2mm}{\multirow{4}{*}{\rotatebox[origin=c]{90}{BST}}}
    & \bstempty & $1.41 \cdot 10^5$ & 2 & 2.16 & 7 & 6.13 & 6 & 6.92 & 1.92 & $1.23 \cdot 10^5$ & 15.21 \\
    & \bstinsert & $2.99 \cdot 10^6$ & 16 & 24.16 & 52 & 60.63 & 43 & 118.81 & 26.52 & $3.48 \cdot 10^6$ & 203.60 \\
    & \bstdelete & $2.99 \cdot 10^6$ & 17 & 24.75 & 49 & 68.29 & 38 & 109.10 & 31.50 & $4.16 \cdot 10^6$ & 202.15 \\
    & \bstfind & $2.87 \cdot 10^5$ & 11 & 13.7 & 22 & 21.57 & 13 & 19.52 & 22.10 & $2.82 \cdot 10^5$ & 54.79 \\
\cmidrule{1-12}
\parbox[t]{2mm}{\multirow{4}{*}{\rotatebox[origin=c]{90}{Stack}}}
    & \stempty & $1.25 \cdot 10^5$ & 2 & 2.12 & 6 & 5.61 & 5 & 5.89 & 2.05 & $1.17 \cdot 10^5$ & 13.63 \\
    & \stpush & $1.73 \cdot 10^6$ & 4 & 4.9 & 10 & 11.16 & 7 & 9.73 & 2.50 & $1.91 \cdot 10^6$ & 25.79 \\
    & \stpop & $1.73 \cdot 10^6$ & 4 & 4.84 & 11 & 12.16 & 8 & 10.97 & 12.02 & $1.93 \cdot 10^6$ & 27.98 \\
    & \stpush, \stpop & $2.62 \cdot 10^5$ & 32 & 51.12 & 77 & 95.72 & 53 & 107.13 & 27.66 & $3.26 \cdot 10^5$ & 253.97 \\
\cmidrule{1-12}
\parbox[t]{2mm}{\multirow{3}{*}{\rotatebox[origin=c]{90}{Queue}}}
    & \qempty & 64 & 2 & 2.18 & 6 & 5.86 & 5 & 9.81 & 2.07 & 62.51 & 17.85 \\
    & \qenqueue & $5.93 \cdot 10^5$ & 4 & 5.4 & 11 & 16.82 & 8 & 270.26 & 200.15 & $9.06 \cdot 10^5$ & 292.48 \\
    & \qdequeue & $5.93 \cdot 10^5$ & 4 & 5.51 & 9 & 12.51 & 6 & 290.28 & 192.02 & $8.24 \cdot 10^5$ & 308.30 \\
\cmidrule{1-12}
\parbox[t]{2mm}{\multirow{4}{*}{\rotatebox[origin=c]{90}{Arithmetic}}}    
    & \ialinsum\xspace (Eq.~\eqref{eq:arith-grammar-1}) & $1.01 \cdot 10^7$ & 7 & 6.69 & 20 & 14.50 & 15 & 15.33 & 5.33 & $7.31 \cdot 10^6$ & 36.52 \\
    & \ialinsum\xspace (Eq.~\eqref{eq:arith-grammar-2}) & $2.90 \cdot 10^{10}$ & 15 & 15.31 & 99 & 71.72 & 88 & 112.87 & 2.70 & $2.10 \cdot 10^{10}$ & 199.90 \\
    & \ianonlinsum\xspace (Eq.~\eqref{eq:arith-grammar-1}) & $1.01 \cdot 10^7$ & 8 & 7.95 & 30 & 21.58 & 25 & 26.48 & 2.53 & $7.25 \cdot 10^6$ & 56.01 \\
    & \ianonlinsum\xspace (Eq.~\eqref{eq:arith-grammar-2}) & $1.48 \cdot 10^{13}$ & 17 & 16.82 & 121 & 102.71 & 107 & 734.23 & 48.21 & $1.26 \cdot 10^{13}$ & 853.77 \\
\bottomrule[.1em]
\bottomrule[.1em]
\end{tabular}
}
\end{table}

\begin{table}[tp]
\caption{Evaluation results of \name for sensitivity analysis and abstract domains.
The Enum. column reports the estimated time required to run \cs for all formulas in the DSL $\lang$. This estimation is achieved by multiplying the size of the grammar by the average running time of the \cs.
A (*) indicates a timeout when attempting to prove precision in the last iteration.
In that case, we report as total time the time at which \name timed out. 
\label{tab:benchmarks-2}}
{\footnotesize
\setlength{\tabcolsep}{3pt}
\begin{tabular}{ccrrrrrrrrrr} 
\toprule[.1em]
\multicolumn{2}{c}{\multirow{2}{*}[-0.4ex]{Problem}}
& \multicolumn{1}{c}{\multirow{2}{*}{$|\lang|$}}
    & \multicolumn{2}{c}{\synth} & \multicolumn{2}{c}{\cs} & \multicolumn{2}{c}{\cp} & Last Iter. & Enum. & Total \\
    \cmidrule{4-12}
    & & & Num & T(sec) & Num & T(sec) & Num & T(sec) & T(sec) & T(sec) & T(sec) \\
\midrule[.1em]
\parbox[t]{2mm}{\multirow{3}{*}{\rotatebox[origin=c]{90}{ArrayList}}}
    & \arrayget, \arrayadd & 1,156 & 13 & 4.19 & 23 & 6.87 & 13 & 5.00 & 5.39 & 345.29 &  16.06 \\
    & \size, \arrayempty & 3 & 3 & 0.55 & 5 & 1.34 & 4 & 1.23 & 0.56 & 0.80 & 3.12 \\
    & \size, \arrayadd & 1,303 & 5 & 1.26 & 9 & 2.54 & 6 & 2.07 & 1.86 & 367.74 & 5.88 \\ 
\midrule[.1em]
\parbox[t]{2mm}{\multirow{3}{*}{\rotatebox[origin=c]{90}{HashMap}}}
    & \mapget, \mapempty & 2 & 2 & 0.67 & 6 & 2.53 & 5 & 2.40 & 0.80 & 0.84 & 5.60 \\
    & \mapget, \mapput & 1,860 & 14 &6.22 & 27 & 295.04 & 16 & 189.98 & 119.87 & $2.03 \cdot 10^4$ & 491.24 \\
    & \mapput, \mapput & 103 & 4 & 1.58 & 9 & 217.01 & 6 & 485.72 & 213.21 & 2483.56 & 704.32 \\ 
\midrule[.1em]
\parbox[t]{2mm}{\multirow{6}{*}{\rotatebox[origin=c]{90}{HashSet}}}
    & \setsize, \setempty & 3 & 2 & 0.62 & 6 & 1.80 & 5 & 1.62 & 0.62 & 0.90 & 4.04 \\ 
    & \setsize, \setadd & 1,315 & 7 & 2.39 & 14 & 4.56 & 9 & 3.40 & 0.70 &  428.31 & 10.36 \\ 
    & \setcontains, \setempty & 3 & 2 & 0.60 & 6 & 1.80 & 5 & 1.64 & 0.64 & 0.90 & 4.04 \\ 
    & \setcontains, \setadd & 55 & 5 & 1.64 & 12 & 3.94 & 9 & 3.28 & 0.69 & 18.06 & 8.86 \\ 
    & \setremove, \setempty & 5 & 1 & 0.34 & 5 & 1.05 & 5 & 1.78 & 0.55 & 1.05 & 3.17 \\ 
    & \setremove, \setadd & 181 & 7 & 3.38 & 15 & 6.20 & 10 & 13.23 & 3.86 & 74.81 & 22.81 \\ 
\midrule[.2em]
\parbox[t]{2mm}{\multirow{9}{*}{\rotatebox[origin=c]{90}{Sensitivity Hamming Dist.}}}
    & \lappend & 1,025 & 25 & 11.24 & 55 & 25.88 & 34 & 19.46 & 0.98 & 482.31 & 56.58 \\
    & \lcons & 193 & 14 & 5.21 & 34 & 12.08 & 24 & 10.55 & 0.80 & 68.57 & 27.84 \\
    & \lconsdelete & 193 & 10 & 3.68 & 24 & 9.02 & 16 & 7.29 & 2.44 & 72.54 & 20.0 \\ 
    & \ldelete & 193 & 13 & 4.97 & 26 & 10.03 & 14 & 6.48 & 5.81 & 74.45 & 21.48 \\
    & \ldeleteall & 193 & 15 & 5.99 & 28 & 11.69 & 14 & 6.56 & 2.45 & 161.02 & 24.24 \\
    & \lreverse & 65 & 7 & 2.47 & 15 & 9.82 & 10 & 4.10 & 0.79 & 42.55 & 16.39 \\
    & \lsnoc & 193 & 17 & 6.62 & 37 & 14.55 & 24 & 11.22 & 0.81 & 75.90 & 32.39 \\
    & \lstutter & 65 & 7 & 2.47 & 17 & 6.38 & 12 & 4.95 & 2.25 & 24.39 & 13.79 \\
    & \ltail & 65 & 7 & 2.39 & 19 & 6.61 & 14 & 5.52 & 2.2 & 22.61 & 14.52 \\
\cmidrule{1-12}
\parbox[t]{2mm}{\multirow{9}{*}{\rotatebox[origin=c]{90}{Sensitivity Edit Dist.}}}
    & \lappend & 16,385 & 43 & 29.18 & 82 & 118.52 & 43 & 55.27 & 7.95 & $2.37 \cdot 10^4$ & 202.97 \\
    & \lcons & 769 & 6 & 2.49 & 10 & 4.18 & 4 & 2.02 & 8.69 & 321.44 & 8.69 \\
    & \lconsdelete & 769 & 19 & 9.00 & 36 & 15.80 & 19 & 12.70 & 11.75 & 337.51 & 37.50 \\ 
    & \ldelete & 769 & 33 & 17.9 & 68 & 115.43 & 42 & 31.53 & 1.31 & 1385.38 & 164.86 \\
    & \ldeleteall & 769 & 26 & 13.42 & 55 & 190.48 & 33 & 23.55 & 3.64 & 2663.26 & 227.45 \\
    & \lreverse & 257 & 13 & 6.00 & 27 & 61.43 & 16 & 10.11 & 3.32 & 594.24 & 77.53 \\
    & \lsnoc & 769 & 29 & 15.44 & 61 & 37.21 & 36 & 25.72 & 5.80 & 700.07 & 78.38 \\
    & \lstutter & 257 & 14 & 6.29 & 26 & 13.66 & 14 & 8.51 & 3.03 & 135.02 & 28.47 \\
    & \ltail & 257 & 14 & 6.17 & 29 & 12.83 & 17 & 10.01 & 9.66 & 113.70 & 29.01 \\
\midrule[.2em]
\parbox[t]{2mm}{\multirow{9}{*}{\rotatebox[origin=c]{90}{Bit-Vector Polyhedra}}}
    & \bvsquare & $1.68 \cdot 10^7$ & 20 & 14.79 & 76 & 24.65 & 60 & 136.76 & 68.99 & $5.44 \cdot 10^6$  & 176.19 \\
    & \bvcube & $1.68 \cdot 10^7$ & 21 & 14.26 & 71 & 22.85 & 55 & 195.03 & 115.98 & $6.97 \cdot 10^6$ & 232.14 \\
    & \bvhalf & $1.68 \cdot 10^7$ & 19 & 16.17 & 56 & 19.27 & 40 & 353.14 & $\ast$ & $5.77 \cdot 10^6$ & 388.57 \\
    & \bvsquareineq & $1.68 \cdot 10^7$ & 45 & 105.24 & 70 & 20.24 & 28 & 95.41 & 98.74 & $4.85 \cdot 10^6$ & 220.89 \\
    & \bvconstneq & $1.68 \cdot 10^7$ & 46 & 65.37 & 63 & 20.23 & 18 & 69.83 & 75.32 & $5.39 \cdot 10^6$ & 155.43 \\
    & \bvconj & $1.68 \cdot 10^7$ & 46 & 73.08 & 67 & 20.49 & 23 & 102.32 & 117.45 & $5.13 \cdot 10^6$ & 195.89 \\
    & \bvsinglepoint & $1.68 \cdot 10^7$ & 9 & 14.61 & 75 & 21.64 & 70 & 202.59 & 5.95 & $4.84 \cdot 10^6$ & 238.84 \\
    & \bvfourpoints & $1.68 \cdot 10^7$ & 4 & 1.32 & 56 & 14.75 & 54 & 72.47 & 0.94 & $4.42 \cdot 10^6$ & 88.53 \\
    & \bvdisj & $1.68 \cdot 10^7$ & 48 & 178.05 & 62 & 22.69 & 15 & 102.73 & 105.71 & $6.14 \cdot 10^6$ & 303.47 \\
\bottomrule[.1em]
\bottomrule[.1em]
\end{tabular}
}
\end{table}

\end{document}